\documentclass[aps,pra,reprint,a4paper,10pt,nofootinbib,floatfix]{revtex4-2}
\usepackage[T1]{fontenc}
\usepackage[utf8]{inputenc}
\usepackage{lmodern}
\usepackage[a4paper,left=18mm,right=18mm,top=20mm,bottom=20mm]{geometry}
\usepackage{amsmath,amssymb,amsthm,mathtools,bm}
\usepackage{microtype,booktabs,array,enumitem,graphicx}
\usepackage[dvipsnames]{xcolor}
\usepackage[colorlinks=true,linkcolor=MidnightBlue,citecolor=MidnightBlue,urlcolor=MidnightBlue]{hyperref}
\usepackage[nameinlink,noabbrev]{cleveref}

\newtheorem{theorem}{Theorem}

\newtheorem{lemma}{Lemma}
\newtheorem{corollary}[theorem]{Corollary}
\theoremstyle{remark}

\DeclareMathOperator{\Tr}{Tr}
\DeclareMathOperator{\supp}{supp}
\DeclareMathOperator{\SEP}{SEP}
\DeclareMathOperator{\Bin}{Bin}
\DeclareMathOperator{\spanop}{span}
\newcommand{\id}{\mathrm{id}}
\newcommand{\ER}{E_R}
\newcommand{\ERi}{E_R^{\infty}}
\newcommand{\cL}{\mathcal L}
\newcommand{\cA}{\mathcal A}

\newcommand{\cP}{\mathcal P}
\newcommand{\cH}{\mathcal H}
\newcommand{\EE}{\mathbb E}

\newcommand{\ket}[1]{\lvert #1\rangle}
\newcommand{\bra}[1]{\langle #1\rvert}

\newcommand{\proj}[1]{\ket{#1}\!\bra{#1}}
\newcommand{\norm}[1]{\left\lVert#1\right\rVert}
\newcommand{\dd}{\mathrm d}

\usepackage{tikz,needspace}
\usetikzlibrary{arrows.meta,positioning,calc,decorations.pathreplacing}
\DeclareMathOperator{\Prb}{Pr}

\allowdisplaybreaks[2]
\begin{document}

\title{Energy-constrained two-way capacities of pure-loss \\ and quantum-limited amplifier channels}
\author{Stefano Pirandola}
\affiliation{Department of Computer Science, University of York, York YO10 5GH, United Kingdom}

\begin{abstract}
We determine the two-way quantum, entanglement-distribution, private, and secret-key capacities of the pure-loss bosonic channel under an unconditional mean transmitted-photon-number constraint. For transmissivity $\eta$ and mean photon number $N$, all four capacities coincide and are given by $g(N)-g((1-\eta)N)$,
where $g$ is the entropy of a thermal mode. This resolves the longstanding energy-constrained capacity problem by establishing the optimality of the reverse-coherent-information rate introduced in 2009. Our central tool is sector teleportation simulation, which we develop specifically to transfer the input-energy constraint to the shared entanglement resource, thereby preserving the constraint throughout converse arguments for general adaptive protocols. For every $N>0$, we further show that the corresponding strong-converse thresholds coincide with the unconstrained value $-\log_2(1-\eta)$. For parallel protocols with a hard total-photon-number cutoff, the same energy-constrained capacity formula instead satisfies an exponential strong converse. As an extension, we determine the energy-constrained two-way capacities of the quantum-limited amplifier, showing that classical assistance does not improve its unassisted quantum and private rates. Together, these results establish a unified finite-energy framework for adaptive bosonic communication and reveal how energy constraints fundamentally reshape both attainable rates and converse bounds.
\end{abstract}
\maketitle

\section{Introduction}\label{sec:introduction}
The pure-loss bosonic channel is the canonical model of attenuation in optical communication~\cite{HolevoWerner,Weedbrook}. Its two-way assisted capacities characterize the ultimate rates of quantum communication, entanglement distribution, private communication, and secret-key generation when unlimited two-way classical communication is available~\cite{Scarani,PirandolaAdvances}. Under a finite transmitted-photon budget, these rates depend jointly on channel loss and input energy.

Exact benchmarks are known for the classical and unassisted quantum capacities of bosonic attenuation~\cite{Giovannetti2004,Wolf2007}. In the two-way-assisted setting, the reverse coherent information~\cite{RCI}, extended to bosonic channels in Ref.~\cite{PGPBL}, provided an achievable rate in 2009; see also Ref.~\cite[Supp. Note 7, Eq.~(257)]{PLOB} for its finite-energy expression. The PLOB bound subsequently established the corresponding unconstrained two-way capacities in 2017~\cite{PLOB}. At finite energy, however, the known energy-dependent upper bounds~\cite{DSW,Goodenough} did not meet the 2009 achievable rate, leaving its optimality open.

Here we resolve this question under an unconditional mean transmitted-photon-number constraint. We prove that the reverse-coherent-information rate is optimal and obtain the exact two-way quantum, entanglement-distribution, private, and secret-key capacities of the pure-loss bosonic channel, including for general adaptive protocols with quantum memories. The key ingredient is \emph{sector teleportation simulation}, which transfers the input-energy constraint to the shared entanglement resource and thereby preserves the photon budget throughout the adaptive converse. Our more advanced construction combines ideas from channel simulation~\cite{PLOB,Simulation}, irreducible covariance~\cite{GBW}, and photon-number conditioning~\cite{Winnel}.

We further show that energy-constrained weak-converse capacities and strong-converse thresholds are distinct. Using vacuum randomization together with the unconstrained strong converse~\cite{WTB,Simulation}, we find that, for every nonzero mean-energy constraint, the strong-converse threshold remains equal to the unconstrained PLOB value. For parallel protocols with a hard total-photon-number cutoff, however, we prove an exponential strong converse at the energy-constrained capacity. This distinguishes an unconditional mean constraint from a constraint on photon-number occupation, a distinction also relevant to classical optical communication~\cite{WW}.

The pure-loss resource calculation also supports an extension to quantum-limited amplification. We show that the known finite-energy unassisted quantum and private rates~\cite{QiWilde,WildeQi} remain optimal with adaptive two-way assistance. The extension identifies each photon-number-conditioned amplifier resource with a conditional loss resource after exchanging the parties, and then uses degradability~\cite{CarusoGiovannetti} to retain the input-energy constraint. The exponential strong converse under a parallel photon-number cutoff extends to this channel as well.

Below we present our main results starting from the pure loss channel and then extending to the quantum-limited amplification. Detailed proofs are given in Appendices~\ref{sec:preliminaries}--\ref{sec:cutoff-strong}.

\section{Main results}\label{sec:results}
Let $\cL_\eta$ be the pure-loss channel with transmissivity
$0<\eta<1$, and write $q=1-\eta$. Alice and Bob start with local
preparations and public randomness, without shared entanglement or
secret key. They may interleave channel uses with adaptive local
operations and unlimited authenticated two-way public classical
communication (LOCC). Each finite protocol has finitely many local
rounds and finite or countable outcome alphabets. Local memories have
no energy or entropy constraint. Where relevant, Eve is assumed to receive the loss environments
and the entire public transcript.

For an $n$-use protocol, let $\rho^{(i)}_{A_i}$ be the unconditional
state of the transmitted mode at use $i$, and $\hat n_{A_i}$ its
photon-number operator. The mean-energy constraint is given by
\begin{equation}
 x_i:=\Tr(\hat n_{A_i}\rho^{(i)}_{A_i}),\qquad
 \frac1n\sum_{i=1}^{n}x_i\le N,
 \label{eq:energy}
\end{equation}
where $0\le N<\infty$. Every random choice, public outcome, and
unsuccessful branch is included, and all $n$ channel slots count
toward the rate. No peak or branchwise energy constraint is imposed.
For $x\ge0$, define
\begin{equation}
 \begin{aligned}
 g(x)&=(x+1)\log_2(x+1)-x\log_2x,\\
 c_\eta(x)&=g(x)-g(qx),\qquad g(0)=0.
 \end{aligned}
 \label{eq:cdef}
\end{equation}

The capacities $Q_2$, $D_2$, $P_2$, and $K_2$ describe, respectively,
quantum transmission, entanglement distribution, private classical
transmission, and secret-key agreement. Here $D_2$ concerns entanglement
generated from an initially unentangled preparation. Each capacity is
the supremum of asymptotic rates $\log_2d_n/n$ under \eqref{eq:energy},
where $d_n$ is the target dimension or classical alphabet size and
the error tends to zero. Quantum tasks use entanglement infidelity;
classical tasks use half the trace distance from a correct uniform
key or message independent of Eve, with her actual marginal retained
\cite{RennerKonig,PortmannRenner} (precise error conventions are given
in Appendix~\ref{sec:task-errors}). Rates are in qubits, ebits, or secret
bits per channel use, as appropriate. Two operational equivalences preserve the optical energy budget:
\begin{equation}
 \begin{aligned}
 Q_2(\cL_\eta,N)&=D_2(\cL_\eta,N),\\
 P_2(\cL_\eta,N)&=K_2(\cL_\eta,N).
 \end{aligned}
 \label{eq:operational-equivalences}
\end{equation}
In fact, quantum transmission distributes entanglement, which conversely
transmits quantum states by teleportation~\cite{BennettTeleport}.
Private transmission of a random message generates key; a key
transmits a private message by a one-time pad and public ciphertext
\cite{ShannonSecrecy,PLOB}. Note that these conversions add only local operations
and public communication.

\subsection{Pure-loss capacities and sector teleportation}
Our central result is the following:
\begin{theorem}[Energy-constrained capacities]\label{thm:capacity}
Under \eqref{eq:energy}, we have
\begin{equation}
 \boxed{\begin{aligned}
 Q_2(\cL_\eta,N)&=D_2(\cL_\eta,N)\\
 &=P_2(\cL_\eta,N)=K_2(\cL_\eta,N)\\
 &=g(N)-g((1-\eta)N).
 \end{aligned}}
 \label{eq:capacity}
\end{equation}
\end{theorem}
\noindent The proof is in Appendices~\ref{sec:weak}--\ref{sec:key}. This result closes the long-standing gap
between reverse coherent information and the previous upper bounds
\cite{DSW,PLOB}. The capacities vanish at $N=0$ and approach the
PLOB value $-\log_2(1-\eta)$ in the limit of infinite energy $N\to\infty$.

An associated question is the analysis of the strong-converse thresholds under the same
mean-energy constraint. For a capacity $C$, its strong-converse
threshold $C^\dagger$ is the infimum of rates $R$ such that every
admissible sequence of $n$-use protocols satisfying
$\liminf_{n\to\infty}(\log_2 d_n)/n>R$
has error tending to one. Here $d_n$ denotes the target dimension
or the message/key alphabet size, depending on the task. For
quantum transmission and entanglement distribution, this means
vanishing entanglement fidelity. The operational conversions
in \eqref{eq:operational-equivalences} also yield the corresponding
equivalences between strong-converse thresholds, with the error
comparisons detailed in Appendix~\ref{sec:strong}. We establish
the following result.

\begin{theorem}[Mean-energy strong-converse thresholds]\label{thm:mean-sc}
For every $N>0$, %
under \eqref{eq:energy}, we have
\begin{equation}
 \boxed{\begin{aligned}
 Q_2^{\dagger}(\cL_\eta,N)&=D_2^{\dagger}(\cL_\eta,N)\\
 &=P_2^{\dagger}(\cL_\eta,N)=K_2^{\dagger}(\cL_\eta,N)\\
 &=-\log_2(1-\eta).
 \end{aligned}}
 \label{eq:mean-sc}
\end{equation}
\end{theorem}

The energy-constrained capacity therefore differs from its strong-converse
threshold under an unconditional mean constraint. Appendix~\ref{sec:strong}
proves the lower threshold by vacuum dilution and uses the established
unconstrained adaptive strong converse for the upper threshold~\cite{WTB,Simulation}.

The capacity identities in Theorem~\ref{thm:capacity} and the
strong-converse thresholds in Theorem~\ref{thm:mean-sc} remain
unchanged under the stronger per-use mean-energy constraint
$x_i\le N$ for every channel use, where $x_i$ is still averaged
over all classical outcomes and random choices. This condition
implies the block-average constraint \eqref{eq:energy}, so both
converse bounds continue to apply. The matching sector-hashing
and vacuum-dilution constructions also satisfy the per-use
constraint, as shown in Appendices~\ref{subapp:lowerbound}
and~\ref{sec:strong}, respectively. Thus both energy conventions
give the same capacities and strong-converse thresholds for
all four communication tasks.

An energy-dependent strong converse is recovered in a more restricted
setting. Suppose that all $n$ optical inputs are prepared before
transmission, with arbitrary two-way LOCC before and after the parallel
channel uses but no feedback into later inputs. Impose a hard cutoff:
the joint input is supported on states with at most $\lfloor nN\rfloor$
photons in total, allowing coherence between different photon-number
sectors. Alice and Bob still start without shared entanglement or
secret key. This support constraint excludes the vacuum-dilution
construction responsible for Theorem~\ref{thm:mean-sc}.

\begin{theorem}[Parallel photon-cutoff strong converse]
\label{thm:cutoff-sc}
For $0<N<\infty$, in the parallel cutoff model just described,
all four capacities and their strong-converse thresholds equal
$c_\eta(N)=g(N)-g((1-\eta)N)$. At any fixed rate above this value,
the quantum target fidelity vanishes exponentially, and the
operational key or private-message error approaches one exponentially.
\end{theorem}
The same threshold holds when the probability of exceeding the cutoff
vanishes with $n$; exponential convergence is guaranteed when this
tail also decays exponentially. Appendix~\ref{sec:cutoff-strong}
gives the finite-error bounds and proof. This result concerns parallel
inputs; it does not assert an energy-dependent strong converse for
interleaved adaptive protocols.

As a consequence of the sector-resource calculation, we also determine the regularized relative
entropy of entanglement of a finite-energy channel output.
For a bipartite state $\rho$, the relative entropy of entanglement
(REE) is
$\ER(\rho)=\inf_{\sigma\in\SEP}D(\rho\Vert\sigma)$,
where $\SEP$ denotes the trace-norm closed set of separable
states and
$D(\rho\Vert\sigma)=\Tr[\rho(\log_2\rho-\log_2\sigma)]$
is the quantum relative entropy. The regularized REE is
$\ERi(\rho)=\lim_{r\to\infty}\ER(\rho^{\otimes r})/r$. We obtain the following result, proved in
Appendix~\ref{sec:direct}:

\begin{corollary}[Regularized entanglement of a finite-energy output]
\label{cor:tmsv}
Let $J_{\eta,N}$ result from sending one half of a two-mode squeezed
vacuum of local mean photon number $N$ through $\cL_\eta$. Then
\begin{equation}
 \ERi(J_{\eta,N})=c_\eta(N).
 \label{eq:tmsv-ree}
\end{equation}
\end{corollary}

The crucial mechanism behind the weak converse results is \emph{sector
teleportation simulation}, which we develop in this paper.
A multimode state with a fixed total photon number occupies a
finite-dimensional subspace and can be teleported perfectly using
a maximally entangled pair whose two shares each contain that same
photon number. Combining this observation with established covariant
channel-simulation principles~\cite{PLOB,GBW} gives simulation resources
that explicitly retain the input photon budget. Their entanglement
can therefore be bounded in terms of the available input energy.

To express the construction, let $\cH_{s,m}$ be the subspace of $s$
modes containing exactly $m$ photons, and let $\Phi_{s,m}$ be a
maximally entangled state between two copies of this subspace.
For pure loss, the resource is
\begin{equation}
 J_{s,m}
 =(\id_R\otimes\cL_\eta^{\otimes s})(\Phi_{s,m}).
 \label{eq:overview-resource}
\end{equation}
It depends on the channel and the photon-number sector, but not on
the particular input state. Passive-mode covariance provides an LOCC
map $\mathcal T_{s,m}$ such that
\begin{equation}
 \begin{aligned}
 \omega_{XB^sY}
 &=(\id_{XY}\otimes\cL_\eta^{\otimes s})(\rho_{XA^sY})\\
 &=\mathcal T_{s,m}(\rho_{XA^sY}\otimes J_{s,m})
 \end{aligned}
 \label{eq:overview-simulation}
\end{equation}
whenever the transmitted block $A^s$ is supported on $\cH_{s,m}$.
The identity includes arbitrary correlations with Alice's and Bob's
retained memories $X$ and $Y$.

The proof extends this sectorwise construction to general finite-energy
inputs by estimating the resource entanglement and controlling
coherences between sectors. For adaptive protocols, these estimates
are applied to the inputs at the same round across independent
replicas of the entire protocol, preserving all feedback within
each run. For an $n$-use protocol with unconditional input photon numbers $x_i$
satisfying $\sum_i x_i\le nN$, the resulting bound is
\begin{equation}
 \ERi(\zeta_n)
 \le\sum_{i=1}^{n}c_\eta(x_i)
 \le n c_\eta(N),
 \label{eq:overview-adaptive}
\end{equation}
where $\zeta_n$ is the retained final quantum state.
Target-state and privacy tests then yield the quantum and secret-key
weak converses. The simulation and entanglement estimates are proven in
Appendix~\ref{sec:weak}, with the operational reductions in
Appendices~\ref{sec:quantum} and~\ref{sec:key}.

\subsection{Extension to quantum-limited amplification}
\label{sec:amplifier-main}
The pure-loss resource calculation also provides a route to an exact
finite-energy result for the quantum-limited amplifier. This extension
uses the same sector teleportation construction and operational
converse tools. Its additional ingredients are a correspondence
between the conditional resources of the two channels and a
degradability-based estimate of the amplifier resource cost.

Let $\cA_\kappa$ have gain $\kappa>1$ and vacuum environment. In
the Heisenberg picture its output annihilation operator is
\begin{equation}
 \hat b=\sqrt\kappa\,\hat a+\sqrt{\kappa-1}\,\hat e^\dagger.
 \label{eq:amp-channel}
\end{equation}
The energy constraint remains \eqref{eq:energy}: only photons in
the transmitted input are counted. An input of mean $x$ produces
an output of mean $\kappa x+\kappa-1$; the photons supplied by
the amplifier do not enter the input budget. Using the thermal
entropy $g$ defined in \eqref{eq:cdef}, set
\begin{equation}
 r_\kappa(x)=g(\kappa x+\kappa-1)
             -g((\kappa-1)(x+1)).
 \label{eq:amp-rate}
\end{equation}
The operational definitions and errors are exactly those used for
pure loss, with Eve receiving the amplifier environments and the
complete public transcript.

\begin{theorem}[Quantum-limited amplifier extension]
\label{thm:amplifier}
Under \eqref{eq:energy}, for every $0\le N<\infty$,
\begin{equation}
 \boxed{\begin{aligned}
 Q_2(\cA_\kappa,N)&=D_2(\cA_\kappa,N)\\
 &=P_2(\cA_\kappa,N)=K_2(\cA_\kappa,N)\\
 &=r_\kappa(N).
 \end{aligned}}
 \label{eq:amp-capacity}
\end{equation}
For every $N>0$, the strong-converse thresholds are
\begin{equation}
 \boxed{\begin{aligned}
 Q_2^\dagger(\cA_\kappa,N)&=D_2^\dagger(\cA_\kappa,N)\\
 &=P_2^\dagger(\cA_\kappa,N)=K_2^\dagger(\cA_\kappa,N)\\
 &=\log_2\frac{\kappa}{\kappa-1}.
 \end{aligned}}
 \label{eq:amp-threshold}
\end{equation}
Both identities remain valid under the per-use unconditional
constraint $x_i\le N$.
For $N>0$, in the parallel photon-cutoff model of
Theorem~\ref{thm:cutoff-sc}, all four capacities and their
strong-converse thresholds equal $r_\kappa(N)$, with exponential
convergence above this rate. The same threshold holds for a
vanishing photon-number tail, with exponential convergence when
the tail decays exponentially.
\end{theorem}

The capacity in \eqref{eq:amp-capacity} is already achievable by
unassisted quantum and private codes~\cite{QiWilde,WildeQi}.
The extension establishes that adaptive two-way classical
assistance cannot increase it. Appendix~\ref{sec:amplifier}
proves the matching weak converse by reusing the pure-loss
analysis. The shared dilution argument in
Appendix~\ref{sec:strong} proves \eqref{eq:amp-threshold}, together
with the unconstrained converse~\cite{WTB,Simulation}.
Appendix~\ref{sec:amp-cutoff} extends the photon-cutoff proof
of Theorem~\ref{thm:cutoff-sc} to amplification, retaining the
same parallel-input restriction.

To describe the link with pure loss, use the same sector pair
$\Phi_{s,m}$ as in \eqref{eq:overview-resource} and define
\begin{equation}
 J^\kappa_{s,m}
   =(\id_R\otimes\cA_\kappa^{\otimes s})(\Phi_{s,m}).
 \label{eq:amp-resource}
\end{equation}
The finite Bell measurement still simulates the channel exactly
on the input sector, including correlations with memories.
Alice's resource share contains $m$ photons, whereas Bob's
has mean $\kappa m+s(\kappa-1)$ and unbounded photon-number
support. If $\ell$ photons are added, the conditional resource
$J^\kappa_{s,m,\ell}$ equals the pure-loss resource obtained from
sector $m+\ell$ and conditioned on receiving $m$ photons, with
the parties exchanged. Figure~\ref{fig:amp-duality} illustrates
this correspondence. It transfers the conditional resource REE
from pure loss to amplification; degradability then bounds its
average in terms of the input energy.

\begin{figure}[t]
\centering
\begin{tikzpicture}[
 x=1cm,y=1cm,
 arr/.style={-{Latex[length=1.5mm]},line width=.55pt},
 box/.style={draw,rounded corners=1pt,align=center,
             inner sep=4pt,font=\small},
 lab/.style={font=\scriptsize,align=center}]

 \node[box] (ain) at (0,2.2) {$\Phi_{s,m}$};
 \node[box] (amp) at (2.7,2.2) {$\cA_\kappa^{\otimes s}$};
 \node[box] (aout) at (5.8,2.2) {$J^\kappa_{s,m,\ell}$};

 \draw[arr] (ain)--(amp);
 \draw[arr] (amp)--(aout)
   node[midway,above=4pt,lab] {$m+\ell$ received};

 \node[box] (lin) at (0,0) {$\Phi_{s,m+\ell}$};
 \node[box] (loss) at (2.7,0) {$\cL_\eta^{\otimes s}$};
 \node[box] (lout) at (5.8,0) {$J_{s,m+\ell,m}$};

 \draw[arr] (lin)--(loss);
 \draw[arr] (loss)--(lout)
   node[midway,below=4pt,lab] {$m$ received};

 \draw[<->,line width=.55pt] (aout.south)--(lout.north)
   node[midway,left=5pt,lab] {exchange\\parties};

\end{tikzpicture}

\caption{Extension from loss to amplification. One share of each
sector pair is transmitted and the other is retained. Conditioning
on the indicated received photon number gives the resources at
right, which coincide after exchanging the parties. The normalized
loss resource is independent of $0<\eta<1$. This correspondence
concerns conditional states, rather than an equivalence of the
channels or their adaptive protocols.}
\label{fig:amp-duality}
\end{figure}
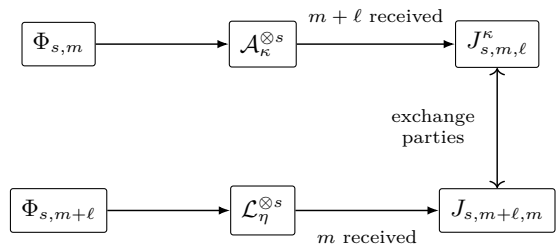


The extension also determines the regularized REE of the
finite-energy Gaussian output. Let $J^{\rm amp}_{\kappa,N}$
result from sending half of a two-mode squeezed vacuum of local
mean $N$ through $\cA_\kappa$. Then
\begin{equation}
 \ERi(J^{\rm amp}_{\kappa,N})=r_\kappa(N).
 \label{eq:amp-output-ree}
\end{equation}
Appendix~\ref{sec:amp-output} obtains this identity directly from
Corollary~\ref{cor:tmsv} by exchanging the parties and matching
the Gaussian covariance matrices.

\section{Conclusions}\label{sec:conclusion}
In this work, we have established the energy-constrained two-way capacities
of the pure-loss bosonic channel for quantum transmission, entanglement
distribution, private communication, and secret-key agreement under an
unconditional mean transmitted-photon-number constraint. We obtain
these capacities by proving that the reverse-coherent-information rate
introduced in 2009~\cite{PGPBL} is optimal even for general adaptive
protocols with quantum memories. The central tool is sector
teleportation simulation, developed to encode the input photon budget
in finite-energy shared resources and retain the energy constraint
throughout the weak-converse proof.

The finite-energy picture exhibits a sharp separation between
weak- and strong-converse behavior. Under the unconditional mean
constraint, the capacities depend on the available energy, whereas
the strong-converse thresholds remain equal to the unconstrained
PLOB value for every nonzero energy budget. For parallel protocols
with a hard total-photon-number cutoff, the energy-constrained capacities
instead satisfy an exponential strong converse. A vanishing
photon-number tail also suffices for a strong converse, with the
convergence rate controlled by that tail.

The amplifier extension shows how this pure-loss analysis can be reused. Its conditional resources are party-exchanged loss resources, so their entanglement cost follows from the same sector calculation. A separate degradability estimate then gives a matching adaptive converse for quantum-limited amplification. Thus the established unassisted quantum and private rates remain optimal with two-way assistance, while the mean-energy strong-converse thresholds again equal the corresponding unconstrained value. Under the parallel photon-number cutoff, its energy-constrained capacity also satisfies an exponential strong converse.

More broadly, transferring operational input constraints into
simulation resources provides a route to converse bounds that retain
physical resource limitations. A natural next step is to apply sector
teleportation simulation to the thermal-loss channel, with the aim of
tightening finite-energy upper bounds~\cite{DSW,Goodenough,PLOB} and
clarifying the remaining gap to achievable rates. Recent progress on
secret-key~\cite{OLCSPB,OPB} and entanglement
distillation~\cite{MLG,BPinterpolation,BPdistillation,BPthermal}
provides useful benchmarks for that investigation.

Further directions include finite-blocklength and second-order
refinements~\cite{Mele2025}, experimentally accessible implementations
of sector-resolved measurements, and extensions to other Gaussian
and non-Gaussian bosonic channels. Establishing an energy-dependent
strong converse for interleaved adaptive protocols under a suitable
photon-number constraint remains an open problem.

\section*{AI Statement}
OpenAI's ChatGPT (GPT-6 Astra model) was used during the
preparation of this manuscript to generate candidate mathematical
derivations and proofs, assist with their development, and improve the presentation of the manuscript. The tool was directed through
author-specified prompts concerning the mathematical assumptions,
derivations, and statements under consideration. All AI-assisted
mathematical results, derivations, and proofs were independently
examined, verified, and revised as necessary by the author, who assumes full responsibility for the accuracy and scientific content of the work.
In full detail, the work involved 43 prompts distributed across the
following stages. The first 10 were devoted to the co-development of
sector teleportation simulation, which grew directly from the author's
ongoing efforts to extend channel-simulation techniques developed in his
earlier work over several years.

\begin{itemize}

    \item Sector teleportation simulation (10 prompts): Co-development of
    sector teleportation simulation, including fixed-photon-number sectors,
    finite-energy resource states, the corresponding teleportation LOCC,
    and the reduction of adaptive protocols to a tractable converse.

    \item Mean-energy constraint and adaptive protocols (6 prompts):
    Development and refinement of the converse under an operational
    mean-energy constraint, including adaptive inputs, branch-dependent
    energies, and transmitted-energy accounting.

    \item LOCC mathematical framework (7 prompts): Analysis of
    finite/countable outcomes, continuous-outcome instruments,
    unbounded-round LOCC, closures, and the assumptions required for the
    operational capacity definitions. These were primarily
    verification-oriented investigations, and their details were
    not incorporated into the paper.

    \item Strong-converse threshold (3 prompts): Analysis of the
    strong-converse behaviour under a finite mean-energy constraint and
    its relation to the unconstrained capacity.

    \item Exact finite-energy capacities and quantum-limited amplifier
    extension (9 prompts): Development of the exact energy-constrained
    two-way quantum, entanglement-distribution, private, and secret-key
    capacities of the pure-loss channel, including the optimality of
    reverse coherent information, together with the extension of
    methods and results to the quantum-limited amplifier channel.

    \item Manuscript development and presentation (8 prompts): Development and refinement
    of the abstract, theorem and lemma structure, appendices, proof-flow
    presentation, citations, and conclusions.

\end{itemize}

\smallskip
\textbf{Acknowledgements.~} The author would like to thank past discussions with L. Banchi, J. Pereira, and G. Spedalieri on topics related to teleportation simulation.

\onecolumngrid
\clearpage
\appendix
\section*{Overview of the appendices}
Appendix~\ref{sec:preliminaries} fixes the common operational setting,
error criteria, photon-sector notation, and entanglement tools.
Appendix~\ref{sec:weak} develops sector teleportation simulation and
the pure-loss resource calculation, proves the squeezed-state output
identity, and derives the adaptive REE bound.
Appendices~\ref{sec:quantum} and~\ref{sec:key} establish the pure-loss
quantum and secret-key capacities, respectively; their target tests
and energy-preserving key reduction are channel independent.
Appendix~\ref{sec:amplifier} applies these tools to the quantum-limited
amplifier, deriving the conditional resource correspondence and the
additional energy estimate without repeating the preliminary or
operational arguments. Appendix~\ref{sec:strong} gives a shared proof
of the mean-energy strong-converse thresholds for both channels.
Appendix~\ref{sec:cutoff-strong} proves the parallel photon-cutoff
strong converse for pure loss and extends it to the amplifier.

Figure~\ref{fig:proof-overview} summarizes the dependencies for the proof of the main Theorem~\ref{thm:capacity}.

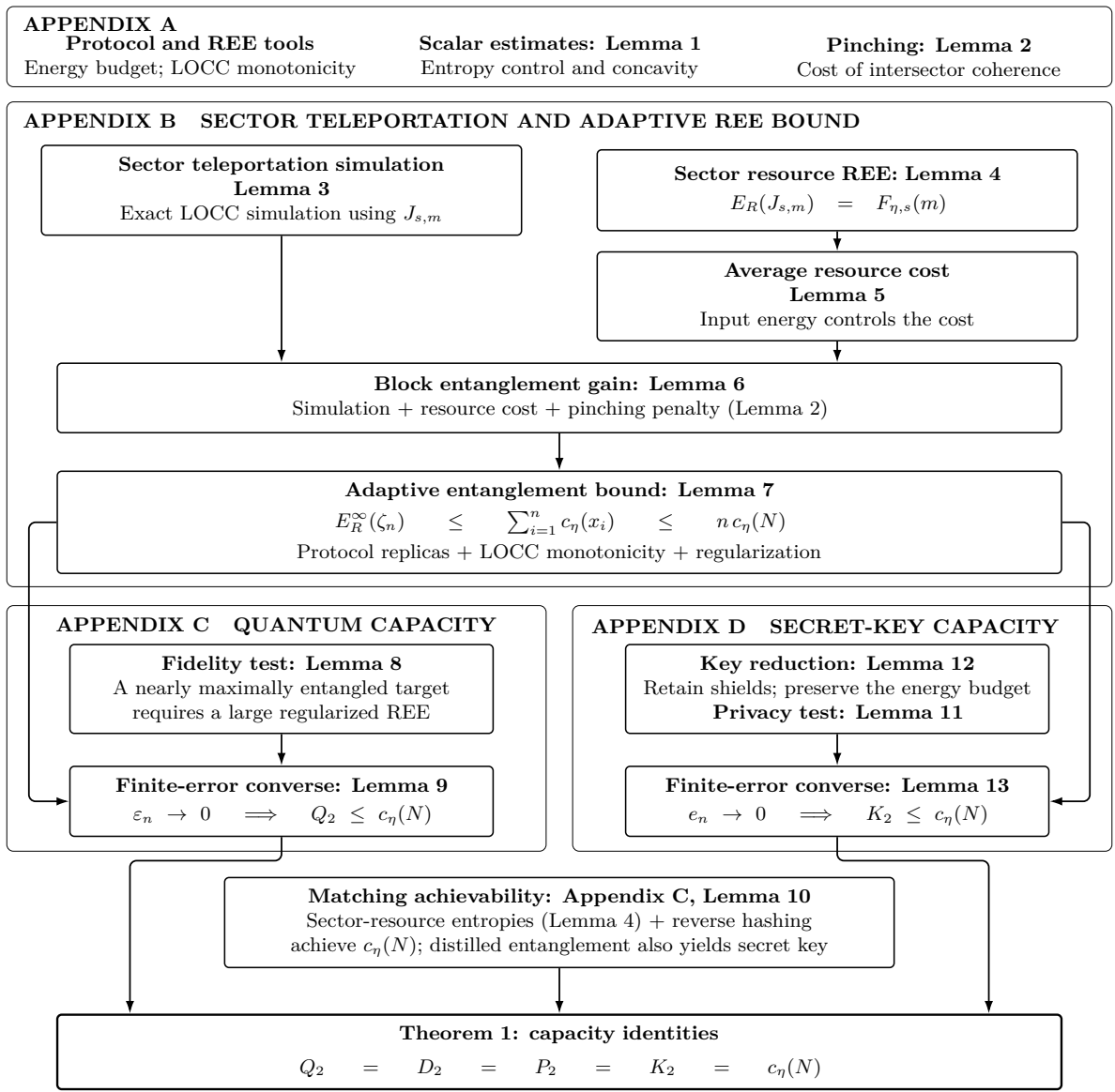
\begin{figure}[!htbp]
\centering
\begingroup
\hypersetup{hidelinks}
\resizebox{0.90\linewidth}{!}{%
\begin{tikzpicture}[x=1cm,y=1cm,
  font=\fontsize{9}{11}\selectfont,
  every node/.style={
    text=black,align=center,inner sep=5pt
  },
  flowbox/.style={
    draw=black,line width=.65pt,
    rounded corners=2pt,fill=white
  },
  panel/.style={
    draw=black,line width=.4pt,rounded corners=3pt
  },
  arrow/.style={
    -{Latex[length=2mm,width=1.35mm]},
    line width=.75pt,draw=black,rounded corners=3pt
  },
  paneltitle/.style={
    anchor=west,
    font=\fontsize{9}{11}\selectfont\bfseries,
    inner sep=0pt
  }
]
\path[use as bounding box]
  (0,1.20) rectangle (17.6,18.60);

\draw[panel]
  (.05,17.28) rectangle (17.55,18.55);

\node[paneltitle] at (.30,18.28)
  {APPENDIX~\ref{sec:preliminaries}};

\node[text width=5.25cm] at (2.95,17.75)
  {\textbf{Protocol and REE tools}\\
   Energy budget; LOCC monotonicity};

\node[text width=5.25cm] at (8.8,17.75)
  {\textbf{Scalar estimates: Lemma~\ref{lem:scalar}}\\
   Entropy control and concavity};

\node[text width=5.25cm] at (14.65,17.75)
  {\textbf{Pinching: Lemma~\ref{lem:pinching}}\\
   Cost of intersector coherence};

\draw[panel]
  (.05,9.36) rectangle (17.55,17.04);

\node[paneltitle] at (.30,16.74)
  {APPENDIX~\ref{sec:weak}\quad
   SECTOR TELEPORTATION AND ADAPTIVE REE BOUND};

\node[flowbox,text width=7.25cm,minimum height=1.28cm]
  (simulation) at (4.4,15.64)
  {\textbf{Sector teleportation simulation}\\
   \textbf{Lemma~\ref{prop:sector}}\\
   Exact LOCC simulation using $J_{s,m}$};

\node[flowbox,text width=7.25cm,minimum height=1.28cm]
  (resource) at (13.2,15.64)
  {\textbf{Sector resource REE: Lemma~\ref{lem:resource}}\\[2pt]
   $E_R(J_{s,m})=F_{\eta,s}(m)$};

\node[flowbox,text width=7.25cm,minimum height=1.08cm]
  (average) at (13.2,13.96)
  {\textbf{Average resource cost}\\
   \textbf{Lemma~\ref{lem:energy}}\\
   Input energy controls the cost};

\draw[arrow] (resource) -- (average);

\node[flowbox,text width=15.55cm,minimum height=1.10cm]
  (block) at (8.8,12.35)
  {\textbf{Block entanglement gain: Lemma~\ref{prop:block}}\\
   Simulation + resource cost + pinching penalty
   (Lemma~\ref{lem:pinching})};

\draw[arrow]
  (simulation.south) -- (simulation.south |- block.north);
\draw[arrow]
  (average.south) -- (average.south |- block.north);

\node[flowbox,text width=15.55cm,minimum height=1.45cm]
  (adaptive) at (8.8,10.38)
  {\textbf{Adaptive entanglement bound:
           Lemma~\ref{lem:adaptive}}\\[3pt]
   $E_R^\infty(\zeta_n)
     \le \sum_{i=1}^{n}c_\eta(x_i)
     \le n\,c_\eta(N)$\\[3pt]
   Protocol replicas + LOCC monotonicity + regularization};

\draw[arrow] (block) -- (adaptive);

\draw[panel]
  (.05,5.17) rectangle (8.59,9.06);
\draw[panel]
  (9.01,5.17) rectangle (17.55,9.06);

\node[paneltitle] at (.82,8.73)
  {APPENDIX~\ref{sec:quantum}\quad QUANTUM CAPACITY};
\node[paneltitle] at (9.30,8.73)
  {APPENDIX~\ref{sec:key}\quad SECRET-KEY CAPACITY};

\node[flowbox,text width=6.35cm,minimum height=1.24cm]
  (qtest) at (4.4,7.74)
  {\textbf{Fidelity test: Lemma~\ref{lem:quantum-test}}\\
   A nearly maximally entangled target\\
   requires a large regularized REE};

\node[flowbox,text width=6.35cm,minimum height=1.24cm]
  (ktest) at (13.2,7.74)
  {\textbf{Key reduction: Lemma~\ref{lem:key-conversion}}\\
   Retain shields; preserve the energy budget\\
   \textbf{Privacy test: Lemma~\ref{lem:privacy}}};

\node[flowbox,text width=6.35cm,minimum height=1.03cm]
  (qupper) at (4.4,5.96)
  {\textbf{Finite-error converse:
           Lemma~\ref{lem:quantum-finite-error}}\\[2pt]
   $\varepsilon_n\to0
     \quad\Longrightarrow\quad Q_2\le c_\eta(N)$};

\node[flowbox,text width=6.35cm,minimum height=1.03cm]
  (kupper) at (13.2,5.96)
  {\textbf{Finite-error converse:
           Lemma~\ref{lem:finite-error}}\\[2pt]
   $e_n\to0
     \quad\Longrightarrow\quad K_2\le c_\eta(N)$};

\draw[arrow] (qtest) -- (qupper);
\draw[arrow] (ktest) -- (kupper);

\draw[arrow]
  (adaptive.west) -- (.40,10.38) |- (qupper.west);
\draw[arrow]
  (adaptive.east) -- (17.20,10.38) |- (kupper.east);

\node[flowbox,text width=10.25cm,minimum height=1.28cm]
  (hashing) at (8.8,4.04)
  {\textbf{Matching achievability:
           Appendix~\ref{sec:quantum}, Lemma~\ref{lem:achieve}}\\
   Sector-resource entropies (Lemma~\ref{lem:resource})
   + reverse hashing\\
   achieve $c_\eta(N)$; distilled entanglement also yields secret key};

\node[flowbox,line width=1.0pt,
      text width=15.55cm,minimum height=1.18cm]
  (capacity) at (8.8,1.99)
  {\textbf{Theorem~\ref{thm:capacity}: capacity identities}\\[4pt]
   $\displaystyle Q_2=D_2=P_2=K_2=c_\eta(N)$};

\draw[arrow]
  (qupper.south) -- (4.4,5.02) -- (2.0,5.02)
  -- ([xshift=-6.8cm]capacity.north);
\draw[arrow]
  (kupper.south) -- (13.2,5.02) -- (15.6,5.02)
  -- ([xshift=6.8cm]capacity.north);
\draw[arrow] (hashing) -- (capacity);

\end{tikzpicture}%
}
\endgroup

\caption{Proof overview for Theorem~\ref{thm:capacity}.
The adaptive REE bound combines with the fidelity and privacy tests
to give the quantum and secret-key weak converses. Sector hashing
supplies matching achievability, and the operational equivalences
$Q_2=D_2$ and $P_2=K_2$ complete the capacity identities.
Here $c_\eta(N)=g(N)-g((1-\eta)N)$, $x_i$ is the unconditional
mean input photon number at use $i$, and $\zeta_n$ is the retained
final state after $n$ channel uses.}
\label{fig:proof-overview}
\end{figure}

\section{Preliminaries}\label{sec:preliminaries}
\subsection{Notation, entropies, and the loss channel}
In this section, we fix the state, entropy, and optical notation used to describe the
simulation and its energy-dependent entanglement cost. All logarithms are base two unless $\ln$ is written. A state is a
positive trace-one operator. We write $\Tr$ for the trace,
$\rho_A=\Tr_B\rho_{AB}$ for a reduced state, $\supp\rho$ for its support,
$I_A$ for the identity operator, and $\id_A$ for the identity channel.
The trace norm is $\|T\|_1=\Tr\sqrt{T^\dagger T}$; an unsubscripted
vector norm is its Hilbert-space norm. A ket denotes a vector, and the
same symbol without a ket denotes its rank-one projector when appropriate. We call $|u;r\rangle$ an $r$-photon condensate:
all $r$ photons occupy the same normalized collective mode $u$. For example, for two modes and $u=(1,1)/\sqrt2$,
\begin{equation}
 \ket{u;2}
 =\frac12\ket{2,0}
  +\frac1{\sqrt2}\ket{1,1}
  +\frac12\ket{0,2}.
\end{equation}
Both photons occupy the same collective mode, while their
occupations in the two physical modes are in superposition. For a state $\rho$ and a probability distribution $p=(p_z)$, define
the von Neumann, Shannon, and binary entropies by
\begin{equation}
 S(\rho)=-\Tr(\rho\log_2\rho),\qquad
 H(p)=-\sum_zp_z\log_2p_z,\qquad
 h_2(u)=-u\log_2u-(1-u)\log_2(1-u).
 \label{eq:entropy-definitions}
\end{equation}
We use $0\log_2 0=0$ and $0\le u\le1$. For a real number $z$,
write $z_+:=\max\{z,0\}$ for its positive part.
For classical random variables, $H(Z)$ is the entropy of their
distribution, $H(Z,W)$ is the joint entropy, and
$H(Z|W)=\sum_w\Pr(W=w)H(Z|W=w)$ is the average conditional entropy.
The symbols $\Pr$ and $\EE$ denote probability and expectation.
$\Bin(m,p)$ is the binomial distribution with $m$ trials and success
probability $p$: its mass at $\ell$ is
$\binom m\ell p^\ell(1-p)^{m-\ell}$ for $0\le\ell\le m$. Recall that $0<\eta<1$ is the transmissivity and $q=1-\eta$.
Pure loss acts on the input, output, and environment annihilation
operators as $\hat b=\sqrt\eta\,\hat a+\sqrt q\,\hat e$, with the
environment initially vacuum. This is the quantum-limited attenuator
in the Gaussian-channel formalism~\cite{BraunsteinLoock,HolevoBook,Weedbrook}.
The input photon-number operator is $\hat n_A=\hat a^\dagger\hat a$.
An input of mean photon number $x$ has mean output photon number $\eta x$.
For the amplifier extension, the channel law is \eqref{eq:amp-channel};
its output mean is $\kappa x+\kappa-1$. Eve receives the channel
environments and the complete public transcript.

\subsection{Operational tasks, adaptive protocols, and error criteria}
\label{sec:task-errors}
\label{sec:protocol-model}
We specify the common protocol model for pure loss and its amplifier
extension. Write $\mathcal N$ for either $\cL_\eta$ or
$\cA_\kappa$. An $n$-use protocol starts from local preparations
and public randomness, without preshared entanglement or secret key.
Between channel calls, Alice and Bob apply LOCC to their memories
and classical records, conditioned on the preceding outcomes.
For each finite protocol, the number of local rounds is finite and
the outcome alphabets are finite or countable; the round bound may
grow with $n$~\cite{ChitambarLOCC}. States and local instruments are
normal on separable Hilbert spaces. Memories and local apparatus
have no dimension, energy, or entropy bound. Eve receives the channel
environments and every public message.

Just before use $i$, the joint state is
$\rho^{(i)}_{X_iA_iY_i}$: Alice holds her memory $X_i$ and the
transmitted mode $A_i$, while Bob holds $Y_i$. The unconditional
mean $x_i=\Tr(\hat n_{A_i}\rho^{(i)}_{A_i})$ averages over every
outcome history and random choice, including unsuccessful branches.
Only these transmitted inputs enter the budget
$\sum_i x_i\le nN$ in \eqref{eq:energy}. A branch means a particular
history of local outcomes and random choices; no separate energy
bound is imposed on each such history.

For a positive integer $d$, define the maximally entangled target
and squared fidelity by
\begin{equation}
 |\Phi_d\rangle=\frac1{\sqrt d}\sum_{j=0}^{d-1}|j\rangle|j\rangle,
 \qquad \Phi_d=|\Phi_d\rangle\langle\Phi_d|,\qquad
 F(\rho,\sigma)=\|\sqrt\rho\sqrt\sigma\|_1^2.
 \label{eq:target-fidelity}
\end{equation}
The quantum capacity $Q_2(\mathcal N,N)$ concerns transmission of a
$d$-dimensional system, tested on half of $\Phi_d$ while its reference
remains at Alice. The entanglement-distribution capacity
$D_2(\mathcal N,N)$ concerns generating $\Phi_d$ from an initially
unentangled preparation. After discarding unused memories, let
$\rho_n$ denote the final state of the two dimension-$d_n$ target
registers. The rate, fidelity, and error are
\begin{equation}
 R_n=\frac{\log_2d_n}{n},\qquad
 f_n=F(\rho_n,\Phi_{d_n})=\Tr(\Phi_{d_n}\rho_n),\qquad
 \varepsilon_n=1-f_n.
 \label{eq:quantum-error}
\end{equation}
The same convention applies to the reference--output state in
quantum transmission. Teleportation and entanglement distribution
give $Q_2=D_2$ without further optical transmission, as explained
in Section~\ref{sec:results}.

The secret-key capacity $K_2(\mathcal N,N)$ concerns agreement on a
uniform classical key independent of Eve. The private capacity
$P_2(\mathcal N,N)$ concerns transmission of Alice's classical
message, with a uniform message used to define the rate and error.
For either task, let $C_A,C_B$ be the classical output registers of
size $d$, and let $E_{\rm all}$ contain Eve's systems and the complete
public transcript. For their actual joint state $\nu$, define
\begin{equation}
 \tau_d^{C_AC_B}=\frac1d\sum_{a=0}^{d-1}|aa\rangle\langle aa|,
 \qquad
 e_{\rm cl}(\nu)=\tfrac12\bigl\|
 \nu_{C_AC_BE_{\rm all}}-\tau_d^{C_AC_B}\otimes\nu_{E_{\rm all}}
 \bigr\|_1.
 \label{eq:classical-task-error}
\end{equation}
For key agreement, $C_A,C_B$ are the generated keys; for private
transmission, they are Alice's message and Bob's decoded message.
The ideal state retains the actual Eve marginal, including any
public ciphertext~\cite{RennerKonig,PortmannRenner}.

Each capacity is the supremum of rates $R$ attained by admissible
sequences with $\liminf_n R_n\ge R$ and the corresponding error
tending to zero. Rates are in qubits, ebits, or secret bits per use.
The systems in the pure-loss quantum and key protocols are described
in Appendices~\ref{sec:quantum-model} and~\ref{sec:key-model}; the
amplifier uses the same operational structure.

\subsection{Relative entropy of entanglement}
\label{sec:ree-facts}
For density operators $\rho,\sigma$, the quantum relative entropy is
$D(\rho\Vert\sigma)=\Tr[\rho(\log_2\rho-\log_2\sigma)]$ when
$\supp\rho\subseteq\supp\sigma$, with its standard extended-value
interpretation; it is $+\infty$ otherwise~\cite{Umegaki,VedralReview}.
For classical distributions, the same notation means
$D(p\Vert q)=\sum_a p_a\log_2(p_a/q_a)$, with the corresponding support
convention. Let $\SEP(A:B)$ be the trace-norm closed convex hull of
product states across $A:B$, including in infinite
dimensions~\cite{HolevoShirokovWerner}.
The relative entropy of entanglement
(REE)~\cite{Vedral1997,VedralPlenio1998} and its regularization are
\begin{equation}
 \ER(A:B)_\rho=\inf_{\sigma\in\SEP(A:B)}D(\rho\Vert\sigma),\qquad
 \ERi(\rho)=\lim_{t\to\infty}\frac1t\ER(\rho^{\otimes t}).
 \label{eq:ree}
\end{equation}
The bipartition is omitted when it is clear; $t$ is a positive integer
counting copies.
When $\ER(\rho)<\infty$, the limit exists by subadditivity.
Regularization is standard in asymptotic entanglement-distillation
bounds~\cite{LimitsEntanglement,HorodeckiReview}.
The following consequences of data processing and joint convexity
will be used repeatedly~\cite{Lindblad,UhlmannEntropy,Watrous}.
We first work in finite-dimensional spaces. The infinite-dimensional
extension of the block inequality is given separately in
the proof of \Cref{prop:block}, so the following entropy manipulations never
subtract two infinite quantities.

For a separability-preserving channel $\Lambda$, data processing gives
$\ER(\Lambda(\rho))\le\ER(\rho)$: apply data processing to each separable
comparison $\sigma$, then take the infimum. Joint convexity of relative
entropy similarly gives convexity of $\ER$. Product separable comparisons
and additivity of relative entropy on tensor products give
\begin{equation}
 \ER(\rho\otimes\xi)\le\ER(\rho)+\ER(\xi).
 \label{eq:subadd}
\end{equation}
These statements also hold for the trace-norm closed separable set in
infinite dimensions, using approximate minimizers when necessary.

For a local instrument with outcomes $a$, apply data processing while
retaining the outcome as a local classical flag. If $p_a,\rho_a$ and
$q_a,\sigma_a$ are the respective outcome ensembles for $\rho$ and a
separable comparison $\sigma$, then
\begin{equation}
 D(\rho\Vert\sigma)
 \ge D(p\Vert q)+\sum_a p_aD(\rho_a\Vert\sigma_a)
 \ge\sum_a p_a\ER(\rho_a).
 \label{eq:strong}
\end{equation}
Thus $\sum_a p_a\ER(\rho_a)\le\ER(\rho)$.
In particular, if the branches have orthogonal supports on one party,
measuring those supports and using convexity in the reverse direction gives
\begin{equation}
 \ER\!\left(\sum_a p_a\rho_a\right)=\sum_a p_a\ER(\rho_a).
 \label{eq:flags}
\end{equation}
The flags may be physical orthogonal subspaces rather than separate
registers. These are the usual local-flags and strong-monotonicity
properties of REE~\cite{VedralPlenio1998,PlenioVirmani}.
The same formulas hold for countably many outcomes. The relative
entropy of classical direct sums is the classical relative entropy
plus the weighted conditional relative entropies; countable
convexity follows by finite grouping and lower semicontinuity.
These facts allow photon-number flags with unbounded support.

For finite-dimensional states, we also use the conditional-entropy
lower bound of Plenio, Virmani, and Papadopoulos
\cite[Lemma~5]{PlenioVirmaniPapadopoulos}:
\begin{equation}
 \ER(R:B)_\rho\ge S(\rho_R)-S(\rho_{RB}).
 \label{eq:rci-lower}
\end{equation}
This supplies the matching lower bound in the sector REE calculation.
The same quantity is achievable by reverse hashing
\cite{DevetakWinter,RCI}, also interpreted through quantum state
merging~\cite{StateMerging}; the REE inequality itself does not
require a distillation argument. Note that, whenever external memories are omitted
from a channel expression, identity channels on them are implicitly understood.

\subsection{Entropy bounds for photon counts}
We bound the uncertainty of a total photon count using only its mean.
These scalar bounds control the overhead from photon-number conditioning
and make that overhead vanish per mode for large blocks. For $y>0$, the one-mode thermal state of mean photon number $y$ is
\begin{equation}
 \tau_y=\sum_{r=0}^{\infty}\frac{y^r}{(1+y)^{r+1}}|r\rangle\langle r|,
 \qquad S(\tau_y)=g(y),
 \qquad \tau_0=|0\rangle\langle0|.
 \label{eq:thermal-state}
\end{equation}
Here $|r\rangle$ is the $r$-photon number state. Recall from
\eqref{eq:cdef} that $g(y)=(y+1)\log_2(y+1)-y\log_2y$ and
$c_\eta(y)=g(y)-g(qy)$, with $q=1-\eta$.

\begin{lemma}[Scalar entropy estimates]\label{lem:scalar}
If a nonnegative integer-valued random variable $Z$ has mean $z<\infty$,
then $H(Z)\le g(z)$. The function $c_\eta$ is increasing and concave, and
$g(ra)/r\to0$ for every fixed finite $a\ge0$.
\end{lemma}
\begin{proof}
The geometric distribution maximizes entropy at fixed nonnegative
integer mean~\cite{CoverThomas}: relative entropy to that distribution
is $g(z)-H(Z)\ge0$ whenever the entropy is finite. To avoid assuming
this finiteness, first apply the same comparison to
$Z_J=\min\{Z,J\}$. Its finite entropy is at most $g(z)$ because
$\mathbb E Z_J\le z$. The associated partitions refine as $J$ grows,
so $H(Z_J)\uparrow H(Z)$ and the bound follows. The case $z=0$ is
deterministic. For $x>0$, differentiation gives
\begin{equation}
 c_\eta''(x)=-\frac{\eta}{\ln2\,x(1+x)(1+qx)}<0,
 \qquad \lim_{x\to\infty}c_\eta'(x)=0.
 \label{eq:main-concavity}
\end{equation}
Thus $c_\eta$ is increasing and concave, including at zero by continuity.
Finally $g(x)=\log_2x+\log_2e+o(1)$ as $x\to\infty$; the case $a=0$
follows from $g(0)=0$.
\end{proof}

\subsection{A local pinching bound}
Here, we quantify the entanglement contribution of coherence between the
outcomes of a local measurement.

\begin{lemma}[Local pinching penalty]\label{lem:pinching}
For a finite-dimensional bipartite state $\zeta$, let $\{P_k\}$
be mutually orthogonal projectors on Bob that sum to the identity;
identities on Alice are implicit. The pinching map
$\cP(\zeta)=\sum_kP_k\zeta P_k$ is this measurement with its outcome
forgotten. If $p=(p_k)$ with $p_k=\Tr(P_k\zeta)$, then
\begin{equation}
 \ER(\zeta)\le\ER(\cP\zeta)+D(\zeta\Vert\cP\zeta)
                  \le\ER(\cP\zeta)+H(p).
 \label{eq:main-pinching}
\end{equation}
\end{lemma}

\begin{proof}
An arbitrarily near-optimal separable comparison for $\cP\zeta$ can
itself be pinched: this preserves separability and cannot increase its
relative entropy to $\cP\zeta$. Denote the resulting block-diagonal
comparison by $\sigma$. The relative-entropy decomposition for
pinching follows by expanding the logarithms
(see also Refs.~\cite{Bhatia,Watrous}):
\begin{equation}
 D(\zeta\Vert\sigma)
 =D(\zeta\Vert\cP\zeta)+D(\cP\zeta\Vert\sigma),
\label{eq:pinching-pythagorean}
\end{equation}
because a block-diagonal observable has the same expectation on $\zeta$
and $\cP\zeta$. Taking the infimum proves the first inequality. To prove the second, purify $\zeta$ to $\ket\psi$. The projected vectors
$P_k\ket\psi/\sqrt{p_k}$ are orthogonal, so the relative entropy of
$\proj\psi$ to its pinching is exactly $H(p)$. Tracing out the purifier
can only decrease relative entropy, giving
$D(\zeta\Vert\cP\zeta)\le H(p)$.
\end{proof}
This lemma quantifies a loss of coherence in an inequality; it does not
claim that LOCC can undo the measurement. For the optical block, the label
$k$ will be one total number, rather than the list of all mode occupations.

\subsection{Photon-number sectors and resources}
We identify fixed total-photon-number blocks as finite-dimensional
quantum systems and define the resources for sector teleportation
simulation. Their photon-number label will carry the input constraint
into the resource entanglement bound. Let $s\ge1$ and $m\ge0$ be integers. Write $A^s=A_1\cdots A_s$
for a block of $s$ input modes and $B^s=B_1\cdots B_s$ for its output.
An occupation vector is $\bm r=(r_1,\ldots,r_s)$ with nonnegative
integer entries; set $|\bm r|=\sum_jr_j$, $\bm r!=\prod_jr_j!$, and
$|\bm r\rangle=\bigotimes_j|r_j\rangle$. Let $\cH_{s,m}$ be the
subspace with total number $m$, $I_{s,m}$ its identity operator, and
$\pi_{s,m}$ its maximally mixed state:
\begin{equation}
 \cH_{s,m}=\spanop\{\ket{\bm r}:|\bm r|=m\},\qquad
 d_{s,m}=\binom{s+m-1}{m},\qquad \pi_{s,m}=I_{s,m}/d_{s,m}.
 \label{eq:sector}
\end{equation}
The total photon-number operator is
$\hat N_A=\sum_{j=1}^s\hat n_{A_j}$. We denote the projector onto
$\cH_{s,m}$ by $P_m^A$, and use $P_k^B$ for the corresponding output
projector.
The sector is the symmetric tensor power $\mathrm{Sym}^m(\mathbb C^s)$,
whose dimension is $d_{s,m}$~\cite{FultonHarris,Harrow}.
For a unit vector $u\in\mathbb C^s$ and an integer $r\ge0$, write
$\ket{u;r}=u^{\otimes r}$ for the normalized $r$-photon condensate,
viewed in $\cH_{s,r}$. These are coherent states of the symmetric
representation~\cite{Perelomov}. With $\dd u$ the normalized uniform
measure over complex unit vectors, their resolution of the identity is
\begin{equation}
 \int\proj{u;r}\,\dd u=\frac{I_{s,r}}{d_{s,r}}=\pi_{s,r}.
 \label{eq:coherent-resolution}
\end{equation}
This is the standard symmetric-subspace identity
of Ref.~\cite[Proposition~6]{Harrow}, also obtainable from Haar
integration~\cite{CollinsSniady}. In particular, the condensates span
each sector. The notation $\ket{\bar u;r}$ conjugates the coefficients
in the fixed occupation basis.

Let $R$ be a reference register isomorphic to $\cH_{s,m}$ and let
$A^{\prime s}$ be an auxiliary input block. The maximally entangled
sector pair and its normalized Choi--Jamio\l{}kowski
resource~\cite{Jamiolkowski,Choi} are
\begin{equation}
 \ket{\Phi_{s,m}}=\frac1{\sqrt{d_{s,m}}}
       \sum_{|\bm r|=m}\ket{\bm r}_R\ket{\bm r}_{A^{\prime s}},\qquad
 J_{s,m}=(\id_R\otimes\cL_\eta^{\otimes s})(\Phi_{s,m}).
 \label{eq:choi}
\end{equation}
(As mentioned before, we use the same letter for a pure-state vector and, without a ket, its projector when the meaning is unambiguous.) The auxiliary block $A^{\prime s}$ used to define the resource is
distinct from the actual channel input $A^s$. The resource is supplied
independently of the state being simulated. Its preparation formula
defines a comparison state; it is not an extra online channel call.
The pair $\Phi_{s,m}$ has exactly $m$ photons in each of its two
shares, hence $2m$ altogether. In $J_{s,m}$, Alice still has exactly
$m$ photons, whereas Bob has at most $m$ and mean photon number
$\eta m$. This is the precise resource accounting associated with an
$m$-photon input sector.

\section{Sector teleportation simulation and adaptive REE bound}\label{sec:weak}
This appendix supplies the common estimate for the quantum and key tasks.
Recall that $q=1-\eta$, $c_\eta(x)=g(x)-g(qx)$, and $x_i$ is the
unconditional mean photon number entering channel use $i$.
We construct the sector simulator, calculate its resource REE, and
average over the total photon number. Replicas of a fixed adaptive
protocol remove the sublinear overhead and bound the regularized REE
of its actual output (adaptive REE bound). The output-state identity in \Cref{cor:tmsv}
is proved alongside the resource calculation. The quantum and key
applications follow in
Appendices~\ref{sec:quantum} and~\ref{sec:key}.

\subsection{Exact simulation on a photon-number sector}\label{sec:teleportation-detail}
We develop \emph{sector teleportation simulation}, a construction that
encodes the input photon budget in a finite-energy resource and
exactly reproduces the pure-loss channel on each fixed-photon-number
sector, preserving all correlations with external memories. This
resource encoding provides the energy accounting needed for our
adaptive weak converse. The construction builds on established teleportation
principles~\cite{BennettTeleport,BraunsteinKimble,WernerTeleport}
and the irreducibly covariant channel-simulation framework of
Ref.~\cite[Proposition~24]{GBW}. Finite-resource Gaussian teleportation and stretching
have also been studied in
Refs.~\cite{LiuzzoScorpo,LaurenzaFinite}. Our contribution is to
develop a simulation based on multimode photon sectors whose
resource retains the input photon constraint, allowing that
constraint to enter the resource-entanglement bound.

Figure~\ref{fig:sector-teleportation} separates the preparation of
the resource from its use to simulate the channel on the actual
input. The following lemma specifies the joint measurement,
classical communication, conditional correction, and resource
energy, and proves exactness for inputs entangled with external
memories. Here $X$ and $Y$ denote Alice's and Bob's retained
memories, respectively, and $A^s$ is the transmitted block of
$s$ modes. The sector $\cH_{s,m}$ and resource $J_{s,m}$ are
defined in \eqref{eq:sector} and \eqref{eq:choi}, respectively.
We state the covariance extension explicitly so that the amplifier
can use the same Bell measurement in Appendix~\ref{sec:amp-simulation}.

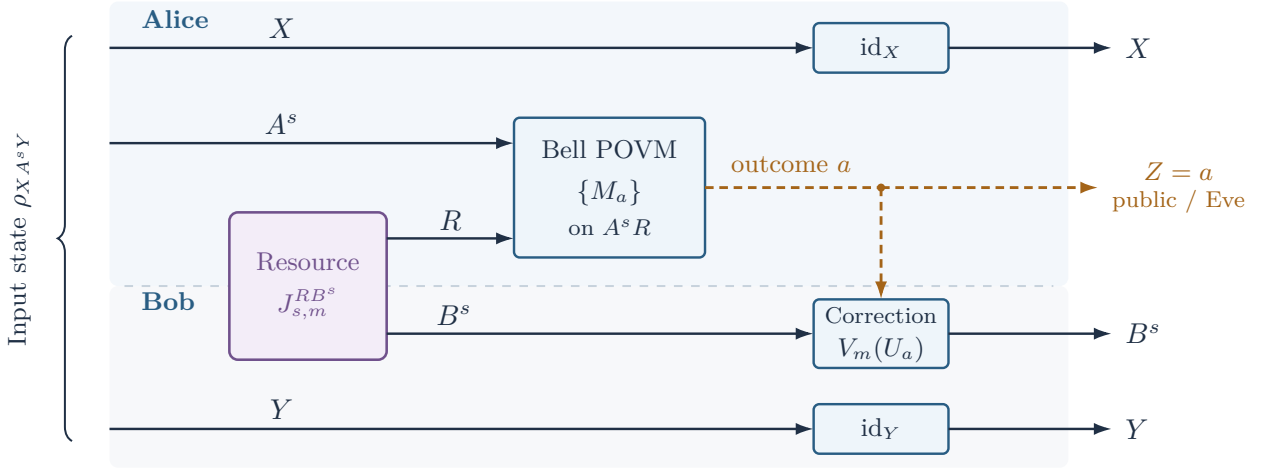
\begin{figure}[t!]
\centering
\begingroup
\definecolor{ink}{HTML}{203045}
\definecolor{muted}{HTML}{5D6876}
\definecolor{blue}{HTML}{295D83}
\definecolor{bluefill}{HTML}{EEF5FA}
\definecolor{purple}{HTML}{70518C}
\definecolor{purplefill}{HTML}{F3EDF8}
\definecolor{amber}{HTML}{A46418}
\definecolor{amberfill}{HTML}{FBF5E8}
\definecolor{alicefill}{HTML}{F3F7FB}
\definecolor{bobfill}{HTML}{F7F8FA}
\definecolor{rulegray}{HTML}{C7D2DC}

\tikzset{
 every node/.style={text=ink},
 quantum/.style={draw=ink,line width=.85pt,
   -{Latex[length=2.2mm,width=1.5mm]}},
 qline/.style={draw=ink,line width=.85pt},
 classical/.style={draw=amber,line width=1pt,densely dashed,
   -{Latex[length=2.2mm,width=1.5mm]}},
 classicalplain/.style={draw=amber,line width=1pt,densely dashed},
 gate/.style={draw=blue,fill=bluefill,line width=.85pt,
   rounded corners=2pt},
 resource/.style={draw=purple,fill=purplefill,line width=.9pt,
   rounded corners=3pt},
 label/.style={font=\normalsize,inner sep=2pt},
 note/.style={font=\footnotesize,text=muted,align=center},
 panel/.style={font=\bfseries,anchor=west},
 lab/.style={font=\small\bfseries,text=blue,anchor=west}
}

\resizebox{0.96\linewidth}{!}{%
\begin{tikzpicture}[x=.92cm,y=.78cm]

\node[panel] at (.1,14.05)
  {(a) Resource definition};

\draw[resource] (1.2,11.05) rectangle (3.55,13.18);
\node[align=center,text=purple] at (2.375,12.12)
  {Sector pair\\[4pt] $\Phi_{s,m}^{RA'^s}$};

\draw[quantum] (3.55,12.8)--(12.15,12.8);
\node[label,anchor=south] at (4.45,12.83) {$R$};
\node[label,anchor=west] at (12.18,12.8) {$R$};

\draw[gate] (6.05,9.9) rectangle (8.95,11.92);
\node[align=center] at (7.5,10.9)
  {Loss dilation\\[4pt] $U_\eta^{\otimes s}$};
\draw[quantum] (3.55,11.4)--(6.05,11.4);
\node[label,anchor=south] at (4.65,11.43) {$A'^s$};
\node[label,anchor=east] at (4.65,10.3) {$|0\rangle_{E_0^s}$};
\draw[quantum] (4.75,10.3)--(6.05,10.3);

\draw[quantum] (8.95,11.4)--(12.15,11.4);
\node[label,anchor=west] at (12.18,11.4) {$B^s$};
\draw[quantum] (8.95,10.3)--(10.75,10.3);
\node[label,anchor=south] at (9.8,10.33) {$E^s$};
\node[note,anchor=west] at (10.82,10.3) {traced out};

\draw[draw=purple,line width=.75pt,decorate,
  decoration={brace,amplitude=4pt}]
  (13.05,13.02)--(13.05,11.18);
\node[align=center,text=purple] at (14.9,12.15)
  {Shared resource\\[4pt] $J_{s,m}^{RB^s}$};

\node[panel] at (.1,8.25)
  {(b) Sector teleportation simulation};

\fill[alicefill,rounded corners=3pt]
  (1.4,3.2) rectangle (14.2,7.68);
\fill[bobfill,rounded corners=3pt]
  (1.4,.34) rectangle (14.2,3.2);
\draw[draw=rulegray,dashed,line width=.7pt]
  (1.6,3.2)--(14.0,3.2);
\node[lab] at (1.7,7.4) {Alice};
\node[lab] at (1.7,2.96) {Bob};

\draw[draw=ink,line width=.7pt,decorate,
  decoration={brace,mirror,amplitude=5pt}]
  (.9,7.14)--(.9,.76);
\node[rotate=90,align=center] at (.2,3.95)
  {Input state $\rho_{XA^sY}$};

\draw[gate] (10.8,6.56) rectangle (12.6,7.34);
\node at (11.7,6.95) {$\mathrm{id}_X$};
\draw[quantum] (1.4,6.95)--(10.8,6.95);
\draw[quantum] (12.6,6.95)--(14.78,6.95);
\node[label,anchor=south] at (3.7,7.0) {$X$};
\node[label,anchor=west] at (14.88,6.95) {$X$};

\draw[resource] (3.0,2.04) rectangle (5.1,4.36);
\node[align=center,text=purple] at (4.05,3.2)
  {Resource\\[4pt] $J_{s,m}^{RB^s}$};

\draw[gate] (6.8,3.65) rectangle (9.35,5.85);
\node[align=center] at (8.075,4.75)
  {Bell POVM\\[4pt] $\{M_a\}$\\[2pt]
   {\footnotesize on $A^sR$}};
\draw[quantum] (1.4,5.45)--(6.8,5.45);
\node[label,anchor=south] at (3.7,5.5) {$A^s$};
\draw[quantum] (5.1,3.95)--(6.8,3.95);
\node[label,anchor=south] at (5.95,3.99) {$R$};

\draw[gate] (10.8,1.91) rectangle (12.6,2.99);
\node at (11.7,2.2) {$V_m(U_a)$};
\node[font=\footnotesize] at (11.7,2.74) {Correction};
\draw[quantum] (5.1,2.45)--(10.8,2.45);
\node[label,anchor=south] at (6.0,2.49) {$B^s$};
\draw[quantum] (12.6,2.45)--(14.78,2.45);
\node[label,anchor=west] at (14.88,2.45) {$B^s$};

\draw[gate] (10.8,.56) rectangle (12.6,1.34);
\node at (11.7,.95) {$\mathrm{id}_Y$};
\draw[quantum] (1.4,.95)--(10.8,.95);
\draw[quantum] (12.6,.95)--(14.78,.95);
\node[label,anchor=south] at (3.7,1.0) {$Y$};
\node[label,anchor=west] at (14.88,.95) {$Y$};

\draw[classicalplain] (9.35,4.75)--(11.7,4.75);
\fill[amber] (11.7,4.75) circle (1.4pt);
\draw[classical] (11.7,4.75)--(11.7,2.99);
\node[font=\small,text=amber,anchor=south] at (10.5,4.85)
  {outcome $a$};
\draw[classical] (11.7,4.75)--(14.6,4.75);
\node[align=center,text=amber,anchor=west] at (14.66,4.75)
  {$Z=a$\\[-1pt]{\footnotesize public / Eve}};

\end{tikzpicture}
}
\endgroup
\caption{Sector teleportation simulation on a fixed total-photon-number subspace.
(a) Definition of the independent resource $J_{s,m}^{RB^s}$ using an
auxiliary optical block $A^{\prime s}$ and a vacuum environment.
The sector pair has $m$ photons on each share; after loss, $R$ still
has $m$ photons and $B^s$ has at most $m$. The
environment output is traced out. (b) Alice jointly measures the actual
input $A^s$ and her resource share $R$, sends the finite outcome $a$, and
Bob corrects $B^s$. The memories $X,Y$ retain all their correlations.
Solid arrows carry quantum systems; the dashed arrow carries the public
outcome, also available to Eve. For an input supported on $\cH_{s,m}$,
the output is exactly the state in \eqref{eq:main-simulation}.}
\label{fig:sector-teleportation}
\end{figure}

\begin{lemma}[Sector telesimulation]\label{prop:sector}
There is an LOCC map $\mathcal T_{s,m}$ such that, for every
$\rho_{XA^sY}$ supported on $\cH_{s,m}$ in $A^s$, we have
\begin{equation}
 \mathcal T_{s,m}(\rho\otimes J_{s,m})
       =(\id_X\otimes\cL_\eta^{\otimes s}\otimes\id_Y)(\rho).
 \label{eq:main-simulation}
\end{equation}
Alice holds $XA^sR$ initially and $X$ finally; Bob holds his resource
share and $Y$ initially, and $B^sY$ finally. The resource has exactly
$m$ photons at Alice and at most $m$ photons at Bob, with Bob's mean
equal to $\eta m$.

More generally, the same Bell measurement simulates any normal channel
$\mathcal N$ from $\cH_{s,m}$ to a separable output space whenever
there are unitary output corrections $V_m(U)$ satisfying the
covariance identity \eqref{eq:sector-covariance} below. Its resource
is $J^{\mathcal N}_{s,m}=(\id_R\otimes\mathcal N)(\Phi_{s,m})$.
The output space need not be finite dimensional.
\end{lemma}

\begin{proof}
Fix $s,m$ and write $d=d_{s,m}$. We use $\mathcal N$ for a channel
with the stated covariance, and verify that pure loss supplies the
required representation. Alice holds the input $A^s$,
her memory $X$, and the resource share $R$; Bob holds his memory
$Y$ and the resource share $B^s$. The auxiliary block $A^{\prime s}$
appears only in the resource preparation formula
\begin{equation}
 (J^{\mathcal N}_{s,m})^{RB^s}
 =(\id_R\otimes\mathcal N)
   (\Phi_{s,m}^{RA^{\prime s}}).
\end{equation}
The state supplied to the simulator is
$\rho_{XA^sY}\otimes(J^{\mathcal N}_{s,m})^{RB^s}$. Thus the resource is
independent of the input, whereas $A^s$ may be arbitrarily
correlated with $X$ and $Y$. We assume only
$\supp\rho_{A^s}\subseteq\cH_{s,m}$.
For $m=0$, the input sector is one dimensional. Its correlations
with the memories factorize, and the resource already contains
$\mathcal N(\proj{0^{\otimes s}})$ at Bob. Discarding the input
therefore proves the claim. We consider $m\ge1$ in the following steps.

\smallskip\noindent\emph{1. Covariance supplies the output corrections.}
A passive mode transformation $U\in U(s)$ acts on the
$m$-photon sector through
$W_m(U)=\mathrm{Sym}^mU$. This representation is
irreducible~\cite{FultonHarris}. Indeed, an operator commuting
with all single-mode phase rotations must be diagonal in the
occupation basis. Commutation with the generators of mode mixing
then equates the diagonal entries of occupations connected by
moving one photon between modes. Such moves connect all
occupations with total photon number $m$, so the operator is
a multiple of the identity.

Identical pure-loss channels are covariant under these passive
transformations. To see this, rotate the system and environment
modes by the same $U$ in the beam-splitter dilation. This rotation
commutes with the identical beam splitters, and the environment
vacuum is invariant. Since an $m$-photon input produces at most
$m$ output photons, define
$V_m(U)=\bigoplus_{k=0}^{m}W_k(U)$ on the relevant output space.
For pure loss, and by assumption for the general channel, every
operator $T$ on $\cH_{s,m}$ satisfies
\begin{equation}
 \mathcal N\!\left(W_m(U)T W_m(U)^\dagger\right)
 =V_m(U)\mathcal N(T)V_m(U)^\dagger.
 \label{eq:sector-covariance}
\end{equation}
Thus a rotation of the input can be corrected by a corresponding
rotation of the output.

\smallskip\noindent\emph{2. Irreducibility supplies a finite Bell POVM.}
On $A^sR$, define the normalized Bell vectors
\begin{equation}
 \ket{\beta_U}
 =(W_m(U)\otimes I_R)\ket{\Phi_{s,m}}.
\label{eq:sector-bell-vector}
\end{equation}
Let $\dd U$ be normalized Haar measure on $U(s)$.
Irreducibility gives the Schur twirl~\cite{FultonHarris,Harrow},
\begin{equation}
 \int W_m(U)T W_m(U)^\dagger\,\dd U
 =\frac{\Tr T}{d}\,I.
\label{eq:sector-schur-twirl}
\end{equation}
Applying this identity to one half of the maximally entangled
state yields
\begin{equation}
 \int\proj{\beta_U}\,\dd U
 =\frac{I_{A^sR}}{d^2}.
 \label{eq:main-bell}
\end{equation}

We now replace this continuous average by a finite one.
The Bell projectors form a compact set in the real affine space
of trace-one Hermitian operators on $A^sR$, whose dimension is
$d^4-1$. Carath\'{e}odory's theorem~\cite{Rockafellar}
therefore provides at most $d^4$ unitaries $U_a$ and positive
weights $w_a$, with $\sum_a w_a=1$, such that
\begin{equation}
 \sum_a w_a\proj{\beta_{U_a}}=\frac{I_{A^sR}}{d^2},
 \qquad
 M_a=d^2w_a\proj{\beta_{U_a}},
 \qquad
 \sum_aM_a=I_{A^sR}.
 \label{eq:finite-bell}
\end{equation}
The effects $M_a$ define a finite-outcome POVM on the joint
sector space $A^sR$. The Bell vectors need not be orthogonal;
the completeness relation in \eqref{eq:finite-bell} is sufficient.

\smallskip\noindent\emph{3. Each measurement outcome gives the correct output.}
First consider the resource before the channel is applied.
Expansion in the common occupation basis gives, for every
$\ket\psi\in\cH_{s,m}$,
\begin{equation}
 d(\bra{\beta_U}_{A^sR}\otimes I_{A^{\prime s}})
 \left(\ket\psi_{A^s}\otimes
       \ket{\Phi_{s,m}}_{RA^{\prime s}}\right)
 =W_m(U)^\dagger\ket\psi_{A^{\prime s}}.
 \label{eq:main-contraction}
\end{equation}
This identity also holds with an arbitrary reference appended.
By linearity, it therefore applies to mixed inputs with arbitrary
memory correlations.

Alice can implement the POVM destructively using the measurement
operators $d\sqrt{w_a}\bra{\beta_{U_a}}$, whose effects are $M_a$.
Write $W_a=W_m(U_a)$, $V_a=V_m(U_a)$, and define the desired output
\begin{equation}
 \omega
 =(\id_X\otimes\mathcal N\otimes\id_Y)(\rho_{XA^sY}).
\label{eq:sector-target-output}
\end{equation}
The measurement on $A^sR$ commutes with the channel acting on the
auxiliary resource block $A^{\prime s}$. Consequently,
\eqref{eq:main-contraction} gives the subnormalized output
for outcome $a$:
\begin{align}
 \widetilde\omega_a
 &=w_a(\id_X\otimes\mathcal N\otimes\id_Y)
 \left[
 (I_X\otimes W_a^\dagger\otimes I_Y)\rho_{XA^sY}
 (I_X\otimes W_a\otimes I_Y)
 \right]\\
 &=w_a(I_X\otimes V_a^\dagger\otimes I_Y)\,
       \omega\,
       (I_X\otimes V_a\otimes I_Y).
 \label{eq:bell-branch}
\end{align}
The second equality follows from covariance
\eqref{eq:sector-covariance}. In particular,
$\Tr\widetilde\omega_a=w_a$, independently of the input.

Alice announces $a$, and Bob applies $V_a$ to $B^s$.
The corrected branch is exactly $w_a\omega$.
Summing over the outcomes gives $\omega$. This proves the general
claim and, with $\mathcal N=\cL_\eta^{\otimes s}$,
\eqref{eq:main-simulation}. If the public outcome is retained
in a register $Z$, the corrected state is
\begin{equation}
 \sum_a w_a\proj a_Z\otimes\omega_{XB^sY}.
\label{eq:sector-public-outcome}
\end{equation}
Thus the outcome probabilities are input independent, and the
public register is uncorrelated with the simulated state
on $XB^sY$. Finally, each share of $\Phi_{s,m}$ contains exactly $m$ photons.
Applying pure loss to one share leaves Alice's photon number
unchanged and gives Bob an output supported on photon numbers
$0,\ldots,m$, with mean $\eta m$. This establishes the
resource-energy properties for pure loss. The Bell measurement acts
only on the finite-dimensional input sector and reference. Thus the
argument also proves the stated generalization when the output
representation is an infinite direct sum of bounded unitaries.
\end{proof}

The construction makes the photon matching explicit: an arbitrary
state in the $m$-photon sector can be teleported perfectly using
$\Phi_{s,m}$, with $m$ photons in each resource share. Applying
loss in the resource preparation converts this into an exact
channel simulation whose finite-energy resource retains the
input photon budget. The resource is supplied to the simulator
as part of the converse construction; its preparation formula
does not represent an additional channel use in the original
protocol.

\subsection{Sector resource REE and the squeezed-state output}
\label{sec:direct}
We evaluate the resource REE exactly, including each received-number
branch. The symmetric-subspace identity
\eqref{eq:coherent-resolution} determines both the conditional entropy
and the relevant action of a separable comparison. We then apply the
same comparison to prove the squeezed-state output identity in
\Cref{cor:tmsv}.

\begin{lemma}[Exact resource REE]\label{lem:resource}
Bob's total received number in $J_{s,m}$ has distribution
\begin{equation}
 p(k|m)=\binom mk\eta^kq^{m-k},\qquad 0\le k\le m.
 \label{eq:binomial}
\end{equation}
Let $J_{s,m,k}$ be the normalized resource state conditioned on Bob
receiving $k$ photons, and write $F_{\eta,s}(m)$ for the REE of
the unconditioned sector resource. Let $L$ denote the lost photon
number, with distribution $\Bin(m,q)$. Then
\begin{align}
 \ER(J_{s,m,k})&=\log_2\frac{d_{s,m}}{d_{s,m-k}},
       \label{eq:conditional-ree}\\
 \ER(J_{s,m})&=F_{\eta,s}(m)
       :=\log_2d_{s,m}-\EE_{L\sim\Bin(m,q)}\log_2d_{s,L}.
       \label{eq:sector-ree}
\end{align}
Here $\EE$ is the expectation over the lost photon number $L=m-k$.
\end{lemma}

\begin{proof}
Fix $k$ and put $\ell=m-k$. For this proof abbreviate
$\cH_r:=\cH_{s,r}$, $d_r:=d_{s,r}$, and $\pi_r:=\pi_{s,r}$.

\smallskip\noindent\emph{1. Conditional loss as a symmetric-tensor split.}
The canonical isometry
$V_{m,k}:\cH_m\longrightarrow\cH_k\otimes\cH_\ell$ groups the first
$k$ and remaining $\ell$ factors of a symmetric tensor. Hence
\begin{equation}
 V_{m,k}\ket{u;m}=\ket{u;k}_B\ket{u;\ell}_E.
 \label{eq:symmetric-split}
\end{equation}
On a condensate, the beam splitter sends each creation operator to
$\sqrt\eta$ times its output-mode operator plus $\sqrt q$ times its
environment-mode operator~\cite{HolevoWerner,Weedbrook}. Its binomial
expansion therefore gives, on the entire $m$-photon sector,
\begin{equation}
 W_\eta\ket\psi
 =\bigoplus_{k=0}^m\sqrt{\binom mk\eta^kq^{m-k}}\,V_{m,k}\ket\psi,
 \qquad \ket\psi\in\cH_m,
 \label{eq:sector-loss-isometry}
\end{equation}
where $W_\eta$ denotes the loss isometry with the environment initially
vacuum. Equality on the spanning condensates proves equality of the
linear maps. Since each $V_{m,k}$ is an isometry and the output sectors
are orthogonal, the probability \eqref{eq:binomial} is independent of
the sector input. The normalized conditional channel is
$\rho\mapsto\Tr_E(V_{m,k}\rho V_{m,k}^\dagger)$.

Let $\ket\Omega=(I_R\otimes V_{m,k})\ket{\Phi_{s,m}}$ and
$J:=\Tr_E\proj\Omega=J_{s,m,k}$. Its reference marginal is
\begin{equation}
 J_R=\pi_m.
 \label{eq:reference-marginal}
\end{equation}
The resolution \eqref{eq:coherent-resolution} gives
\begin{equation}
 \Omega_{BE}=V_{m,k}\pi_m V_{m,k}^\dagger
 =\int\proj{u;k}_B\otimes\proj{u;\ell}_E\,\dd u,
 \qquad \Omega_E=\pi_\ell.
 \label{eq:conditional-environment}
\end{equation}
Purity of $\Omega$ then implies
\begin{equation}
 J=\Pi_{s,m,k}/d_\ell,\qquad S(J)=\log_2d_\ell,
 \label{eq:spectrum-result}
\end{equation}
where $\Pi_{s,m,k}$ projects onto $\supp J$.

\smallskip\noindent\emph{2. A separable comparison from the same resolution.}
Define the unit-trace separable state
\begin{equation}
 \sigma:=\sigma_{s,m,k}
 =\int\proj{\bar u;m}_R\otimes\proj{u;k}_B\,\dd u.
 \label{eq:sigma}
\end{equation}
The Bell-state contraction and \eqref{eq:symmetric-split} give
\begin{equation}
     (\bra{\bar u;m}_R\otimes\bra{u;k}_B\otimes I_E)\ket\Omega
 =\frac{\ket{u;\ell}_E}{\sqrt{d_m}}.
\end{equation}
Consequently,
\begin{align}
 (\sigma\otimes I_E)\ket\Omega
 &=\frac1{\sqrt{d_m}}\int
   \ket{\bar u;m}_R\ket{u;k}_B\ket{u;\ell}_E\,\dd u\\
 &=\frac1{\sqrt{d_m}}(I_R\otimes V_{m,k})
   \int\ket{\bar u;m}_R\ket{u;m}\,\dd u
 =\frac{\ket\Omega}{d_m}.
 \label{eq:comparison-purification}
\end{align}
The last equality is the vectorized form of the identity resolution:
$\int\ket{\bar u;m}\ket{u;m}\,\dd u
=\ket{\Phi_{s,m}}/\sqrt{d_m}$.
Multiplying \eqref{eq:comparison-purification} by $\bra\Omega$ and
tracing out $E$ yields $\sigma J=J/d_m$. Equivalently,
\begin{equation}
 \sigma_{s,m,k}\Pi_{s,m,k}=\Pi_{s,m,k}/d_m.
 \label{eq:eigen-proved}
\end{equation}
This establishes the action on the full support, including the absence
of any component outside it.

\smallskip\noindent\emph{3. Matching the bounds and averaging the flags.}
Equations~\eqref{eq:spectrum-result} and~\eqref{eq:eigen-proved} give
\begin{equation}
 \ER(J)\le D(J\Vert\sigma)
 =\log_2d_m-S(J)=\log_2\frac{d_m}{d_\ell}.
 \label{eq:resource-relative-entropy}
\end{equation}
The conditional-entropy lower bound \eqref{eq:rci-lower}
\cite[Lemma~5]{PlenioVirmaniPapadopoulos} gives the reverse inequality,
since $S(J_R)=\log_2d_m$. This proves \eqref{eq:conditional-ree}.
Finally, Bob's received numbers occupy orthogonal local subspaces,
so the local-flags identity \eqref{eq:flags} averages their REEs with
probabilities \eqref{eq:binomial} and proves \eqref{eq:sector-ree}.
\end{proof}

Only the action of the separable comparison on the support of the
conditional resource enters its relative entropy; the other
eigenvalues are not needed.

\begin{proof}[Proof of \Cref{cor:tmsv}]
Recall that $J_{\eta,N}$ is the output from sending one half of a
two-mode squeezed vacuum of local mean photon number $N$ through the
channel. Write $\tau_x$ for the one-mode thermal state with mean
photon number $x$. We use the sector comparison from \Cref{lem:resource},
starting from
the standard number-basis form of the two-mode squeezed
vacuum~\cite{BraunsteinLoock,Weedbrook}:
\begin{equation}
 \ket{\Psi_N}_{RA}=\sum_{r\ge0}
 \sqrt{\frac{N^r}{(1+N)^{r+1}}}\ket r_R\ket r_A,\qquad
 J_{\eta,N}=(\id\otimes\cL_\eta)(\proj{\Psi_N}).
\end{equation}
For $s$ copies, conditioning the reference on total number $M=m$ gives
the uniform maximally entangled sector input $\Phi_{s,m}$, with probability
\begin{equation}
 w_m=d_{s,m}\frac{N^m}{(1+N)^{s+m}}.
\end{equation}
Conditioning also on Bob's total number $K=k$ gives $J_{s,m,k}$, with
joint probability
\begin{equation}
 p_{m,k}=w_m\binom mk\eta^kq^{m-k}.
 \label{eq:joint-counts}
\end{equation}
Recall the separable conditional comparison $\sigma_{s,m,k}$ from
\eqref{eq:sigma}. Define the unit-trace separable state
\begin{equation}
 \Sigma_s=\sum_{m\ge0}\sum_{k=0}^m p_{m,k}\sigma_{s,m,k}.
 \label{eq:sigma-s}
\end{equation}
The series converges in trace norm, and separability follows from closure
of the separable set. Its blocks are orthogonal in local total-number
subspaces. On the support of $J_{s,m,k}$, \eqref{eq:eigen-proved} gives
eigenvalue $p_{m,k}/d_{s,m}$ for $\Sigma_s$.
These blocks together contain the support of $J_{\eta,N}^{\otimes s}$,
including its coherences between different number sectors.
Indeed, for any positive operator $T$,
$|\langle u|T|v\rangle|^2\le\langle u|T|u\rangle
\langle v|T|v\rangle$. A vector in the kernel of a diagonal
compression is therefore annihilated by $T$ itself. Applying this to
the local number blocks shows that off-diagonal coherences cannot
add support outside the direct sum of their diagonal supports. Thus
the logarithm of $\Sigma_s$ is well defined on the full state support.

Because $\log\Sigma_s$ is block diagonal, those coherences do not affect
the cross entropy. Thus
\begin{align}
 -\Tr(J_{\eta,N}^{\otimes s}\log_2\Sigma_s)
 &=\sum_{m,k}p_{m,k}[-\log_2p_{m,k}+\log_2d_{s,m}]\\
 &=H(M,K)+\EE_M\log_2d_{s,M}\\
 &=s g(N)+H(K|M).
 \label{eq:sigma-crossentropy}
\end{align}
The last equality uses the block decomposition of the thermal reference
marginal: $H(M)+\EE\log_2d_{s,M}=S(\tau_N^{\otimes s})=sg(N)$.
The loss environment marginal is $\tau_{qN}^{\otimes s}$, and the global
reference--output--environment state is pure. Hence
$S(J_{\eta,N}^{\otimes s})=sg(qN)$ and
\begin{equation}
 D(J_{\eta,N}^{\otimes s}\Vert\Sigma_s)
 =s c_\eta(N)+\EE_M H(\Bin(M,\eta))
 \le s c_\eta(N)+g(s\eta N).
 \label{eq:direct-upper}
\end{equation}
All entropies in this calculation are finite. In particular,
$H(K|M)\le H(K)\le g(s\eta N)$ and the reference entropy is $sg(N)$.
This proves $\ERi(J_{\eta,N})\le c_\eta(N)$.
For the lower bound in this infinite-dimensional setting, every separable
$\sigma_{RB}$ obeys $\sigma_{RB}\le\sigma_R\otimes I_B$
\cite{Reduction}. Operator monotonicity of the logarithm~\cite{Bhatia}
therefore gives
$D(J_{\eta,N}^{\otimes s}\Vert\sigma_{RB})\ge sg(N)-sg(qN)$,
as in the proof of Ref.~\cite[Theorem~4]{PlenioVirmaniPapadopoulos}.
Both entropies are finite, so the subtraction is well defined;
comparisons with infinite relative entropy satisfy the bound trivially.
Taking the infimum yields
$\ER(J_{\eta,N}^{\otimes s})\ge s c_\eta(N)$, completing
\eqref{eq:tmsv-ree}.
\end{proof}

\subsection{Replacing the sector label by the energy constraint}
We now average the sector cost so that the bound depends only on the
mean input energy. Let $M$ denote the total input photon number in a
block of $s$ modes, and define
\begin{equation}
 x:=\frac{\EE M}{s}\ge0.
 \label{eq:block-mean-definition}
\end{equation}
Thus $x$ is the mean transmitted photon number per mode and $sx$ is
the total block mean. Recall that $q=1-\eta$,
$c_\eta(x)=g(x)-g(qx)$, and
$F_{\eta,s}(m)=\ER(J_{s,m})$ is the fixed-sector resource cost from
\Cref{lem:resource}.

\begin{lemma}[Average resource cost]\label{lem:energy}
For every $s,m$,
\begin{equation}
 F_{\eta,s}(m)\le s c_\eta(m/s)+H(\Bin(m,q)).
 \label{eq:resource-upper}
\end{equation}
Consequently, for a nonnegative integer-valued $M$ with finite
mean $\EE M=sx$, where $x$ is defined in \eqref{eq:block-mean-definition},
\begin{equation}
 \EE F_{\eta,s}(M)\le s c_\eta(x)+g(sqx).
 \label{eq:main-average-cost}
\end{equation}
\end{lemma}

\begin{proof}
Take $y=m/s>0$ and let $\tau_y$ be the geometric thermal state of mean
$y$, with entropy $g(y)$~\cite{HolevoBook,Weedbrook}. Its product is constant on the $m$-photon sector, so
\begin{equation}
 D(\pi_{s,m}\Vert\tau_y^{\otimes s})=s g(y)-\log_2d_{s,m}.
 \label{eq:main-thermal-input}
\end{equation}
Send both states through $\cL_q^{\otimes s}$, the channel to the lost
photons. Their images are, respectively,
$\sum_\ell\Pr(L=\ell)\pi_{s,\ell}$ with $L\sim\Bin(m,q)$,
and $\tau_{qy}^{\otimes s}$. The output relative entropy is
\begin{equation}
 s g(qy)-H(L)-\EE\log_2d_{s,L}.
 \label{eq:main-thermal-output}
\end{equation}
Data processing says \eqref{eq:main-thermal-input} is at least
\eqref{eq:main-thermal-output}; rearranging proves
\eqref{eq:resource-upper}. For $m=0$ all terms vanish. Thermal states
are reference states for this comparison; no input optimization is used.

By \Cref{lem:scalar}, $c_\eta$ is increasing and concave, and an
integer-valued count of mean $z$ has entropy at most $g(z)$.
Define $L|M=m\sim\Bin(m,q)$.
Then $\EE L=sqx$, and
\begin{equation}
 \EE F_{\eta,s}(M)
 \le s\EE c_\eta(M/s)+H(L|M)
 \le s c_\eta(x)+H(L)
 \le s c_\eta(x)+g(sqx).
\label{eq:average-resource-chain}
\end{equation}
This proves \eqref{eq:main-average-cost}.
\end{proof}

\subsection{Combining the tools for a general input}\label{sec:truncation}
We extend the sector estimate to arbitrary multimode inputs, including
coherence between different total photon numbers and correlations with
the parties' memories. The resulting bound uses only the mean photon
number per mode $x$ and two scalar count-entropy corrections.

\begin{lemma}[Block entanglement gain]\label{prop:block}
Let $\rho_{XA^sY}$ have $\ER(XA^s:Y)_\rho<\infty$ and
$\Tr(\hat N_A\rho)=sx<\infty$, where
$\hat N_A=\sum_{j=1}^s\hat n_{A_j}$. For
$\omega=(\id_{XY}\otimes\cL_\eta^{\otimes s})(\rho)$,
\begin{equation}
 \ER(X:B^sY)_\omega
 \le\ER(XA^s:Y)_\rho+s c_\eta(x)+g(s\eta x)+g(sqx).
 \label{eq:block}
\end{equation}
The memories $X,Y$ need not be finite-dimensional or unentangled.
\end{lemma}

\begin{proof}
First suppose the memories are finite-dimensional and the optical input
has a finite total-number cutoff. Measure the total input number $M$,
obtaining probabilities $p_m$ and conditional states $\rho_m$.
Let $\omega_m=\cL_\eta^{\otimes s}(\rho_m)$ and let $\cP_A,\cP_B$
denote total-number pinchings at input and output. Phase covariance gives
\begin{equation}
 \cP_B\omega=\cL_\eta^{\otimes s}(\cP_A\rho)
                   =\sum_m p_m\omega_m.
 \label{eq:main-phase}
\end{equation}
Using the input cut $XA^s:Y$ and the output cut $X:B^sY$, we obtain
\begin{align}
 \ER(\cP_B\omega)
 &\le\sum_m p_m\ER(\omega_m)
                   &&\text{by convexity}\\
 &\le\sum_m p_m\bigl[\ER(\rho_m)+F_{\eta,s}(m)\bigr]
                   &&\text{by \Cref{prop:sector,lem:resource}}\\
 &\le\ER(\rho)+s c_\eta(x)+g(sqx)
                   &&\text{by \eqref{eq:strong} and \Cref{lem:energy}}.
 \label{eq:main-block-chain}
\end{align}
The second line uses LOCC monotonicity and subadditivity on the conditional
product $\rho_m\otimes J_{s,m}$. It requires no independence between
$M$ and the memories.

Let $K$ be Bob's total received number. Its mean is $s\eta x$.
\Cref{lem:pinching} and the count-entropy bound yield
$\ER(\omega)\le\ER(\cP_B\omega)+H(K)
\le\ER(\cP_B\omega)+g(s\eta x)$, proving \eqref{eq:block} in
finite dimensions.

\smallskip\noindent\emph{Extension to arbitrary memories and mean-energy inputs.}
Let now $\rho_{XA^sY}$ have arbitrary separable memory spaces, finite
$\ER(XA^s:Y)_\rho$, and finite total transmitted mean $sx$.
Choose increasing finite-rank projectors $R_{X,t}$ and $R_{Y,t}$ converging
strongly to the respective identities. Choose unit vectors $\ket{0_X}$,
$\ket{0_Y}$ in their ranges. On either memory $Z=X,Y$, define the local
channel
\begin{equation}
 \mathcal C_{Z,t}(T)=R_{Z,t}TR_{Z,t}
       +\Tr[(I-R_{Z,t})T]\proj{0_Z}.
\end{equation}
On the optical input use $Q_t=\sum_{m=0}^tP^A_m$ and
\begin{equation}
 \mathcal C_{A,t}(T)=Q_tTQ_t+\Tr[(I-Q_t)T]\proj{0^{\otimes s}}.
 \label{eq:vac-truncate}
\end{equation}
These are trace-preserving local channels. The output state
\begin{equation}
 \rho_t=(\mathcal C_{X,t}\otimes\mathcal C_{A,t}\otimes
          \mathcal C_{Y,t})(\rho)
\end{equation}
has finite-dimensional support and converges to $\rho$ in trace norm.
This follows from $R_{Z,t},Q_t\to I$ strongly and trace-class approximation;
equivalently, each discarded probability tends to zero and the off-diagonal
blocks vanish by the gentle measurement
estimate~\cite{WinterGentle,Watrous}.

Because $Q_t$ commutes with $\hat N_A$ and the replacement state is vacuum,
\begin{equation}
 sx_t:=\Tr(\hat N_A\rho_t)
 =\Tr(\hat N_A Q_t\rho_A Q_t)\le sx,
 \qquad \ER(\rho_t)\le\ER(\rho).
 \label{eq:truncate-bounds}
\end{equation}
The memory channels do not change the optical marginal. Pure loss cannot
increase the total photon cutoff, so
$\omega_t=(\id\otimes\cL_\eta^{\otimes s})(\rho_t)$ also has finite
support. It converges in trace norm to $\omega$ by channel contractivity.
The finite-support inequality gives
\begin{align}
 \ER(\omega_t)
 &\le\ER(\rho_t)+s c_\eta(x_t)+g(s\eta x_t)+g(sqx_t)\\
 &\le\ER(\rho)+s c_\eta(x)+g(s\eta x)+g(sqx).
\end{align}
The last inequality uses monotonicity of $g$ and $c_\eta$.
REE is lower semicontinuous on the full infinite-dimensional state space
with the closed separable set~\cite[Corollary~20]{LamiShirokov}. Consequently
\begin{equation}
 \ER(\omega)\le\liminf_{t\to\infty}\ER(\omega_t)
 \le\ER(\rho)+s c_\eta(x)+g(s\eta x)+g(sqx).
 \label{eq:infinite-block}
\end{equation}
This proves \Cref{prop:block} without a memory entropy or energy constraint.
Only lower semicontinuity of the \emph{unregularized} REE was invoked.
There is no exchange of the truncation limit with an REE regularization limit.
\end{proof}

The channel-dependent inputs to this argument are phase covariance,
the averaged sector-resource cost, and the mean output photon number.
Appendix~\ref{sec:amp-block} reuses this argument for amplification,
adding the output truncation required when finite input support
does not imply finite output support.

Only two scalar entropies remain: $g(sqx)$ bounds uncertainty in the loss
count, and $g(s\eta x)$ bounds the cost of output number pinching. Both
are $O(\log s)$ for fixed $x$. Pinching every mode separately would
generally produce an extensive overhead and would not serve this purpose.
\par\smallskip\noindent\emph{Consequence: regularized one-use gain.}
The block estimate also gives a single-use bound that includes
correlations with both parties' memories. Such entanglement-gain
estimates underlie earlier adaptive relative-entropy
bounds~\cite{ChristandlMullerHermes}.
For a fixed input $\rho_{XAY}$ with finite $\ER(XA:Y)_\rho$ and
finite mean transmitted number $x$, let
$\omega=(\id_{XY}\otimes\cL_\eta)(\rho)$. Then
\begin{equation}
 \ERi(X:BY)_\omega\le\ERi(XA:Y)_\rho+c_\eta(x).
 \label{eq:amort}
\end{equation}
Indeed, applying \Cref{prop:block} to $s$ copies of this fixed input gives
\begin{equation}
 \frac{\ER(\omega^{\otimes s})}{s}
 \le\frac{\ER(\rho^{\otimes s})}{s}+c_\eta(x)
     +\frac{g(s\eta x)+g(sqx)}{s}.
\label{eq:regularized-gain-copies}
\end{equation}
The output REE is finite by \Cref{prop:block}. Taking $s\to\infty$
and using $g(sy)/s\to0$ proves \eqref{eq:amort}. These copies define
the regularization; they are not successive uses of the channel.
For the capacity proofs, the next subsection applies the finite-block
estimate directly to replicas of the entire adaptive protocol.

\subsection{Replicas of the entire adaptive protocol}
The large block needed to suppress the scalar overhead is built from
independent copies of the entire protocol. At a fixed channel round,
their transmitted modes can be grouped without regrouping successive
inputs within any run; see Fig.~\ref{fig:adaptive-replicas}.

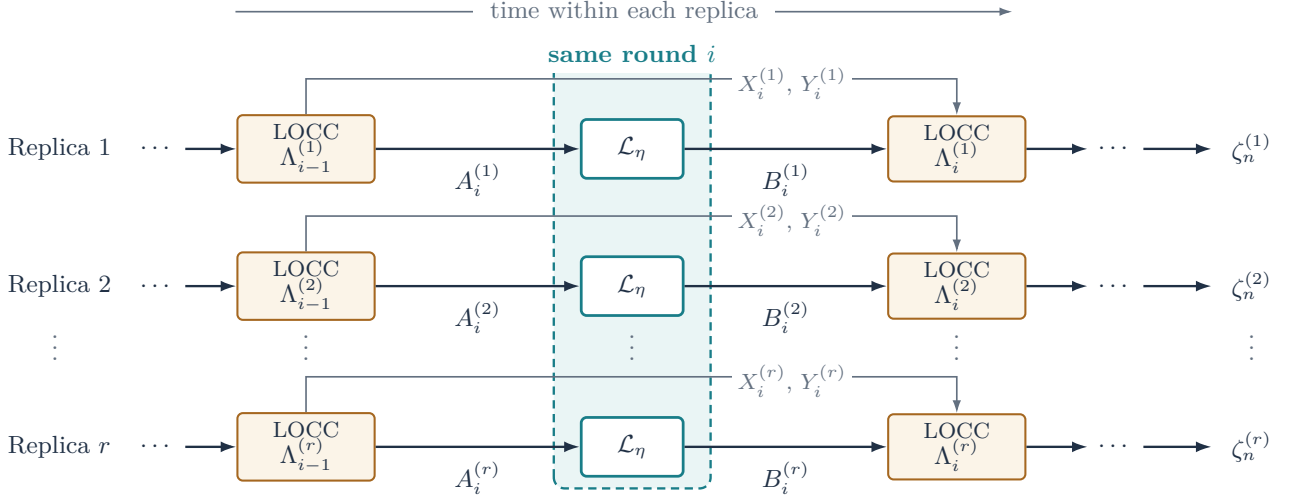
\begin{figure}[!htbp]
\centering
\begingroup
\definecolor{ink}{HTML}{203045}
\definecolor{muted}{HTML}{637080}
\definecolor{teal}{HTML}{197D88}
\definecolor{tealfill}{HTML}{EAF5F5}
\definecolor{amber}{HTML}{A56720}
\definecolor{amberfill}{HTML}{FBF4E8}
\definecolor{rulegray}{HTML}{D1D9E1}
\definecolor{palefill}{HTML}{F5F7FA}

\tikzset{
  every node/.style={text=ink},
  panel/.style={font=\bfseries,anchor=west},
  note/.style={font=\small,text=muted,align=center},
  arrow/.style={draw=ink,line width=.85pt,
    -{Latex[length=2mm,width=1.4mm]}},
  memory/.style={draw=muted,line width=.6pt,
    -{Latex[length=1.65mm,width=1.1mm]}},
  locc/.style={draw=amber,fill=amberfill,line width=.75pt,
    rounded corners=2pt,minimum width=1.75cm,minimum height=.78cm,
    align=center},
  channel/.style={draw=teal,fill=white,line width=1pt,
    rounded corners=2pt,minimum width=1.3cm,minimum height=.75cm},
  result/.style={draw=teal,fill=tealfill,line width=.7pt,
    rounded corners=3pt},
  proofbox/.style={draw=rulegray,fill=palefill,line width=.6pt,
    rounded corners=3pt}
}

\resizebox{0.98\linewidth}{!}{%
\begin{tikzpicture}[x=1cm,y=1cm]

\draw[memory] (3.12,14.83)--(13.03,14.83);
\node[note,fill=white,inner xsep=6pt] at (8.08,14.83)
  {time within each replica};

\fill[tealfill,rounded corners=4pt] (7.18,8.75) rectangle (9.18,14.3);
\draw[teal,densely dashed,line width=.85pt,rounded corners=4pt]
  (7.18,8.75) rectangle (9.18,14.3);
\node[font=\small\bfseries,text=teal,fill=white,inner sep=4pt]
  at (8.18,14.3) {same round $i$};

\foreach \rep/\yy in {1/13.08,2/11.33,r/9.28} {
  \node[font=\small,anchor=west] at (.1,\yy) {Replica $\rep$};
  \node at (2.12,\yy) {$\cdots$};
  \node[locc] (pre\rep) at (4.02,\yy)
    {\footnotesize LOCC\\[-1pt] $\Lambda_{i-1}^{(\rep)}$};
  \node[channel] (ch\rep) at (8.18,\yy) {$\mathcal L_\eta$};
  \node[locc] (post\rep) at (12.32,\yy)
    {\footnotesize LOCC\\[-1pt] $\Lambda_i^{(\rep)}$};
  \node at (14.34,\yy) {$\cdots$};
  \node at (16.08,\yy) {$\zeta_n^{(\rep)}$};
  \draw[arrow] (2.48,\yy)--(pre\rep.west);
  \draw[arrow] (pre\rep.east)--(ch\rep.west)
    node[midway,below=2pt] {$A_i^{(\rep)}$};
  \draw[arrow] (ch\rep.east)--(post\rep.west)
    node[midway,below=2pt] {$B_i^{(\rep)}$};
  \draw[arrow] (post\rep.east)--(13.99,\yy);
  \draw[arrow] (14.7,\yy)--(15.57,\yy);
  \draw[memory] (pre\rep.north) -- ++(0,.45)
    -| (post\rep.north);
  \node[font=\footnotesize,text=muted,fill=white,inner sep=2pt]
    at (10.22,{\yy+.84}) {$X_i^{(\rep)},\,Y_i^{(\rep)}$};
}
\foreach \xx in {.82,4.02,8.18,12.32,16.08} {
  \node[text=muted] at (\xx,10.65) {$\vdots$};
}
\end{tikzpicture}
}
\endgroup
\caption{Replica grouping for an adaptive protocol. Each row is an
independent copy of the same fixed $n$-use protocol. The highlighted
column groups the $i$th transmission across $r$ replicas and forms the
block of size $s=r$ used in \Cref{prop:block}. Gray wires retain the
local memories $X_i^{(j)},Y_i^{(j)}$ and classical records. The LOCC
order and feedback within every row are unchanged. Applying the block
inequality column by column and using LOCC monotonicity gives
\eqref{eq:adaptive-replicas}. The limit $r\to\infty$ is taken with
$n$ and the protocol fixed, before the capacity limit $n\to\infty$.}
\label{fig:adaptive-replicas}
\end{figure}

\begin{lemma}[Finite-replica adaptive bound]\label{lem:adaptive}
Fix an $n$-use adaptive protocol with unconditional energies $x_i$
and final state $\zeta_n$. For every positive integer $r$,
\begin{equation}
 \ER(\zeta_n^{\otimes r})
 \le r\sum_{i=1}^n c_\eta(x_i)
       +\sum_{i=1}^n[g(rqx_i)+g(r\eta x_i)].
 \label{eq:adaptive-replicas}
\end{equation}
The state includes any additional local registers retained during the protocol.
\end{lemma}

\begin{proof}
All intermediate unregularized REEs are finite. The initial state
is separable, every $x_i$ is finite, and \Cref{prop:block} with
$s=1$ adds a finite amount at each transmission. LOCC monotonicity
then proves finiteness inductively.

Run $r$ independent copies of this fixed protocol, each following
its own outcomes and feedback. Let $\rho^{(i)}$ and $\omega^{(i)}$
be the states immediately before and after call $i$ in one run,
including all memories and stored records. Group the $i$th
transmissions from the $r$ runs. Their total mean number is $rx_i$.
\Cref{prop:block} gives
\begin{equation}
 \ER((\omega^{(i)})^{\otimes r})
 \le \ER((\rho^{(i)})^{\otimes r})
       +r c_\eta(x_i)+g(rqx_i)+g(r\eta x_i).
 \label{eq:adaptive-stage}
\end{equation}
Only the transmitted modes change sides of the Alice--Bob cut.
An adaptive local operation is a fixed channel on memories and
records, controlled by the stored history. If $\Lambda_i$ is the
between-call LOCC in one run, its replicated version
$\Lambda_i^{\otimes r}$ is LOCC. Hence
\begin{equation}
 \ER((\rho^{(i+1)})^{\otimes r})
 \le \ER((\omega^{(i)})^{\otimes r}).
\label{eq:replicated-locc-monotonicity}
\end{equation}
Telescope with \eqref{eq:adaptive-stage}, starting from zero REE
and including the final LOCC, to obtain
\eqref{eq:adaptive-replicas}.
\end{proof}

Independence is used between replicas, not between successive
inputs within a run. Each run preserves its causal order. The
pinching in \Cref{prop:block} is only a step in proving an
inequality for its actual, unpinched output.

Fix the protocol before taking $r\to\infty$. Since $g(ry)/r\to0$
for every finite $y\ge0$, \eqref{eq:adaptive-replicas} implies
\begin{equation}
 \ERi(\zeta_n)\le\sum_{i=1}^n c_\eta(x_i).
 \label{eq:adaptive-regularized}
\end{equation}
Finiteness and subadditivity ensure that this regularization limit
exists. No additivity of REE is assumed. Concavity and monotonicity
of $c_\eta$, together with $\sum_i x_i\le nN$, now give the common adaptive REE bound
\begin{equation}
 \boxed{\ERi(\zeta_n)\le\sum_{i=1}^{n}c_\eta(x_i)
 \le n c_\eta\!\left(\frac1n\sum_i x_i\right)
 \le n c_\eta(N).}
 \label{eq:telescope}
\end{equation}
The replica limit is taken for each fixed $n$ and fixed protocol.
Equation~\eqref{eq:telescope} holds before any choice of a fidelity or
key-security test; local processing of the output preserves this bound.

The replica argument itself is channel independent. Suppose a
memoryless channel $\mathcal N$ obeys, for the same arbitrary-memory
states and every block size $s$,
\begin{equation}
 E_R(\omega)\le E_R(\rho)+s f(x)+h_s(x),
 \qquad \lim_{s\to\infty}\frac{h_s(x)}s=0,
 \label{eq:generic-block-hypothesis}
\end{equation}
where $f$ is finite, increasing, and concave, and $h_s(x)$ is finite
and nonnegative at each finite $x$. The preceding proof, with
$c_\eta$ and its overhead replaced by $f$ and $h_s$, gives
\begin{equation}
 E_R^\infty(\zeta_n)\le\sum_{i=1}^n f(x_i)\le n f(N).
 \label{eq:generic-adaptive-conclusion}
\end{equation}
The $s=1$ bound first ensures finite intermediate REEs, and the
replica limit is again taken at fixed protocol. This form will be
used for the amplifier extension.

\section{Two-way quantum capacity}\label{sec:quantum}
We apply the common adaptive bound \eqref{eq:telescope} to the
entanglement-generation formulation of quantum communication established
in Section~\ref{sec:results}. We first identify the physical systems
and the target state, then prove a finite-error converse and a matching
achievable rate. Private states are only needed in the next section.

Previously, an achievable rate was $c_\eta(N)$ and an available
upper bound was~\cite{DSW,PLOB}
\begin{equation}
 \min\!\left\{
 g\!\left(\frac{1+\eta}{2}N\right)
 -g\!\left(\frac q2N\right),\,-\log_2q\right\}.
 \label{eq:prior-upper}
\end{equation}
The matching converse below closes this finite-energy gap.

\subsection{Systems and adaptive quantum protocol}\label{sec:quantum-model}
This subsection specifies which systems stay in the laboratories,
which system crosses the channel, and which final state is tested for
successful quantum communication. Let $X_i$ be Alice's retained memory, $A_i$ her next transmitted optical
mode, and $Y_i$ Bob's retained memory just before use $i$. Memories
include local ancillas, stored quantum outputs, and classical records.
After the initial LOCC preparation $\Lambda_0$, the first input state
is separable across $X_1A_1:Y_1$. Later input states need not be
separable: earlier transmissions can have generated shared entanglement.
For the unconditional state $\rho^{(i)}_{X_iA_iY_i}$, the $i$th call gives
\begin{equation}
 \omega^{(i)}_{X_iB_iY_i}
 =\bigl(\id_{X_i}\otimes\cL_\eta^{A_i\to B_i}\otimes\id_{Y_i}\bigr)
       (\rho^{(i)}_{X_iA_iY_i}).
 \label{eq:adaptive-round-state}
\end{equation}
Only $A_i$ passes through the optical channel. Bob receives $B_i$ and
may process it jointly with $Y_i$; Alice keeps $X_i$. The Alice--Bob
bipartition therefore changes from $X_iA_i:Y_i$ to $X_i:B_iY_i$.
For $i<n$, an LOCC map $\Lambda_i$, controlled by the previous public
outcomes, produces the next input and memories:
\begin{equation}
 \rho^{(i+1)}_{X_{i+1}A_{i+1}Y_{i+1}}
       =\Lambda_i(\omega^{(i)}_{X_iB_iY_i}).
 \label{eq:adaptive-round-locc}
\end{equation}
Neither the transmitted mode nor either memory is assumed independent
of the others. The maps include averaging over all public branches,
while the records remain available for subsequent feedback~\cite{PLOB}.

The final LOCC $\Lambda_n$ produces a state
$\zeta_n$ on the retained laboratories $X_{\rm f}:Y_{\rm f}$.
Write $X_{\rm f}=\bar A F_A$ and $Y_{\rm f}=\bar B F_B$, where
$\bar A,\bar B$ are the target registers of dimension $d_n$ and
$F_A,F_B$ contain unused local systems. Entanglement distribution tests
\begin{equation}
 \rho_n^{\bar A\bar B}=\Tr_{F_AF_B}\zeta_n,\qquad
 f_n=\Tr(\Phi_{d_n}\rho_n),\qquad
 \varepsilon_n=1-f_n,\qquad R_n=\frac{\log_2d_n}{n}.
 \label{eq:quantum-output-registers}
\end{equation}
For quantum transmission, $\bar A$ is instead the untouched reference
of a locally prepared maximally entangled message--reference pair,
and $\bar B$ is Bob's decoded message; $f_n$ is then the entanglement
fidelity. We use the equivalent generation task for the converse.

The REE in \eqref{eq:telescope} concerns the joint state of Alice's and
Bob's retained laboratories, with the channel environments traced out.
It is entanglement across $X_{\rm f}:Y_{\rm f}$, rather than a sum of
quantities of the two separate memory marginals. Discarding $F_A,F_B$
is local, so the same upper bound holds for $\ERi(\rho_n)$.
The binary entropy $h_2$ used below is defined in
\eqref{eq:entropy-definitions}. Figure~\ref{fig:adaptive-quantum-protocol} displays the channel/LOCC
interleaving and the extraction of the final quantum target from the
retained local memories.

\begin{figure}[!htbp]
\centering

\begingroup
\definecolor{protink}{HTML}{203045}
\definecolor{protblue}{HTML}{295D83}
\definecolor{protbluefill}{HTML}{EEF5FA}
\definecolor{protamber}{HTML}{A46418}
\definecolor{protamberfill}{HTML}{FBF5E8}
\definecolor{protpurple}{HTML}{70518C}
\definecolor{protpurplefill}{HTML}{F3EDF8}
\definecolor{protteal}{HTML}{197D88}
\definecolor{protgray}{HTML}{C7D2DC}
\tikzset{
  every node/.style={text=protink,inner sep=2pt},
  pq/.style={draw=protblue,line width=.8pt,-{Latex[length=2mm,width=1.4mm]}},
  pqline/.style={draw=protblue,line width=.8pt},
  pc/.style={draw=protamber,densely dashed,line width=.85pt,
    -{Latex[length=2mm,width=1.4mm]}},
  pcc/.style={draw=protamber,densely dashed,line width=.85pt,
    {Latex[length=1.8mm]}-{Latex[length=1.8mm]}},
  pk/.style={draw=protteal,double,double distance=1.1pt,line width=.4pt,
    -{Latex[length=2mm,width=1.4mm]}},
  pe/.style={draw=protpurple,line width=.8pt,
    -{Latex[length=2mm,width=1.4mm]}},
  pgate/.style={draw=protblue,fill=protbluefill,line width=.8pt,rounded corners=2pt},
  plocal/.style={draw=protamber,fill=protamberfill,line width=.8pt,rounded corners=2pt},
  pstate/.style={draw=protpurple,fill=protpurplefill,line width=.8pt,rounded corners=2pt},
  ppanel/.style={font=\small\bfseries,anchor=west},
  powner/.style={font=\small\bfseries,anchor=west,text=protblue},
  plabel/.style={font=\small},
  psmall/.style={font=\footnotesize}
}

\resizebox{0.99\linewidth}{!}{%
\begin{tikzpicture}[x=1cm,y=.65cm,font=\small]
\path[use as bounding box] (0,-.4) rectangle (17,10.5);
\node[ppanel] at (0,10.2) {(a) Adaptive channel use and memory update ($i<n$)};
\node[powner] at (0,9.6) {Alice};
\node[powner] at (0,6.4) {Bob};
\draw[draw=protgray,dashed] (0,6.9)--(17,6.9);

\foreach \x in {1.5,11.5}{
  \draw[plocal] (\x,7.5) rectangle ++(1.7,1.65);
  \node at (\x+.85,8.32) {Local};
  \draw[plocal] (\x,4.95) rectangle ++(1.7,1.2);
  \node at (\x+.85,5.55) {Local};
  \draw[pcc] (\x+.85,6.28)--(\x+.85,7.38);
}
\node[plabel] at (2.35,9.58) {$\Lambda_{i-1}$};
\node[plabel] at (12.35,9.58) {$\Lambda_i$};
\node[psmall,anchor=west,text=protamber] at (2.65,6.76) {LOCC};
\node[psmall,anchor=west,text=protamber] at (12.65,6.76) {LOCC};
\node at (.55,8.72) {$\cdots$};
\node at (.55,5.4) {$\cdots$};
\draw[pq] (.82,8.72)--(1.5,8.72);
\draw[pq] (.82,5.4)--(1.5,5.4);

\draw[pq] (3.2,8.72)--(11.5,8.72);
\node[plabel,anchor=south] at (7.15,8.76) {$X_i$};
\draw[pgate] (6.2,6.3) rectangle (8.2,8.3);
\node at (7.2,7.3) {$\cL_\eta$};
\draw[pq] (3.2,7.95)--(6.2,7.95);
\node[plabel,anchor=south] at (4.65,7.99) {$A_i$};
\draw[pq] (8.2,6.65)--(9.55,6.65)--(9.55,5.85)--(11.5,5.85);
\node[plabel,anchor=south] at (10.5,5.88) {$B_i$};
\draw[pq] (3.2,5.3)--(11.5,5.3);
\node[plabel,anchor=south] at (7.15,5.34) {$Y_i$};
\draw[pq] (13.2,8.72)--(16.4,8.72);
\node[plabel,anchor=south] at (14.85,8.76) {$X_{i+1}$};
\draw[pq] (13.2,7.95)--(16.4,7.95);
\node[plabel,anchor=south] at (14.85,7.99) {$A_{i+1}$};
\draw[pq] (13.2,5.3)--(16.4,5.3);
\node[plabel,anchor=south] at (14.85,5.34) {$Y_{i+1}$};
\node at (16.75,8.3) {$\cdots$};
\node at (16.75,5.3) {$\cdots$};

\node[ppanel] at (0,3.9) {(b) Final LOCC and quantum target};
\draw[plocal] (3.0,.0) rectangle (5.0,2.9);
\node[align=center] at (4.0,1.45) {LOCC\\[2pt]$\Lambda_n$};
\node[plabel,anchor=south] at (1.4,2.4) {$X_n$};
\node[plabel,anchor=south] at (1.4,.65) {$B_nY_n$};
\draw[pq] (.25,2.3)--(3,2.3);
\draw[pq] (.25,.55)--(3,.55);
\draw[pstate,dashed] (5.25,.05) rectangle (8.92,3.05);
\draw[pq] (5,2.3)--(9.15,2.3);
\draw[pq] (5,.55)--(9.15,.55);
\node[plabel,anchor=south] at (7.1,2.36) {$X_{\rm f}=\bar A F_A$};
\node[plabel,anchor=south] at (7.1,.61) {$Y_{\rm f}=\bar B F_B$};
\node[psmall,text=protpurple] at (7.1,1.45) {$\zeta_n^{X_{\rm f}Y_{\rm f}}$};
\draw[pgate] (9.15,1.8) rectangle (11.05,2.8);
\draw[pgate] (9.15,.05) rectangle (11.05,1.05);
\node at (10.1,2.3) {$\Tr_{F_A}$};
\node at (10.1,.55) {$\Tr_{F_B}$};
\draw[pq] (11.05,2.3)--(14.35,2.3);
\draw[pq] (11.05,.55)--(14.35,.55);
\node[plabel,anchor=south] at (12.6,2.34) {$\bar A$};
\node[plabel,anchor=south] at (12.6,.59) {$\bar B$};
\draw[draw=protpurple,decorate,decoration={brace,amplitude=4pt}]
  (14.55,2.55)--(14.55,.3);
\node[align=center,text=protpurple] at (15.85,1.45)
  {$\rho_n^{\bar A\bar B}$\\[4pt]target $\Phi_{d_n}$};
\end{tikzpicture}%
}
\endgroup
\caption{Adaptive quantum transmission and entanglement distribution.
(a) The local operations within $\Lambda_{i-1}$ prepare Alice's mode
$A_i$ and retain memories $X_i,Y_i$. Only $A_i$ crosses the channel;
the next LOCC $\Lambda_i$ processes $X_i:B_iY_i$ and prepares the next
round. Dashed double arrows denote two-way public classical communication.
At the first round, $\Lambda_0$ is an initially unentangled preparation.
(b) After the last channel use, $\Lambda_n$ produces the retained state
$\zeta_n$. Local partial traces leave the two target registers.
For entanglement distribution, $\bar A$ is Alice's generated half of
the pair; for quantum transmission it is the untouched reference,
on which the protocol operations act as the identity, while $\bar B$
is Bob's decoded message.}
\label{fig:adaptive-quantum-protocol}
\end{figure}
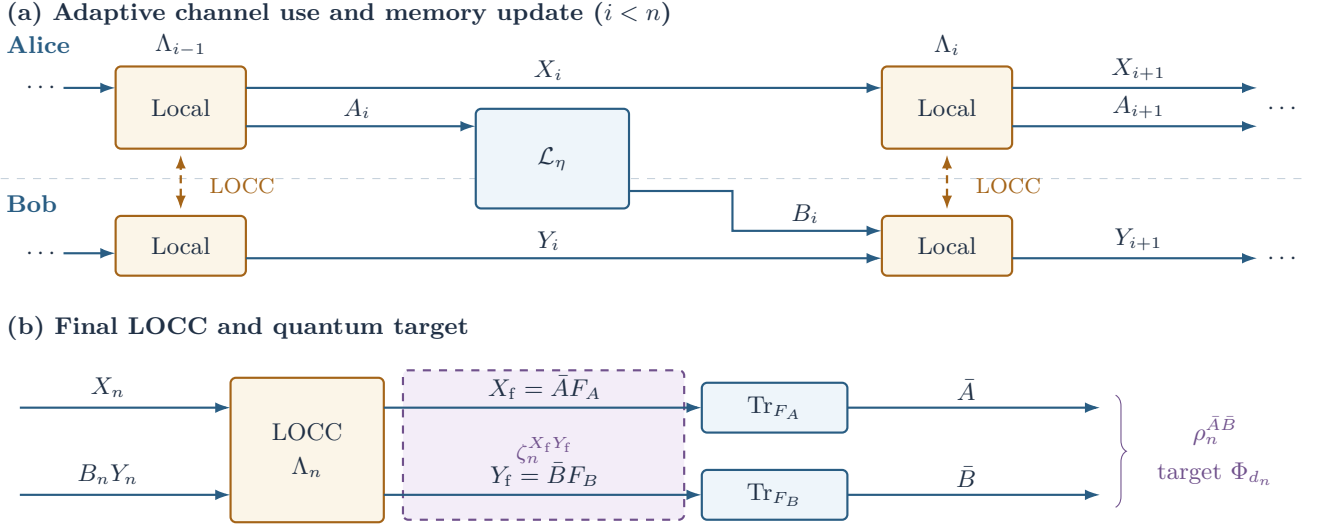

\subsection{From entanglement fidelity to regularized REE}
We show that an output close to a maximally entangled target must
contain enough regularized REE. This converts the common entanglement
upper bound into a limit on the number of qubits generated.

\begin{lemma}[Maximally entangled-state test]\label{lem:quantum-test}
Let $\rho$ be a state on two registers of local dimension $d$, and let
$0\le\varepsilon<1$. If $\Tr(\Phi_d\rho)\ge1-\varepsilon$, then
\begin{equation}
 \ERi(\rho)\ge(1-\varepsilon)\log_2d-h_2(\varepsilon).
 \label{eq:quantum-test}
\end{equation}
\end{lemma}
\begin{proof}
The case $d=1$ is immediate. A normalized product vector has squared
overlap at most $1/d$ with $|\Phi_d\rangle$, by the Cauchy--Schwarz inequality.
Convexity therefore gives $\Tr(\Phi_d\sigma)\le1/d$ for every separable
$\sigma$. For $r$ copies, $\Phi_d^{\otimes r}$ is a maximally entangled
projector of local dimension $d^r$. Every state separable across the
aggregate Alice--Bob cut consequently passes these $r$ tests with
probability at most $d^{-r}$, even if it correlates the copies.

Fix $0<\delta<1-\varepsilon$ and test $t$ copies independently with
$\{\Phi_d,I-\Phi_d\}$. Accept if at least
$r=\lfloor(1-\varepsilon-\delta)t\rfloor$ tests pass. Let $a_t$ denote
acceptance on $\rho^{\otimes t}$. The weak law gives $a_t\to1$.
For a globally separable comparison, any specified subset of $r$ tests
passes with probability at most $d^{-r}$; taking a union bound over
the $\binom tr$ subsets gives acceptance at most $\binom tr d^{-r}$.
Binary data processing for relative entropy~\cite{Lindblad,Watrous}
therefore yields
\begin{equation}
 \ER(\rho^{\otimes t})
 \ge a_t\left[r\log_2d-\log_2\binom tr\right]-h_2(a_t).
\label{eq:quantum-threshold-ree}
\end{equation}
Divide by $t$ and let $t\to\infty$. The binomial-coefficient
asymptotics~\cite{CoverThomas} give
$t^{-1}\log_2\binom tr\to h_2(1-\varepsilon-\delta)$.
Finally let $\delta\downarrow0$ and use
$h_2(1-\varepsilon)=h_2(\varepsilon)$.
\end{proof}

The threshold test is what makes the estimate a bound on
$\ERi$, rather than only on single-copy REE. It controls all
separable comparison states across the full tensor-product cut.

\subsection{Finite-error quantum weak converse}\label{sec:quantum-converse}
We combine the adaptive output bound with the entanglement-fidelity
test. The result limits the achievable quantum rate while retaining
the error of each finite protocol.

\begin{lemma}[Quantum finite-error bound]\label{lem:quantum-finite-error}
For an $n$-use protocol satisfying \eqref{eq:energy}, let $\rho_n$ be
its output on target registers of local dimension $d$.
If $F(\rho_n,\Phi_d)\ge1-\varepsilon$ with $0\le\varepsilon<1$, then
\begin{equation}
 (1-\varepsilon)\log_2d
 \le\sum_{i=1}^n c_\eta(x_i)+h_2(\varepsilon)
 \le n c_\eta(N)+h_2(\varepsilon).
 \label{eq:quantum-finite-error}
\end{equation}
Here $x_i$ is the unconditional mean input photon number at use $i$,
and $c_\eta(N)=g(N)-g((1-\eta)N)$.
\end{lemma}
\begin{proof}
The common adaptive bound \eqref{eq:telescope}, including the final
discarding of unused memories, gives
$\ERi(\rho_n)\le\sum_i c_\eta(x_i)\le n c_\eta(N)$.
Combine it with \Cref{lem:quantum-test} and rearrange.
\end{proof}

\subsection{A matching finite-dimensional protocol}\label{subapp:lowerbound}
We show that photon-number-sector resources also give an achievable
rate matching the converse in the large-block limit. We use entanglement distillation and
hashing~\cite{BennettPurification,BennettMixed,DevetakWinter}, applied to the
photon-number-conditioned resources. Related loss-conditioned optical
distillation is studied in Ref.~\cite{Winnel}.
\begin{lemma}[Sector hashing achieves the capacity rate]\label{lem:achieve}
For every $s,m$ with $m/s\le N$,
\begin{equation}
 Q_2(\cL_\eta,N)\ge F_{\eta,s}(m)/s.
 \label{eq:main-achievability}
\end{equation}
Moreover, $F_{\eta,s}(\lfloor sN\rfloor)/s\to c_\eta(N)$.
\end{lemma}

\begin{proof}
Alice prepares $\Phi_{s,m}$ locally and transmits its $s$ optical modes,
using exactly $m$ photons. The transmitted marginal $\pi_{s,m}$
is invariant under mode permutations, so each transmitted mode has
mean photon number $m/s$. Thus the construction also obeys the
per-use unconditional constraint whenever $m/s\le N$.
Bob measures the total received number $k$ and announces it. Write $p_k=\binom mk\eta^kq^{m-k}$. Reverse hashing
on each conditional class achieves the rate
\begin{equation}
 r_k:=S((J_{s,m,k})_R)-S(J_{s,m,k})
       =\log_2d_{s,m}-\log_2d_{s,m-k}\ge0
\label{eq:sector-hashing-rate}
\end{equation}
per occurrence~\cite{DevetakWinter}. There are finitely many classes, all
finite-dimensional, so only finite-dimensional hashing is needed.

To obtain a fixed output length, repeat the block $b$ times and fix a
small slack $\delta>0$. For each class with $p_k>0$, retain a prescribed
$\ell_{k,b}=\lfloor b(p_k-\delta)_+\rfloor$ occurrences and run a
hashing code with target dimension
$2^{\lfloor\ell_{k,b}(r_k-\delta)_+\rfloor}$. The probability that
any class has too few occurrences tends to zero by the law of large
numbers and the finite number of classes. On a shortage, output a fixed
separable state in the same prescribed total target dimension; otherwise
discard excess occurrences and use the retained blocks. The shortage
probability and the hashing errors both vanish. Hashing errors need not
be heralded. Taking $b\to\infty$ and then $\delta\downarrow0$ gives
$\sum_kp_kr_k=F_{\eta,s}(m)$ per block, with all $bs$ transmissions
and all branches counted.

For $m_s=\lfloor sN\rfloor$ with $N>0$, Stirling's formula
and the weak law of large numbers~\cite{CoverThomas} give
$s^{-1}\log_2d_{s,m_s}\to g(N)$. If $L_s\sim\Bin(m_s,q)$, then
$L_s/s\to qN$ in probability, so
$s^{-1}\log_2d_{s,L_s}\to g(qN)$ in probability. The uniform bound
\begin{equation}
 0\le s^{-1}\log_2d_{s,L_s}\le g(L_s/s)\le g(N)
\label{eq:sector-dimension-domination}
\end{equation}
also gives convergence of expectations. Substitution into
\eqref{eq:sector-ree} proves the limit. For each desired accuracy one
first fixes a sufficiently large finite $s$, then repeats that block
arbitrarily many times. More explicitly, for any rate below $c_\eta(N)$,
choose such an $s$ and a fixed positive slack before taking the coding
limit $b\to\infty$. For lengths not divisible by $s$, pad the remaining
slots with vacuum; the asymptotic rate and energy budget are unchanged.
At $N=0$ all transmitted states are vacuum and
the rate is zero.
\end{proof}

\subsection{Proof of the quantum-capacity identity}
We combine the finite-error converse and the sector-hashing protocol,
then take the vanishing-error limit to determine $Q_2$.

\begin{proof}[Quantum-capacity part of \Cref{thm:capacity}]
Consider any admissible sequence with target dimension $d_n$,
rate $R_n=\log_2d_n/n$, and quantum error
$\varepsilon_n=1-F(\rho_n,\Phi_{d_n})\to0$.
\Cref{lem:quantum-finite-error} gives
\begin{equation}
 R_n\le
 \frac{c_\eta(N)+h_2(\varepsilon_n)/n}{1-\varepsilon_n}.
 \label{eq:quantum-weak-limit}
\end{equation}
Hence $\limsup_nR_n\le c_\eta(N)$, proving
$Q_2(\cL_\eta,N)\le c_\eta(N)$.
The matching sector-hashing construction in \Cref{lem:achieve}
gives the reverse inequality. Therefore
\begin{equation}
 Q_2(\cL_\eta,N)=g(N)-g((1-\eta)N).
\label{eq:quantum-capacity-conclusion}
\end{equation}
The replica limit was taken separately for every fixed finite protocol;
the communication limit $n\to\infty$ is taken only here. All channel
uses in the actual protocol retain their adaptive LOCC interleaving.
\end{proof}

\section{Secret-key and private capacities}\label{sec:key}
We now specify the adaptive secret-key agreement task and connect its
classical security criterion to the entanglement bound
\eqref{eq:telescope}. Private states and an energy-preserving conversion
supply the additional ingredients beyond the quantum proof. The
private-capacity conclusion follows from the operational equivalence
in Section~\ref{sec:results}.

\subsection{Systems and adaptive secret-key agreement}\label{sec:key-model}
This subsection identifies the information kept by Alice, Bob, and Eve,
the final classical key, and the retained quantum state whose
entanglement will bound the key rate. The protocol starts from local preparations and public randomness,
without shared entanglement or secret key. As in
\eqref{eq:adaptive-round-state}--\eqref{eq:adaptive-round-locc}, Alice
holds $X_iA_i$ before the $i$th call and Bob holds $Y_i$; transmitting
$A_i$ produces Bob's mode $B_i$. Eve collects the corresponding loss
environment $E_i$. Each subsequent LOCC acts on $X_i:B_iY_i$ and can
choose later preparations and measurements using the full public
history. Public messages are copied to Eve's transcript and to the
parties' classical records. The reduced input
$\rho^{(i)}_{X_iA_iY_i}$ is unconditional, so its transmitted marginal
obeys precisely the mean-energy accounting in \eqref{eq:energy}.

After the last transmission, local processing and public communication
produce classical key registers $K_A,K_B$, each with alphabet
$\{0,\ldots,d-1\}$. Let $E$ contain all of Eve's quantum systems and
let $Z$ be her classical copy of the complete public transcript. Write
$\kappa_{K_AK_BEZ}$ for the resulting state. In this notation,
the ideal shared key and operational error from
\eqref{eq:classical-task-error} are
\begin{equation}
 \tau_d^{\rm key}=\frac1d\sum_{a=0}^{d-1}|aa\rangle\langle aa|,
 \qquad
 e=\tfrac12\left\|\kappa_{K_AK_BEZ}
       -\tau_d^{\rm key}\otimes\kappa_{EZ}\right\|_1.
 \label{eq:key-error}
\end{equation}
The ideal reference uses the actual Eve-and-transcript marginal
$\kappa_{EZ}$, and $0\le e\le1$. Agreement, uniformity, and independence
from Eve and the transcript are all included in this error, as in
the composable trace-distance framework
of Refs.~\cite{RennerKonig,RennerSecurity,PortmannRenner}.

For an $n$-use sequence, write $d_n$ for the key alphabet size,
$e_n$ for its error, and $R_n=\log_2d_n/n$ for its rate.
The capacity $K_2(\cL_\eta,N)$ is the supremum of
rates $R$ achievable with $\liminf_nR_n\ge R$ and $e_n\to0$, under
the same unconditional energy constraint \eqref{eq:energy} and
adaptive LOCC model as the quantum capacity. The parties have no
initial shared entanglement or secret key.

For the entanglement argument, Alice and Bob retain local purifications
of their operations in auxiliary registers $S_A,S_B$, called shields,
and postpone the final key readout. The resulting quantum output
$\zeta_n$ is shared across $K_AS_A:K_BS_B$; measuring its key registers
recovers the original key protocol, including Eve's state in the
corresponding extension. The construction in
\Cref{lem:key-conversion} preserves every transmitted optical marginal.
Thus the bound \eqref{eq:telescope} applies to this retained quantum
state with the original energy budget. It does not apply as a lower
bound on the entanglement of the measured classical key registers:
their reduced state is separable. No additional energy constraint is
placed on the local shields.

Figure~\ref{fig:adaptive-key-protocol} shows the adaptive key protocol,
Eve's accumulated information, and the relation between the retained
quantum output and the measured key state.

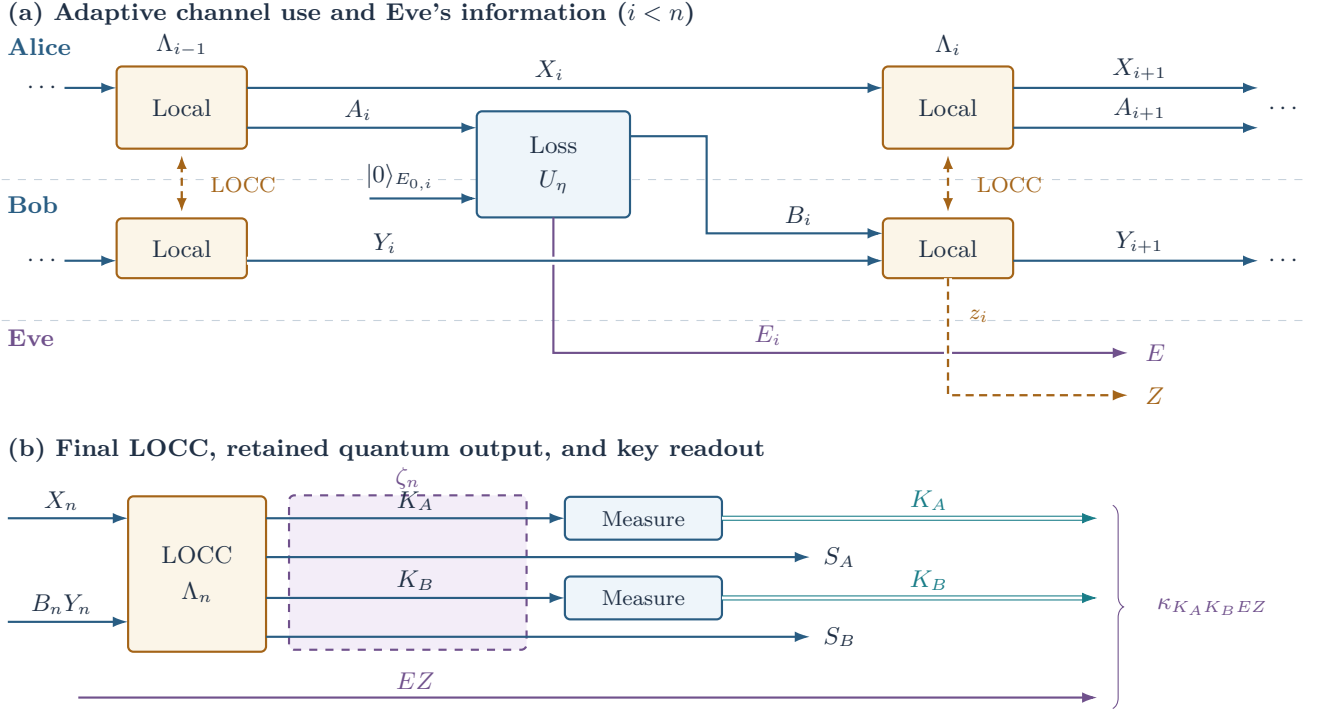
\begin{figure}[!htbp]
\centering
\begingroup
\definecolor{protink}{HTML}{203045}
\definecolor{protblue}{HTML}{295D83}
\definecolor{protbluefill}{HTML}{EEF5FA}
\definecolor{protamber}{HTML}{A46418}
\definecolor{protamberfill}{HTML}{FBF5E8}
\definecolor{protpurple}{HTML}{70518C}
\definecolor{protpurplefill}{HTML}{F3EDF8}
\definecolor{protteal}{HTML}{197D88}
\definecolor{protgray}{HTML}{C7D2DC}
\tikzset{
  every node/.style={text=protink,inner sep=2pt},
  pq/.style={draw=protblue,line width=.8pt,-{Latex[length=2mm,width=1.4mm]}},
  pqline/.style={draw=protblue,line width=.8pt},
  pc/.style={draw=protamber,densely dashed,line width=.85pt,
    -{Latex[length=2mm,width=1.4mm]}},
  pcc/.style={draw=protamber,densely dashed,line width=.85pt,
    {Latex[length=1.8mm]}-{Latex[length=1.8mm]}},
  pk/.style={draw=protteal,double,double distance=1.1pt,line width=.4pt,
    -{Latex[length=2mm,width=1.4mm]}},
  pe/.style={draw=protpurple,line width=.8pt,
    -{Latex[length=2mm,width=1.4mm]}},
  pgate/.style={draw=protblue,fill=protbluefill,line width=.8pt,rounded corners=2pt},
  plocal/.style={draw=protamber,fill=protamberfill,line width=.8pt,rounded corners=2pt},
  pstate/.style={draw=protpurple,fill=protpurplefill,line width=.8pt,rounded corners=2pt},
  ppanel/.style={font=\small\bfseries,anchor=west},
  powner/.style={font=\small\bfseries,anchor=west,text=protblue},
  plabel/.style={font=\small},
  psmall/.style={font=\footnotesize}
}

\resizebox{0.99\linewidth}{!}{%
\begin{tikzpicture}[x=1cm,y=.65cm,font=\small]
\path[use as bounding box] (0,-1.55) rectangle (17,13.1);
\node[ppanel] at (0,12.8) {(a) Adaptive channel use and Eve's information ($i<n$)};
\node[powner] at (0,12.15) {Alice};
\node[powner] at (0,8.98) {Bob};
\node[powner,text=protpurple] at (0,6.3) {Eve};
\draw[draw=protgray,dashed] (0,9.48)--(17,9.48);
\draw[draw=protgray,dashed] (0,6.65)--(17,6.65);

\foreach \x in {1.5,11.5}{
  \draw[plocal] (\x,10.1) rectangle ++(1.7,1.65);
  \node at (\x+.85,10.92) {Local};
  \draw[plocal] (\x,7.5) rectangle ++(1.7,1.2);
  \node at (\x+.85,8.1) {Local};
  \draw[pcc] (\x+.85,8.83)--(\x+.85,9.97);
}
\node[plabel] at (2.35,12.18) {$\Lambda_{i-1}$};
\node[plabel] at (12.35,12.18) {$\Lambda_i$};
\node[psmall,anchor=west,text=protamber] at (2.65,9.37) {LOCC};
\node[psmall,anchor=west,text=protamber] at (12.65,9.37) {LOCC};
\node at (.55,11.32) {$\cdots$};
\node at (.55,7.85) {$\cdots$};
\draw[pq] (.82,11.32)--(1.5,11.32);
\draw[pq] (.82,7.85)--(1.5,7.85);
\draw[pq] (3.2,11.32)--(11.5,11.32);
\node[plabel,anchor=south] at (7.15,11.36) {$X_i$};

\draw[pgate] (6.2,8.72) rectangle (8.2,10.85);
\node[align=center] at (7.2,9.78) {Loss\\[2pt]$U_\eta$};
\draw[pq] (3.2,10.52)--(6.2,10.52);
\node[plabel,anchor=south] at (4.65,10.56) {$A_i$};
\draw[pq] (4.8,9.1)--(6.2,9.1);
\node[plabel,anchor=south,fill=white,inner sep=1pt]
  at (5.2,9.2) {$|0\rangle_{E_{0,i}}$};
\draw[pq] (8.2,10.35)--(9.2,10.35)--(9.2,8.4)--(11.5,8.4);
\node[plabel,anchor=south] at (10.4,8.44) {$B_i$};
\draw[pe] (7.2,8.72)--(7.2,6.0)--(14.7,6.0);
\node[plabel,anchor=south,text=protpurple] at (10,6.04) {$E_i$};
\node[plabel,anchor=west,text=protpurple] at (14.85,6.0) {$E$};
\draw[pq,preaction={draw=white,line width=3pt}]
  (3.2,7.85)--(11.5,7.85);
\node[plabel,anchor=south] at (5.0,7.89) {$Y_i$};
\draw[pq] (13.2,11.32)--(16.4,11.32);
\node[plabel,anchor=south] at (14.85,11.36) {$X_{i+1}$};
\draw[pq] (13.2,10.52)--(16.4,10.52);
\node[plabel,anchor=south] at (14.85,10.56) {$A_{i+1}$};
\draw[pq] (13.2,7.85)--(16.4,7.85);
\node[plabel,anchor=south] at (14.85,7.89) {$Y_{i+1}$};
\node at (16.75,10.9) {$\cdots$};
\node at (16.75,7.85) {$\cdots$};
\draw[pc,preaction={draw=white,line width=3pt}]
  (12.35,7.5)--(12.35,5.15)--(14.7,5.15);
\node[psmall,anchor=west,text=protamber] at (12.55,6.75) {$z_i$};
\node[plabel,anchor=west,text=protamber] at (14.85,5.15) {$Z$};

\node[ppanel] at (0,4.05) {(b) Final LOCC, retained quantum output, and key readout};
\draw[plocal] (1.65,0) rectangle (3.45,3.12);
\node[align=center] at (2.55,1.56) {LOCC\\[3pt]$\Lambda_n$};
\draw[pq] (.08,2.68)--(1.65,2.68);
\draw[pq] (.08,.6)--(1.65,.6);
\node[plabel,anchor=south] at (.78,2.72) {$X_n$};
\node[plabel,anchor=south] at (.78,.64) {$B_nY_n$};
\draw[pstate,dashed] (3.75,.04) rectangle (6.85,3.14);
\node[plabel,text=protpurple] at (5.3,3.48) {$\zeta_n$};
\draw[pq] (3.45,2.68)--(7.35,2.68);
\node[plabel,anchor=south] at (5.4,2.72) {$K_A$};
\draw[pq] (3.45,1.9)--(10.55,1.9);
\node[plabel,anchor=west] at (10.65,1.9) {$S_A$};
\draw[pq] (3.45,1.08)--(7.35,1.08);
\node[plabel,anchor=south] at (5.4,1.12) {$K_B$};
\draw[pq] (3.45,.3)--(10.55,.3);
\node[plabel,anchor=west] at (10.65,.3) {$S_B$};
\draw[pgate] (7.35,2.25) rectangle (9.4,3.11);
\draw[pgate] (7.35,.65) rectangle (9.4,1.51);
\node[psmall] at (8.375,2.68) {Measure};
\node[psmall] at (8.375,1.08) {Measure};
\draw[pk] (9.4,2.68)--(14.3,2.68);
\node[plabel,anchor=south,text=protteal] at (12.1,2.73) {$K_A$};
\draw[pk] (9.4,1.08)--(14.3,1.08);
\node[plabel,anchor=south,text=protteal] at (12.1,1.13) {$K_B$};
\draw[pe] (1,-.92)--(14.3,-.92);
\node[plabel,anchor=south,text=protpurple] at (5.4,-.87) {$EZ$};
\draw[draw=protpurple,decorate,decoration={brace,amplitude=4pt}]
  (14.5,2.94)--(14.5,-1.14);
\node[align=center,text=protpurple] at (15.8,.9) {$\kappa_{K_AK_BEZ}$};
\end{tikzpicture}%
}
\endgroup
\caption{Adaptive secret-key agreement, including Eve's information.
(a) The channel dilation $U_\eta$ takes the transmitted mode $A_i$
and vacuum input $E_{0,i}$ to Bob's output $B_i$ and Eve's loss mode
$E_i$. The memories $X_i,Y_i$ bypass the channel. Successive LOCCs
retain the public history for feedback; the messages $z_i$ generated
within $\Lambda_i$ are also appended to Eve's transcript $Z$.
Eve collects all loss modes in $E$ and all public messages, including
those from the initial and final LOCCs, in $Z$.
(b) In the retained implementation, the final LOCC prepares the quantum
state $\zeta_n$ on $K_AS_A:K_BS_B$. Local measurements of $K_A,K_B$
give the classical key (double lines); the shields stay local and are
traced out when forming $\kappa_{K_AK_BEZ}$. The $EZ$ wire carries
the correlated Eve systems throughout. The converse evaluates REE on
$\zeta_n$ before the measurements and operational security on
$\kappa_{K_AK_BEZ}$ after them.}
\label{fig:adaptive-key-protocol}
\end{figure}

\subsection{Private states and privacy tests}
\label{sec:privacy-detail}
We introduce quantum states whose local measurements yield secret key,
and tests that quantify the entanglement required to approximate them.
These tests extend the maximally entangled-state test to include local
auxiliary registers, called shields. Private states express the entanglement structure underlying secret
key~\cite{SecureBound,HHHO}. Let $K_A,K_B$ be key registers of dimension
$d$, and let $S_A,S_B$ be arbitrary local shield registers. Write
$S=S_AS_B$ and let $\theta_S$ be a density operator on the shield.
For a family of unitaries $U_{ij}$ on $S$, define the controlled
twisting $U_{\rm tw}$ and private state $\gamma$ by
\begin{equation}
 U_{\rm tw}=\sum_{i,j=0}^{d-1}|ij\rangle\langle ij|_{K_AK_B}\otimes U_{ij},
 \qquad
 \gamma=U_{\rm tw}(\Phi_d^{K_AK_B}\otimes\theta_S)U_{\rm tw}^\dagger.
 \label{eq:private-state-definition}
\end{equation}
Here $\Phi_d$ is the maximally entangled projector defined in
\eqref{eq:target-fidelity}. The bipartition for REE is
$K_AS_A:K_BS_B$. The corresponding privacy-test projector is
\begin{equation}
 \Pi_\gamma=U_{\rm tw}(\Phi_d^{K_AK_B}\otimes I_S)U_{\rm tw}^\dagger,
 \qquad \Tr(\Pi_\gamma\gamma)=1.
 \label{eq:privacy-test-definition}
\end{equation}
This test is the untwisted maximally entangled overlap underlying the
earlier private-state relative-entropy bound~\cite[Appendix~XV.H]{HHHO}.
With trivial shields, this reduces to the maximally entangled-state
test used in Appendix~\ref{sec:quantum}.

\begin{lemma}[Privacy-test bound for regularized REE]\label{lem:privacy}
Let $\zeta$ be a state on $K_AK_BS_AS_B$ with finite
$\ER(K_AS_A:K_BS_B)_\zeta$, and let $0\le\varepsilon<1$.
For the privacy test $\Pi_\gamma$ in \eqref{eq:privacy-test-definition},
if $\Tr(\Pi_\gamma\zeta)\ge1-\varepsilon$, then
\begin{equation}
 \ERi(\zeta)\ge(1-\varepsilon)\log_2d-h_2(\varepsilon).
 \label{eq:main-privacy}
\end{equation}
No bound on the dimension of the shield is needed.
\end{lemma}

\begin{proof}
The separable privacy-test bound is already established in
Ref.~\cite[Eq.~(281)]{HHHO}:
\begin{equation}
 \Tr(\Pi_\gamma\sigma)\le\frac1d,
 \qquad \sigma\in\SEP(K_AS_A:K_BS_B).
 \label{eq:privacy-test-product-overlap}
\end{equation}
Its product-vector proof uses only unitary invariance of the norm and
Cauchy--Schwarz, so it applies to arbitrary separable shield spaces.
Convexity and trace-norm continuity extend it to the closed separable set.

For every positive integer $r$, the tensor product of $r$ privacy
tests is itself a privacy test with key dimension $d^r$. Consequently
any separable state across the aggregate Alice--Bob cut passes all
$r$ tests with probability at most $d^{-r}$.
Apply the threshold-testing argument in the proof of
\Cref{lem:quantum-test}, replacing $\Phi_d$ by $\Pi_\gamma$.
That argument uses only the acceptance probability on the tested state
and the bound $d^{-r}$ for each subset of $r$ tests. Both properties
hold here, so it gives \eqref{eq:main-privacy} with no restriction on
the shield dimension.
\end{proof}

\subsection{Preserving the optical inputs in the key-to-private-state reduction}
\label{sec:key-detail}
We relate an operational key protocol to a retained quantum state on
which a privacy test can be applied. The essential energy point is that
retaining local purifications preserves every transmitted input marginal.
The construction concerns local instruments and public communication;
it is independent of the optical channel and will also be used in
Appendix~\ref{sec:amplifier}.

\begin{lemma}[Energy-preserving key-to-private-state reduction]
\label{lem:key-conversion}
For a protocol in the finite/countable-outcome model with operational
error $e$ in \eqref{eq:key-error}, there is a protocol retaining local
shield registers whose final quantum state $\zeta_n$, before the
key registers are measured, passes a privacy test of key dimension
$d$ with probability at least $(1-e)^2$.
Every transmitted marginal and hence every unconditional input energy
$x_i=\Tr(\hat n_{A_i}\rho^{(i)}_{A_i})$ is unchanged.
\end{lemma}

\begin{proof}
We make the retained implementation explicit for normal local
instruments on separable Hilbert spaces. At a fixed public history,
write Alice's instrument, with finite or countable outcome $z$, as
\begin{equation}
 \mathcal I_z(T)=\sum_\mu A_{z\mu}T A_{z\mu}^\dagger,
 \qquad \sum_{z,\mu}A_{z\mu}^\dagger A_{z\mu}=I,
 \label{eq:key-countable-instrument}
\end{equation}
where the latter sum converges strongly. Retain a local Kraus register
$S_A^{\rm K}$ and set
$\widetilde V_z=\sum_\mu A_{z\mu}\otimes\ket\mu_{S_A^{\rm K}}$.
Then tracing this register recovers $\mathcal I_z$, and
$\sum_z\widetilde V_z^\dagger\widetilde V_z=I$, as in a
Stinespring dilation~\cite{Stinespring}.
For Alice's public message, use the purifying isometry
\begin{equation}
 V_{\rm pub}=\sum_z\widetilde V_z\otimes
       \ket z_{S_A^{\rm p}}\ket z_{Z_B}\ket z_{Z_E}.
 \label{eq:key-public-isometry}
\end{equation}
Here $Z_B,Z_E$ are Bob's and Eve's public records; $S_A^{\rm p}$ is an
extra local copy that is never subsequently altered. Alice's active
record may be included among the outputs of $\mathcal I_z$.
Orthogonality of the records gives $V_{\rm pub}^\dagger V_{\rm pub}=I$;
the countable sum therefore defines an isometry into separable spaces.

Two reductions establish the operational claims. Tracing
$S_A^{\rm K}S_A^{\rm p}$ recovers the original public instrument,
including Eve's classical record. Tracing $Z_E$ instead gives
\begin{equation}
 \sum_z\widetilde V_zT\widetilde V_z^\dagger\otimes
       \proj z_{S_A^{\rm p}}\otimes\proj z_{Z_B}.
 \label{eq:key-retained-locc}
\end{equation}
This is a local instrument retaining its dilation systems, followed by
ordinary classical communication of $z$, and is therefore LOCC across
the enlarged laboratories. Equation~\eqref{eq:key-public-isometry}
describes a purification of that process; it supplies no extra quantum
communication to the parties.

Use the analogous construction for Bob's messages, purify initial
local preparations locally, and retain discarded local systems.
Future operations follow the original active records while leaving
the added registers unused. Induction over the finite local rounds
and channel calls shows that tracing those registers recovers the
original evolution jointly with Eve at every stage. In particular,
\begin{equation}
 \widetilde\rho^{(i)}_{A_i}=\rho^{(i)}_{A_i},
 \qquad \widetilde x_i=x_i
 \quad\text{for every channel use }i.
 \label{eq:key-exact-energy}
\end{equation}
The enlarged initial state is separable and all between-call maps
remain LOCC after Eve is traced out. This implements the coherent
protocol construction of Ref.~\cite[Sec.~IX.A]{HHHO} in the stated
countable-outcome model, without any transcript-entropy assumption.
The initial assumption excludes any preshared secret key, including
a key encoded in a separable Alice--Bob state.

Implement the final key measurements by coherent local copying,
retaining the copies in the shields $S=S_AS_B$. Denote all of Eve's
systems, including her transcript, by $E_{\rm all}$. The global
vector $|\Psi\rangle_{K_AK_BSE_{\rm all}}$ purifies the actual
measured key state $\kappa_{K_AK_BE_{\rm all}}$; tracing its local
shields $S$ recovers that state. Tracing $E_{\rm all}$ instead gives
the retained quantum state $\zeta_n$. With squared fidelity
$F(\rho,\sigma)=\|\sqrt{\rho}\sqrt{\sigma}\|_1^2$, the
Fuchs--van de Graaf inequality~\cite{FuchsGraaf} gives
\begin{equation}
 F(\kappa,\tau_d^{\rm key}\otimes\kappa_{E_{\rm all}})
 \ge (1-e)^2.
\label{eq:key-ideal-fidelity}
\end{equation}
Uhlmann's theorem~\cite{UhlmannFidelity} supplies a purification of the ideal state,
after local enlargement of the shields if needed, of the form
\begin{equation}
 |\Gamma\rangle=\frac1{\sqrt d}\sum_{a=0}^{d-1}
    |aa\rangle|\phi_a\rangle_{SE_{\rm all}},\qquad
 \Tr_S|\phi_a\rangle\langle\phi_a|=\kappa_{E_{\rm all}},
\label{eq:key-ideal-purification}
\end{equation}
with squared overlap at least $(1-e)^2$ with $|\Psi\rangle$.
All $|\phi_a\rangle$ purify the same Eve state, so they are related
by isometries between their purifying support subspaces. To extend
these to unitaries, append a pure countably infinite ancillary register
locally and embed the original shield in its zero-ancilla subspace.
The purifying supports have equal dimensions, and their orthogonal
complements can both be made countably infinite. An orthonormal-basis
extension then gives a unitary on the full enlarged shield space.
Therefore
$\gamma=\Tr_{E_{\rm all}}|\Gamma\rangle\langle\Gamma|$
is a private state~\cite{HHHO}. The twisting specifies a comparison
state; it need not be implemented by the parties.

Fidelity increases under tracing out Eve, giving
$F(\zeta_n,\gamma)\ge(1-e)^2$. Data processing under
$\{\Pi_\gamma,I-\Pi_\gamma\}$, whose acceptance on $\gamma$ is one,
then gives
\begin{equation}
 \Tr(\Pi_\gamma\zeta_n)\ge F(\zeta_n,\gamma)\ge(1-e)^2.
\label{eq:key-privacy-acceptance}
\end{equation}
No bound on the shield dimension or memory entropy is used.
\end{proof}

\subsection{Finite-error key weak converse}\label{sec:finite-error}
We combine the privacy test and the energy-preserving reduction with
the common adaptive entanglement bound. This yields a rate bound for
the actual classical key, with its finite security error kept explicit.

\begin{lemma}[Private-state and operational-key bounds]\label{lem:finite-error}
Suppose an $n$-use protocol obeys \eqref{eq:energy} and its retained
quantum output $\zeta_n$ passes a privacy test $\Pi_\gamma$ of key
dimension $d$ with probability at least $1-\varepsilon$, where
$0\le\varepsilon<1$. Then
\begin{equation}
 (1-\varepsilon)\log_2d
 \le\sum_{i=1}^{n}c_\eta(x_i)+h_2(\varepsilon)
 \le n c_\eta(N)+h_2(\varepsilon).
 \label{eq:finite-error}
\end{equation}
For an operational classical key of alphabet size $d$ and error $e<1$
defined in \eqref{eq:key-error},
\begin{equation}
 (1-e)^2\log_2d
 \le\sum_{i=1}^{n}c_\eta(x_i)+h_2(2e-e^2)
 \le n c_\eta(N)+h_2(2e-e^2).
 \label{eq:operational-error}
\end{equation}
Here $h_2$ is the binary entropy and
$c_\eta(x)=g(x)-g((1-\eta)x)$, as in the quantum bound.
\end{lemma}
\begin{proof}
The common output bound \eqref{eq:telescope} applies to $\zeta_n$
including its local shield registers.
\Cref{lem:privacy} gives
$\ERi(\zeta_n)\ge(1-\varepsilon)\log_2d-h_2(\varepsilon)$;
combining the two inequalities proves \eqref{eq:finite-error}.

For an operational key, \Cref{lem:key-conversion} produces a retained
protocol with the same optical input marginals and privacy-test
acceptance at least $(1-e)^2$. Its failure is therefore at most
$1-(1-e)^2=2e-e^2$. Apply \eqref{eq:finite-error} with
$\varepsilon=2e-e^2$ to obtain \eqref{eq:operational-error}.
\end{proof}

The error $e$ refers to the measured classical key and Eve, whereas
$\ERi$ is evaluated on the retained quantum state before the key
measurement. The conversion connects these quantities without changing
any transmitted energy. As in the quantum proof, the replica limit
is taken for each fixed finite protocol before the communication limit.

\subsection{Proof of the secret-key-capacity identity}
We take the vanishing-error limit of the key converse. Measuring the
entanglement supplied by the quantum protocol gives the matching
achievable key rate, completing the determination of $K_2$.

\begin{proof}[Secret-key part of \Cref{thm:capacity}]
For an admissible sequence with operational key error $e_n\to0$,
target size $d_n$, and rate $R_n=\log_2d_n/n$,
\Cref{lem:finite-error} gives
\begin{equation}
 R_n\le
 \frac{c_\eta(N)+h_2(2e_n-e_n^2)/n}{(1-e_n)^2}.
 \label{eq:key-weak-limit}
\end{equation}
Taking the limit proves
$K_2(\cL_\eta,N)\le c_\eta(N)$.
Conversely, measuring distilled maximally entangled pairs gives secret
key at the same asymptotic rate~\cite{DevetakWinter,HHHO}.
More explicitly, if the distributed target has infidelity
$\varepsilon$, Uhlmann's theorem and the fidelity--trace-distance
inequality~\cite{UhlmannFidelity,FuchsGraaf} give half trace distance
at most $\sqrt{\varepsilon}$ from an ideal maximally entangled
state tensored with some Eve state. Local measurement is contractive;
replacing that comparison Eve state by the actual marginal costs
at most a further $\sqrt{\varepsilon}$. Hence the operational error
in \eqref{eq:classical-task-error} satisfies
\begin{equation}
 e_{\rm cl}\le\min\{1,2\sqrt{\varepsilon}\}.
 \label{eq:entanglement-to-key-error}
\end{equation}
This conversion is independent of the optical channel and includes
Eve's complete public record. The quantum result established in
Appendix~\ref{sec:quantum} therefore implies
\begin{equation}
 K_2(\cL_\eta,N)
 \ge Q_2(\cL_\eta,N)=c_\eta(N).
\label{eq:key-achievable-rate}
\end{equation}
The operational equivalences \eqref{eq:operational-equivalences} also
give $D_2(\cL_\eta,N)=Q_2(\cL_\eta,N)$ and
$P_2(\cL_\eta,N)=K_2(\cL_\eta,N)$, completing
\Cref{thm:capacity} for all four tasks.
\end{proof}

\section{Extension to the quantum-limited amplifier}
\label{sec:amplifier}
This appendix proves the capacity extension in
Theorem~\ref{thm:amplifier} using the pure-loss analysis. We first
identify the conditional amplifier resources with the loss resources
of Lemma~\ref{lem:resource}. Their REE therefore requires no new
separable comparison. The additional energy estimate uses
degradability, after which the block, replica, and operational
arguments developed above apply. The strong-converse threshold is
proved jointly with pure loss in Appendix~\ref{sec:strong}.
Throughout, $\kappa>1$ and $r_\kappa$ is defined in
\eqref{eq:amp-rate}.

\subsection{Using the same sector teleportation measurement}
\label{sec:amp-simulation}
We verify the covariance hypothesis of Lemma~\ref{prop:sector};
its Bell measurement and contraction identity then give the exact
amplifier simulator without another teleportation proof.
In the two-mode-squeezing dilation of
\eqref{eq:amp-channel}, rotating the system modes by $U\in U(s)$
and the vacuum environment modes by $\bar U$ preserves the
interaction. Thus the input representation
$W_m(U)=\mathrm{Sym}^mU$ has output correction
$V(U)=\bigoplus_{k\ge0}W_k(U)$, a bounded unitary on the full
output Fock space. For each operator $T$ on $\cH_{s,m}$,
\begin{equation}
 \cA_\kappa^{\otimes s}(W_m(U)T W_m(U)^\dagger)
 =V(U)\cA_\kappa^{\otimes s}(T)V(U)^\dagger.
 \label{eq:amp-covariance}
\end{equation}
Applying Lemma~\ref{prop:sector} to
$J^\kappa_{s,m}$ in \eqref{eq:amp-resource} gives an LOCC map
$\mathcal T^\kappa_{s,m}$ such that
\begin{equation}
 (\id_{XY}\otimes\cA_\kappa^{\otimes s})(\rho_{XA^sY})
 =\mathcal T^\kappa_{s,m}(\rho_{XA^sY}\otimes J^\kappa_{s,m})
 \quad\text{if }\supp\rho_{A^s}\subseteq\cH_{s,m}.
 \label{eq:amp-simulation}
\end{equation}
There are at most $d_{s,m}^4$ Bell outcomes, whose probabilities
are input independent. The output correction may act on an
infinite-dimensional space; finiteness of the Bell POVM requires
only the finite input sector. Alice retains exactly $m$ photons
in the resource, and Bob has mean $\kappa m+s(\kappa-1)$.

\subsection{Conditional resources inherited from pure loss}
\label{sec:amp-resources}
We establish the state correspondence that transfers the pure-loss
REE calculation to amplification. For a fixed input occupation
$\mathbf r$ with $|\mathbf r|=m$, the amplifier isometry is
\begin{equation}
 \mathcal V_\kappa\ket{\mathbf r}
 =\kappa^{-(s+m)/2}
 \sum_{\mathbf v\in\mathbb N_0^s}
 \left(\frac{\kappa-1}{\kappa}\right)^{|\mathbf v|/2}
 \sqrt{\prod_{j=1}^s\binom{r_j+v_j}{v_j}}\,
 \ket{\mathbf r+\mathbf v}_B\ket{\mathbf v}_E.
 \label{eq:amp-fock}
\end{equation}
This is the number-basis form of the vacuum two-mode-squeezing
dilation~\cite{IvanSabapathySimon}. The environment photon number
equals the number of added photons.

\begin{lemma}[Conditional loss--amplifier correspondence]
\label{lem:amp-duality}
For an arbitrary input in $\cH_{s,m}$, the added photon number
$L$ has the input-independent distribution
\begin{equation}
 p_\ell=\binom{s+m+\ell-1}{\ell}
        \frac{(\kappa-1)^\ell}{\kappa^{s+m+\ell}},
 \qquad \EE L=(\kappa-1)(s+m).
 \label{eq:amp-added-distribution}
\end{equation}
Let $J^\kappa_{s,m,\ell}$ be the normalized resource conditioned
on $L=\ell$. With $J_{s,k,m}$ denoting the pure-loss resource
of Lemma~\ref{lem:resource} for input number $k$ and received
number $m$, one has
\begin{equation}
 J^\kappa_{s,m,\ell}\cong J_{s,m+\ell,m},
 \label{eq:amp-duality}
\end{equation}
where $\cong$ denotes exchange of the two parties. The normalized
conditional states do not depend on the gain or transmissivity.
\end{lemma}

\begin{proof}
Multiplying the generating functions $(1-t)^{-(r_j+1)}$ gives
\begin{equation}
 \sum_{|\mathbf v|=\ell}\prod_{j=1}^s\binom{r_j+v_j}{v_j}
 =\binom{s+m+\ell-1}{\ell}=:a_{s,m,\ell}.
 \label{eq:amp-normalization}
\end{equation}
For a fixed $\ell$, the joint $BE$ images of distinct occupations
$\mathbf r$ are orthogonal. Thus the conditional branch effect is
$p_\ell I_{s,m}$, proving input independence also for coherent
sector inputs and inputs correlated with memories. Normalization
and the mean in \eqref{eq:amp-added-distribution} follow from
the same generating function.

Applying \eqref{eq:amp-fock} to half of $\Phi_{s,m}$ gives the
conditional purification
\begin{equation}
 \ket{\Omega^\kappa_{s,m,\ell}}
 =\frac1{\sqrt{d_{s,m}a_{s,m,\ell}}}
 \sum_{\substack{|\mathbf r|=m\\|\mathbf v|=\ell}}
 \sqrt{\prod_{j=1}^s\binom{r_j+v_j}{v_j}}\,
 \ket{\mathbf r}_R\ket{\mathbf r+\mathbf v}_B\ket{\mathbf v}_E.
 \label{eq:amp-conditional-purification}
\end{equation}
Put $k=m+\ell$. The pure-loss purification
$(I\otimes V_{k,m})\ket{\Phi_{s,k}}$, with the symmetric split
from \eqref{eq:symmetric-split}, has the same occupation
coefficients after exchanging reference and received systems.
Indeed, the coefficient for $\mathbf r+\mathbf v=\mathbf n$
in $V_{k,m}\ket{\mathbf n}$ is
$[\prod_j\binom{n_j}{r_j}/\binom{k}{m}]^{1/2}$, and
\begin{equation}
 d_{s,m}a_{s,m,\ell}=d_{s,m+\ell}\binom{m+\ell}{m}.
 \label{eq:amp-dimension-normalization}
\end{equation}
The normalized purifications therefore agree. Tracing their
common environment proves \eqref{eq:amp-duality}.
\end{proof}

\begin{lemma}[Amplifier resource cost]
\label{lem:amp-resource}
The conditional and unconditional resource REEs are
\begin{align}
 \ER(J^\kappa_{s,m,\ell})
   &=\log_2\frac{d_{s,m+\ell}}{d_{s,\ell}},
   \label{eq:amp-conditional-ree}\\
 F_{\kappa,s}(m):=\ER(J^\kappa_{s,m})
   &=\sum_{\ell\ge0}p_\ell
      \log_2\frac{d_{s,m+\ell}}{d_{s,\ell}}.
   \label{eq:amp-resource-ree}
\end{align}
For $\beta_m=\cA_\kappa^{\otimes s}(\pi_{s,m})$ and
$\gamma_m=(\cA_\kappa^c)^{\otimes s}(\pi_{s,m})$, where
$\cA_\kappa^c$ is the complementary channel, this cost also equals
\begin{equation}
 F_{\kappa,s}(m)=S(\beta_m)-S(\gamma_m).
 \label{eq:amp-resource-coherent}
\end{equation}
\end{lemma}

\begin{proof}
Equation~\eqref{eq:amp-conditional-ree} follows immediately from
\eqref{eq:conditional-ree}, Lemma~\ref{lem:amp-duality}, and
invariance of REE under exchanging the parties. Bob's numbers
$m+\ell$ occupy orthogonal local subspaces, so the countable
local-flags identity \eqref{eq:flags} proves
\eqref{eq:amp-resource-ree}. The series is finite because
$d_{s,m+\ell}/d_{s,\ell}
=\prod_{j=1}^{s-1}(1+m/(\ell+j))\le d_{s,m}$.

The same correspondence transfers the marginal identities in
\eqref{eq:reference-marginal} and
\eqref{eq:conditional-environment}: conditionally, Bob's marginal
is $\pi_{s,m+\ell}$ and the environment marginal is
$\pi_{s,\ell}$. Hence
\begin{equation}
 \begin{aligned}
 S(\beta_m)&=H(L)+\sum_{\ell\ge0}p_\ell\log_2d_{s,m+\ell},\\
 S(\gamma_m)&=H(L)+\sum_{\ell\ge0}p_\ell\log_2d_{s,\ell}.
 \end{aligned}
 \label{eq:amp-marginal-entropies}
\end{equation}
Both optical states have finite mean energy and therefore finite
entropy. Moreover, $H(L)\le g((\kappa-1)(s+m))<\infty$ by
Lemma~\ref{lem:scalar}. Subtraction cancels the same count entropy
and proves \eqref{eq:amp-resource-coherent}.
\end{proof}

\subsection{Transferring the mean input energy to the resource cost}
\label{sec:amp-energy}
This is the additional energy estimate needed for amplification.
It compares the output and environment with product thermal
states through the degrading map; all sector dimensions and
entanglement identities have already been established.

\begin{lemma}[Amplifier energy estimate]
\label{lem:amp-energy}
For every $s\ge1$ and $m\ge0$,
\begin{equation}
 F_{\kappa,s}(m)\le s r_\kappa(m/s).
 \label{eq:amp-sector-cost}
\end{equation}
The function $r_\kappa$ is increasing and concave. Thus, for any
nonnegative integer-valued $M$ with $\EE M=sx<\infty$,
\begin{equation}
 \EE F_{\kappa,s}(M)\le s r_\kappa(x).
 \label{eq:amp-average-cost}
\end{equation}
\end{lemma}

\begin{proof}
The amplifier is degradable~\cite{CarusoGiovannetti}:
$\cA_\kappa^c=\mathcal D_\kappa\circ\cA_\kappa$.
A phase-conjugating realization of $\mathcal D_\kappa$, with an
independent vacuum mode $\hat f$, is
\begin{equation}
 \hat c=\sqrt{\frac{\kappa-1}{\kappa}}\,\hat b^\dagger
        +\sqrt{\frac{2\kappa-1}{\kappa}}\,\hat f.
 \label{eq:amp-degrading}
\end{equation}
Substitution of \eqref{eq:amp-channel} gives input coefficient
$\sqrt{\kappa-1}\,\hat a^\dagger$ and vacuum annihilation noise
whose total squared coefficient is $\kappa$, as in the complementary
channel. This verifies the degrading relation on arbitrary inputs.

Set $x=m/s$, $y_x=\kappa x+\kappa-1$, and
$z_x=(\kappa-1)(x+1)$. The states $\beta_m,\gamma_m$ have total
mean photon numbers $sy_x,sz_x$, and the thermal states from
\eqref{eq:thermal-state} obey
$\mathcal D_\kappa(\tau_{y_x})=\tau_{z_x}$.
Their logarithms are affine in photon number, so
\begin{equation}
 \begin{aligned}
 D(\beta_m\Vert\tau_{y_x}^{\otimes s})&=s g(y_x)-S(\beta_m),\\
 D(\gamma_m\Vert\tau_{z_x}^{\otimes s})&=s g(z_x)-S(\gamma_m).
 \end{aligned}
 \label{eq:amp-thermal-comparison}
\end{equation}
Data processing under $\mathcal D_\kappa^{\otimes s}$ says that
the first relative entropy is at least the second. Rearranging
and using \eqref{eq:amp-resource-coherent} gives
\eqref{eq:amp-sector-cost}. The reference thermal states are
faithful because $\kappa>1$, and all terms are finite, also at
$m=0$.

For $x\ge0$, differentiation gives
\begin{equation}
 r_\kappa''(x)=
 -\frac{2\kappa-1}{\ln2\,(x+1)(\kappa x+\kappa-1)
                         ((\kappa-1)(x+1)+1)}<0,
 \qquad \lim_{x\to\infty}r_\kappa'(x)=0.
 \label{eq:amp-concavity}
\end{equation}
Consequently $r_\kappa$ is increasing and concave, and Jensen's
inequality gives \eqref{eq:amp-average-cost}. We also record
\begin{equation}
 r_\kappa(0)=0,\qquad
 \lim_{x\to\infty}r_\kappa(x)=\log_2\frac{\kappa}{\kappa-1}.
 \label{eq:amp-rate-limits}
\end{equation}
Unlike the pure-loss estimate in Lemma~\ref{lem:energy}, this
resource estimate has no residual photon-count entropy term.
\end{proof}

\subsection{The block bound and its unbounded output tail}
\label{sec:amp-block}
We reuse the block argument of Lemma~\ref{prop:block}, changing
only the resource cost, the output mean, and the treatment of
the output support. This last step is needed because the amplifier
can create arbitrarily many photons from a finite input sector.

\begin{lemma}[Amplifier block gain]
\label{lem:amp-block}
Let $\rho_{XA^sY}$ have finite $\ER(XA^s:Y)_\rho$ and
$\Tr(\hat N_A\rho)=sx<\infty$. For
$\omega=(\id_{XY}\otimes\cA_\kappa^{\otimes s})(\rho)$,
\begin{equation}
 \ER(X:B^sY)_\omega
 \le\ER(XA^s:Y)_\rho+s r_\kappa(x)
       +g\bigl(s(\kappa x+\kappa-1)\bigr).
 \label{eq:amp-block}
\end{equation}
The memory spaces may be arbitrary separable Hilbert spaces.
\end{lemma}

\begin{proof}
First use finite-dimensional memories and a finite input cutoff.
Phase covariance gives the pinching identity
\eqref{eq:main-phase} with $\cA_\kappa$ in place of
$\cL_\eta$. The convexity, sector simulation, and selective
monotonicity argument of \eqref{eq:main-block-chain}, now using
\eqref{eq:amp-average-cost}, therefore yields
\begin{equation}
 \ER(\cP_B\omega)\le\ER(\rho)+s r_\kappa(x).
 \label{eq:amp-pinched-block}
\end{equation}
Only the subsequent output pinching step needs an additional limit.
With $Q_t^B=\sum_{k=0}^tP_k^B$, use the local channel
\begin{equation}
 \mathcal C_{B,t}(T)=Q_t^BTQ_t^B+
          \Tr[(I-Q_t^B)T]\proj{0^{\otimes s}}.
 \label{eq:amp-output-cutoff}
\end{equation}
It commutes with total-number pinching and does not increase the
output mean $s(\kappa x+\kappa-1)$. The truncated state
$\omega_t=(\id_{XY}\otimes\mathcal C_{B,t})(\omega)$ has finite
support, so Lemmas~\ref{lem:pinching} and~\ref{lem:scalar} give
\begin{equation}
 \ER(\omega_t)
 \le\ER(\cP_B\omega)+g\bigl(s(\kappa x+\kappa-1)\bigr)
 \le\ER(\rho)+s r_\kappa(x)+g\bigl(s(\kappa x+\kappa-1)\bigr).
 \label{eq:amp-output-tail-bound}
\end{equation}
The discarded probability tends to zero, hence $\omega_t\to\omega$
in trace norm by the gentle measurement estimate. Lower
semicontinuity of unregularized REE proves \eqref{eq:amp-block}
for finite memory and input support, without an output cutoff.

For arbitrary memories and finite-mean input, use exactly the
local memory truncations and vacuum input replacement from the
proof of Lemma~\ref{prop:block}. They give
$\rho_t\to\rho$, $\ER(\rho_t)\le\ER(\rho)$, and $x_t\le x$
by \eqref{eq:truncate-bounds}. The preceding bound already permits
infinite output support for each $\rho_t$. Monotonicity of
$r_\kappa$ and $g$ makes its right-hand side uniformly bounded
by that of \eqref{eq:amp-block}. Channel contractivity and lower
semicontinuity complete the limit, as in
\eqref{eq:infinite-block}. No memory entropy enters the proof.
\end{proof}

\subsection{Adaptive capacities from the shared operational tools}
\label{sec:amp-capacity-proof}
This subsection applies the replica and target-test arguments
already proved for pure loss. The new block estimate is the only
channel-dependent input to these converse steps.

In \eqref{eq:generic-block-hypothesis}, take $f=r_\kappa$ and
$h_s(x)=g(s(\kappa x+\kappa-1))$. Lemma~\ref{lem:scalar}
ensures $h_s(x)/s\to0$. Thus
\eqref{eq:generic-adaptive-conclusion} gives, for every fixed
adaptive amplifier protocol,
\begin{equation}
 \ERi(\zeta_n)\le\sum_{i=1}^n r_\kappa(x_i)\le n r_\kappa(N).
 \label{eq:amp-adaptive}
\end{equation}
As before, the replica limit precedes the communication limit,
and $\zeta_n$ may include all retained local shields.

\begin{proof}[Proof of the capacity identity in Theorem~\ref{thm:amplifier}]
The quantum target test of Lemma~\ref{lem:quantum-test} and the
channel-independent key reduction and privacy test of
Lemmas~\ref{lem:key-conversion} and~\ref{lem:privacy} imply
\begin{equation}
 \begin{aligned}
 (1-\varepsilon)\log_2d&\le n r_\kappa(N)+h_2(\varepsilon),\\
 (1-e)^2\log_2d&\le n r_\kappa(N)+h_2(2e-e^2).
 \end{aligned}
 \label{eq:amp-finite-errors}
\end{equation}
Here $\varepsilon$ is entanglement infidelity and $e$ is the
operational classical error of \eqref{eq:classical-task-error}.
Letting the corresponding error vanish proves the quantum and
key weak converses. The operational conversions in
Section~\ref{sec:results}, which add no optical calls, also give
the distribution and private-message converses.

For achievability, the established energy-constrained unassisted
quantum coding theorem supplies every rate below
$r_\kappa(N)$~\cite{QiWilde,WildeQi}. Its coherent information is
the output entropy minus the complementary entropy of a thermal
input of mean $N$, namely the two terms in \eqref{eq:amp-rate}.
An achievable quantum code distributes entanglement at the same
rate. Measuring the entangled target generates key with operational
error bounded by \eqref{eq:entanglement-to-key-error}; a one-time
pad gives private transmission. Thus all four lower bounds match.

To meet the stronger per-use unconditional constraint, publicly
permute the $n$ modes of an achievable parallel code uniformly,
and undo the permutation before decoding. The product channel
commutes with this permutation, preserving the conditional error
and rate. Each unconditional input mean becomes
\begin{equation}
 \overline x_i=\frac1n\sum_{j=1}^n x_j\le N.
 \label{eq:amp-per-use}
\end{equation}
The permutation alphabet is finite at each blocklength. The
converse already applies to this smaller protocol class.
At $N=0$, every transmitted mode is vacuum and independent of
the memories; amplification only creates a local state at Bob.
All four capacities are therefore zero, in agreement with
\eqref{eq:amp-rate-limits}.
\end{proof}

\subsection{Gaussian output entanglement from the pure-loss corollary}
\label{sec:amp-output}
The output-state identity \eqref{eq:amp-output-ree} follows
directly from Corollary~\ref{cor:tmsv}. This avoids a second
calculation of regularized REE for a Gaussian state.

\begin{proof}[Proof of \eqref{eq:amp-output-ree}]
For $N>0$, put
\begin{equation}
 M=\kappa N+\kappa-1,\qquad \eta'=\frac{N}{M}\in(0,1).
 \label{eq:amp-gaussian-parameters}
\end{equation}
In the convention where the vacuum covariance matrix is $I_2$,
the centered amplifier output $J^{\rm amp}_{\kappa,N}$ has local
blocks $(2N+1)I_2$, $(2M+1)I_2$, and correlation block
$2\sqrt{\kappa N(N+1)}\,Z$, where $Z=\mathrm{diag}(1,-1)$.
The pure-loss output $J_{\eta',M}$ has these local blocks in
reverse order and correlation block
$2\sqrt{\eta'M(M+1)}\,Z$. Since
$M+1=\kappa(N+1)$, the covariance matrices agree after
exchanging the parties. Centered Gaussian states are determined
by their covariance matrices~\cite{Weedbrook}, so
\begin{equation}
 J^{\rm amp}_{\kappa,N}\cong J_{\eta',M},\qquad
 \ERi(J^{\rm amp}_{\kappa,N})
 =g(M)-g((1-\eta')M)=r_\kappa(N).
 \label{eq:amp-gaussian-correspondence}
\end{equation}
The REE equality is Corollary~\ref{cor:tmsv}, together with
$(1-\eta')M=(\kappa-1)(N+1)$. At $N=0$, the output is a
product state, and both sides of \eqref{eq:amp-output-ree} vanish.
\end{proof}

\section{Mean-energy strong-converse thresholds}\label{sec:strong}
Recall that $R_n=\log_2d_n/n$ is the rate, $f_n$ is the entanglement
fidelity in \eqref{eq:quantum-error}, and $e_n$ is the operational key
error in \eqref{eq:key-error}. The threshold
$Q_2^{\dagger}$ is the infimum of values $C$ such that
every admissible sequence with $\liminf_nR_n>C$ has $f_n\to0$.
The threshold $K_2^{\dagger}$ is defined analogously
with $e_n\to1$. Both use the same mean-energy constraint
$\sum_i x_i\le nN$ as the vanishing-error capacities. The
error-preserving equivalences in Section~\ref{sec:results} give the
same thresholds for entanglement distribution and private transmission.
This appendix treats both pure loss and the amplifier extension using
the same dilution construction and error comparisons.

The operational conversions in \eqref{eq:operational-equivalences}
preserve the errors needed here. Standard teleportation has entanglement
fidelity equal to the resource's overlap with $\Phi_d$. For a uniform
message, one-time-pad encoding is a reversible classical relabeling
after appending an independent uniform register; its public ciphertext
is uniform and independent of Eve's previous marginal. Consequently
the quantum/distribution and key/private equivalences hold also for
the dagger thresholds. In particular, for every $N>0$ and $0<R<-\log_2(1-\eta)$, admissible
fixed-length protocols exist with asymptotic rate $R$ and quantum fidelity
bounded away from zero, or classical-task error bounded below one.
The proof below establishes this stronger operational statement.

The following lemma quantifies the resulting rate--error tradeoff and
proves the lower threshold in \Cref{thm:mean-sc}. The construction
already uses parallel inputs, so it is an admissible special case of
the adaptive model considered throughout the paper.

\subsection{Why unconditional mean energy has a different threshold}
We mix an energetic protocol with vacuum transmission using a public
flag. This construction respects the unconditional energy budget while
retaining a nonzero chance of success at rates above the energy-constrained
vanishing-error capacity.

\begin{lemma}[Vacuum dilution of an entire protocol]\label{lem:sc-dilution}
Let $E>0$ be a mean photon budget per use for an active
entanglement-generation protocol.
Suppose this $n$-use protocol has total mean input energy at most $nE$,
target dimension $d$, and entanglement fidelity $f_n^{\rm on}$.
For each $0<p\le1$ there is a protocol with the same target dimension,
mean energy at most $pnE$, and fidelity
\begin{equation}
 f_n=p f_n^{\rm on}+(1-p)/d.
 \label{eq:sc-dilution-f}
\end{equation}
For a key protocol of error $e_n^{\rm on}$, the analogous
construction has error
\begin{equation}
 e_n=pe_n^{\rm on}+(1-p)(1-1/d).
 \label{eq:sc-dilution-key}
\end{equation}
The key error is measured against an ideal key with the actual Eve
marginal. If the active code obeys a per-use unconditional mean bound
$E$, the diluted code obeys the corresponding bound $pE$.
\end{lemma}

\begin{proof}
Alice samples a Bernoulli flag and announces it publicly. With
probability $p$ the parties run the original protocol. Otherwise Alice
sends vacuum in all $n$ slots, Bob discards every optical output,
and the parties output $\ket0\ket0$ in
their $d$-dimensional output registers. Energy is linear in the input
state, so its mean is at most $pnE$. Let $\rho_n^{\rm on}$ be the
active protocol's entanglement output. After discarding their flag,
the diluted entanglement output is
$p\rho_n^{\rm on}+(1-p)\proj{00}$, whose overlap with $\Phi_d$ is
\eqref{eq:sc-dilution-f}.
This construction concerns entanglement generation, for which the
off-branch overlap is $\langle\Phi_d|00\rangle\langle00|\Phi_d\rangle=1/d$.
Quantum transmission is obtained afterwards by teleportation with the
generated state, preserving this entanglement fidelity and requiring
no additional optical transmission.

For the classical key, Eve retains the public flag. The two flag blocks
are orthogonal, so normalized trace distances add with their branch
weights. The distance from the fixed key $00$ to a uniform correlated
key of size $d$ is $1-1/d$, proving \eqref{eq:sc-dilution-key}.
Both branches use $n$ channel slots and produce the same prescribed
output length. The unsuccessful branch is counted in the error;
there is no conditioning of the rate on success.
\end{proof}

\subsection{A shared proof of the two mean-energy thresholds}
The dilution construction applies to either optical channel:
the inactive branch transmits vacuum and discards the received
state, which need not itself be vacuum. We combine it with the
established unconstrained converse, using a single order-of-limits
argument for both channels.

Write $f_{\mathcal N}$ for the finite-energy achievable rate
and $b_{\mathcal N}$ for its infinite-energy limit:
\begin{equation}
 \begin{array}{c|cc}
 \mathcal N & f_{\mathcal N}(x) & b_{\mathcal N}\\ \hline
 \cL_\eta & c_\eta(x) & -\log_2(1-\eta)\\
 \cA_\kappa & r_\kappa(x) & \log_2[\kappa/(\kappa-1)]
 \end{array}
 \label{eq:sc-channel-rates}
\end{equation}
Both $f_{\mathcal N}$ are strictly increasing and converge to
$b_{\mathcal N}$ as $x\to\infty$.

\begin{proof}[Proof of Theorem~\ref{thm:mean-sc} and \eqref{eq:amp-threshold}]
Fix $N>0$ and $0<R<b_{\mathcal N}$. Choose a finite $E>N$
with $R<f_{\mathcal N}(E)$. For pure loss, the sector-hashing
construction of Lemma~\ref{lem:achieve} supplies such rates with
vanishing error. For the amplifier, use the achievable quantum
codes discussed in Appendix~\ref{sec:amp-capacity-proof}.
In both cases the codes can satisfy the per-use unconditional
mean bound $E$: pure-loss sectors are permutation symmetric,
and the amplifier codes can be symmetrized by
\eqref{eq:amp-per-use}. Their distributed entanglement also gives
key with vanishing operational error.

Apply Lemma~\ref{lem:sc-dilution} with $p=N/E>0$. The diluted
protocols obey \eqref{eq:energy}, separately at every use, and
\begin{equation}
 \liminf_{n\to\infty}f_n\ge p>0,\qquad
 \lim_{n\to\infty}e_n=1-p<1.
 \label{eq:sc-positive-success}
\end{equation}
All branches have the same target size and use the same number
of channel slots. Fixed block lengths can be accommodated by
vacuum padding. Thus no threshold below $b_{\mathcal N}$ is
possible under either mean-energy convention.

For the upper threshold, use the unconstrained pure-loss and
quantum-limited-amplifier strong converses of
Ref.~\cite[Theorem~24, Corollary~25, and Eq.~(8.22)]{WTB}, with
the adaptive bosonic simulation control of
Ref.~\cite[Sec.~X.D]{Simulation}. In the squared-fidelity
convention, the imported finite-error bound is
\begin{equation}
 \log_2d\le n b_{\mathcal N}+C(\epsilon),\qquad
 C(\epsilon)=\log_2 6+
       2\log_2\frac{1+\epsilon}{1-\epsilon},
 \qquad 0<\epsilon<1,
 \label{eq:sc-imported-bound}
\end{equation}
for an output of fidelity at least $1-\epsilon$ to a private
state of key dimension $d$. Maximally entangled targets are
included as private states with trivial shields.

For clarity, the simulation limit is taken at fixed protocol,
before the communication limit. If $\mathcal N^\mu$ is a
finite-resource simulation, set
\begin{equation}
 \delta_\mu(E)=
 \sup_{\rho_{RA}:\,\Tr(\hat n_A\rho_A)\le E}
 \bigl\|(\id_R\otimes(\mathcal N-\mathcal N^\mu))
                         (\rho_{RA})\bigr\|_1.
 \label{eq:sc-input-energy-distance}
\end{equation}
For each finite $E$, $\delta_\mu(E)\to0$; only the optical
input is constrained, and the reference is arbitrary. One may
obtain this formulation by an input photon cutoff, whose
discarded probability is bounded by the energy, followed by
uniform convergence on the finite input space. The reference
dimension in that finite-space norm can be bounded by the input
dimension, so no reference-energy condition is needed.

For any fixed admissible protocol, $x_i\le nN$. A hybrid
expansion with original prefixes and simulated suffixes
evaluates every replacement on an actual protocol input.
Contractivity gives
\begin{equation}
 \|\zeta_n-\zeta_n^\mu\|_1
 \le\sum_{i=1}^n
   \bigl\|(\id\otimes(\mathcal N-\mathcal N^\mu))
                      (\rho^{(i)})\bigr\|_1
 \le n\,\delta_\mu(nN)\longrightarrow0.
 \label{eq:sc-hybrid}
\end{equation}
Here $n$ is fixed while $\mu\to\infty$. Thus the use of
\eqref{eq:sc-imported-bound} requires no bound on the memories
or on the inputs produced inside a simulated protocol.

At any rate strictly above $b_{\mathcal N}$, a subsequence
with quantum fidelity bounded below would contradict
\eqref{eq:sc-imported-bound}, since $C(\epsilon)/n\to0$ at fixed
$\epsilon<1$. Therefore the quantum fidelity tends to zero.
For operational key, if $e_n\le e_0<1$ along a subsequence,
Lemma~\ref{lem:key-conversion} gives retained states of
private-state fidelity at least $(1-e_0)^2>0$. Applying the same
bound rules out this subsequence, proving $e_n\to1$.
The error-preserving operational conversions give the identical
thresholds for entanglement distribution and private transmission.
The converse covers the block-average constraint, while the
dilution construction already obeys the per-use constraint.
Both threshold identities follow.
\end{proof}

This threshold proof uses the achievable rates and the imported
unconstrained converses. It does not rely on the finite-energy
weak-converse upper bounds.

\section{Strong converse with a photon-number cutoff}
\label{sec:cutoff-strong}
We prove Theorem~\ref{thm:cutoff-sc} by combining the sector resource
comparison with privacy testing, binomial concentration, and the
standard pinching inequality. We first bound a fixed-sector test and
then control loss fluctuations and intersector coherence. The final
subsection extends the argument to the amplifier, proving the
photon-cutoff statement in Theorem~\ref{thm:amplifier}.

\subsection{Parallel model and finite-error bound}
We specify the input constraint and the quantitative conclusion.
All $n$ optical inputs are prepared before transmission through
$\cL_\eta^{\otimes n}$, with two-way LOCC allowed before and after
this block. The initial state $\rho_{XA^nY}$ is separable across
$XA^n:Y$, with unrestricted correlations between Alice's input and
her reference. The local-round and public-record conventions are
those of Appendix~\ref{sec:task-errors}.

Let $P_m^A$ project onto the $m$-photon sector and impose
\begin{equation}
 M_n=\lfloor nN\rfloor,\qquad
 P^A_{\le M_n}=\sum_{m=0}^{M_n}P_m^A,\qquad
 \rho=P^A_{\le M_n}\rho P^A_{\le M_n}.
 \label{eq:sc-cutoff}
\end{equation}
Identities on the memories are implicit. Coherence between sectors
is allowed. For $N>0$, write $R_n=\log_2d_n/n$ and
\begin{equation}
 a_n=(1+N)[n(1+N)+1].
 \label{eq:sc-prefactor}
\end{equation}
For every $\delta>0$, any final dimension-$d_n$ privacy test has
acceptance probability at most
\begin{equation}
 b_n(\delta):=\min\!\left\{1,(M_n+1)
 \left[e^{-2n\delta^2/N}
       +a_n\,2^{-n[R_n-c_\eta(N)-g(\delta)]}\right]\right\}.
 \label{eq:sc-exponential}
\end{equation}
This also bounds squared fidelity with the private-state target,
including entanglement fidelity when the shields are trivial.
For operational key and private communication, it implies
\begin{equation}
 e_n\ge1-\sqrt{b_n(\delta)}.
 \label{eq:sc-key-error}
\end{equation}
Consequently, the parallel hard-cutoff capacities and their
strong-converse thresholds satisfy
\begin{equation}
 C_{2,\parallel}^{\mathrm{cut}}(\cL_\eta,N)
 =C_{2,\parallel}^{\dagger,\mathrm{cut}}(\cL_\eta,N)
 =c_\eta(N),\qquad C\in\{Q,D,P,K\}.
 \label{eq:sc-cutoff-capacity}
\end{equation}

The same threshold holds under the vanishing-tail condition
\begin{equation}
 \Tr(P^A_{\le M_n}\rho_{A^n})\ge1-\tau_n,\qquad
 0\le\tau_n<1,\qquad \tau_n\longrightarrow0:
 \label{eq:sc-tail}
\end{equation}
replace $b_n(\delta)$ by
$\min\{1,b_n(\delta)+\sqrt{\tau_n}\}$.
Convergence remains exponential if $\tau_n$ decays exponentially.
Such occupation constraints also appear in the classical pure-loss
strong converse~\cite{WW}.

\subsection{Testing a fixed-sector resource}
We reuse the separable comparison of \Cref{lem:resource} to obtain
a finite-error bound, retaining the distribution of the loss count.

\begin{lemma}[Fixed-sector test bound]\label{lem:sc-sector}
For a parallel protocol with input supported on $\cH_{n,m}$ and
separable initial Alice--Bob state, any final dimension-$d$ privacy
test has acceptance probability at most
\begin{equation}
 B_{n,m}(d):=
 \sum_{\ell=0}^{m}\binom m\ell q^\ell\eta^{m-\ell}
 \min\!\left\{1,\frac{d_{n,m}}{d\,d_{n,\ell}}\right\}.
 \label{eq:sc-sector-bound}
\end{equation}
Arbitrary input-reference correlations and two-way LOCC decoding
are allowed.
\end{lemma}

\begin{proof}
By \eqref{eq:spectrum-result} and \eqref{eq:eigen-proved}, the
conditional resource and its separable comparison satisfy
\begin{equation}
 J_{s,m,k}\le\frac{d_{s,m}}{d_{s,m-k}}\,\sigma_{s,m,k}.
 \label{eq:sc-domination}
\end{equation}
Indeed, $J_{s,m,k}=\Pi_{s,m,k}/d_{s,m-k}$ and the positive operator
$\sigma_{s,m,k}$ equals $I/d_{s,m}$ on this support; Hermiticity
eliminates cross terms with its orthogonal complement.

Replace the channel by \Cref{prop:sector}'s exact simulator.
Conditional on losing $\ell$ photons, its resource is
$J_{n,m,m-\ell}$. Tensoring \eqref{eq:sc-domination} with the
separable initial state and applying the simulator and decoder
preserves operator order and separability of the comparison.
The separable privacy-test bound from \Cref{lem:privacy} therefore
limits acceptance to $\min\{1,d_{n,m}/(d\,d_{n,\ell})\}$.
Averaging over $\ell\sim\Bin(m,q)$ proves the claim.
This decomposition is used only in the bound; the protocol need not
measure the loss count.
\end{proof}

\subsection{Concentration, coherence, and the cutoff theorem}
We make the sector bound uniform, control coherence by pinching,
and then allow a vanishing photon-number tail.

\begin{proof}[Proof of \Cref{thm:cutoff-sc} and its tail extension]
The binary binomial-coefficient bounds~\cite{CoverThomas}, applied to
$d_{n,m}=\frac n{n+m}\binom{n+m}{m}$, give
\begin{equation}
 a_n^{-1}2^{ng(m/n)}\le d_{n,m}\le2^{ng(m/n)},
 \qquad 0\le m\le\lfloor nN\rfloor.
 \label{eq:sc-dimension}
\end{equation}
For the prefactor, set $r=n+m$ and $u=m/r$. The value $m$ is a mode
of $\Bin(r,u)$, so its probability
$\binom rm u^m(1-u)^{r-m}$ lies between $1/(r+1)$ and $1$;
also $rh_2(u)=ng(m/n)$ and $n/(n+m)\ge1/(1+N)$.

For $L\sim\Bin(m,q)$ and $m>0$, Hoeffding's inequality~\cite{Hoeffding1963}
yields
\begin{equation}
 \Prb\{L<qm-n\delta\}\le e^{-2n^2\delta^2/m}
                          \le e^{-2n\delta^2/N}.
 \label{eq:sc-hoeffding}
\end{equation}
For $m=0$ this event is empty. On its complement, writing $x=m/n$,
\begin{equation}
 \frac{d_{n,m}}{d_{n,L}}
 \le a_n2^{n[g(x)-g((qx-\delta)_+)]}
 \le a_n2^{n[c_\eta(N)+g(\delta)]}.
 \label{eq:sc-ratio}
\end{equation}
The last step uses monotonicity of $c_\eta$ from \Cref{lem:scalar}
and $g(y)-g((y-\delta)_+)\le g(\delta)$, a consequence of
concavity and $g(0)=0$.
Splitting \eqref{eq:sc-sector-bound} into this event and its complement gives
\begin{equation}
 \max_{0\le m\le M_n}B_{n,m}(d_n)
 \le e^{-2n\delta^2/N}
       +a_n2^{-n[R_n-c_\eta(N)-g(\delta)]}.
 \label{eq:sc-uniform}
\end{equation}

For a coherent cutoff input, the standard pinching inequality
\cite[Appendix~B, Lemma~4]{OgawaHayashi2004} gives
\begin{equation}
 \rho\le(M_n+1)\cP_A(\rho),\qquad
 \cP_A(\rho)=\sum_{m=0}^{M_n}P_m^A\rho P_m^A.
 \label{eq:sc-operator-pinch}
\end{equation}
Arbitrary separable Hilbert spaces for the memories are allowed:
compress them to finite rank, apply the inequality, and take the
trace-norm limit.
The channel and decoder preserve this order. Since $\cP_A(\rho)$
is a mixture of locally conditioned, initially separable sector
inputs, test acceptance is at most
$(M_n+1)\max_m B_{n,m}(d_n)$. This proves
\eqref{eq:sc-exponential}, including intersector coherence.

A privacy test accepts its target with certainty, so fidelity
monotonicity under the binary test bounds squared fidelity by
acceptance. For operational key, \Cref{lem:key-conversion} gives
acceptance at least $(1-e_n)^2$, proving \eqref{eq:sc-key-error}.
In the parallel model that local dilation preserves the entire joint
optical input, hence its cutoff and tail, as well as Eve's public
record. The conversions in \eqref{eq:operational-equivalences}
transfer the bounds to quantum and private transmission.

If $\Delta=\liminf_nR_n-c_\eta(N)>0$, choose $\delta>0$ with
$g(\delta)<\Delta/2$. Equation~\eqref{eq:sc-exponential} then
decays exponentially, since $M_n+1$ and $a_n$ are polynomial.
For achievability, prepare and transmit $t$ sector-hashing blocks
from \Cref{lem:achieve} in parallel. With $m\le sN$ photons per
block, their total number is $tm\le\lfloor tsN\rfloor$, and their
achievable rates approach $c_\eta(N)$. The energy-preserving task
conversions establish \eqref{eq:sc-cutoff-capacity}.

Finally, under \eqref{eq:sc-tail}, put $P=P^A_{\le M_n}$,
$p=\Tr(P\rho)$, and $\widehat\rho=P\rho P/p$.
Projecting a purification and applying the fidelity--trace-distance
inequality~\cite{UhlmannFidelity,FuchsGraaf} gives
\begin{equation}
 F(\rho,\widehat\rho)\ge p,\qquad
 \tfrac12\norm{\rho-\widehat\rho}_1\le\sqrt{1-p}
                                      \le\sqrt{\tau_n}.
 \label{eq:sc-tail-distance}
\end{equation}
The local projection preserves initial separability. Contractivity
of trace distance therefore adds at most $\sqrt{\tau_n}$ to the
hard-cutoff test bound, proving the tail extension.
\end{proof}

\subsection{Quantum-limited amplifier extension}
\label{sec:amp-cutoff}
We retain the parallel model and input cutoff \eqref{eq:sc-cutoff}.
The conditional resource correspondence changes the dimension ratio,
while a geometric Chernoff bound replaces the binomial estimate.
For $\kappa>1$, $N>0$, and $0<\delta<\kappa-1$, define
\begin{equation}
 a_n^\kappa=[1+\kappa(1+N)]\{n[1+\kappa(1+N)]+1\},
 \label{eq:sc-amp-prefactor}
\end{equation}
and, with $M_n=\lfloor nN\rfloor$ and $R_n=\log_2d_n/n$ as before,
\begin{equation}
 b_n^\kappa(\delta)=\min\!\left\{1,(M_n+1)
 \left[e^{-\frac{n\delta^2}{2\kappa(\kappa-1)(1+N)}}
       +a_n^\kappa\,2^{-n[R_n-r_\kappa(N)-g(\delta)]}\right]\right\}.
 \label{eq:sc-amp-exponential}
\end{equation}
This bounds the final privacy-test acceptance and quantum target
fidelity; the operational classical error obeys
$e_n\ge1-\sqrt{b_n^\kappa(\delta)}$.

\begin{proof}[Proof of the photon-cutoff statement in Theorem~\ref{thm:amplifier}]
Lemma~\ref{lem:amp-duality} transfers the operator domination
\eqref{eq:sc-domination} to
\begin{equation}
 J^\kappa_{n,m,\ell}\le
 \frac{d_{n,m+\ell}}{d_{n,\ell}}\,\sigma^\kappa_{n,m,\ell},
 \qquad
 B^\kappa_{n,m}(d)=\EE_L
 \min\!\left\{1,\frac{d_{n,m+L}}{d\,d_{n,L}}\right\},
 \label{eq:sc-amp-sector}
\end{equation}
where $\sigma^\kappa_{n,m,\ell}$ is the party-exchanged separable
comparison for $J_{n,m+\ell,m}$. The same simulator and privacy-test
argument as in Lemma~\ref{lem:sc-sector} bounds acceptance by
$B^\kappa_{n,m}(d)$. Here $L$ has distribution
\eqref{eq:amp-added-distribution} with $s=n$: it is a sum of
$n+m$ independent geometric variables $G_j$ of mean $\kappa-1$.
Their centered logarithmic Laplace transform satisfies, for $t\ge0$,
\begin{equation}
 \begin{aligned}
 \varphi(t)&:=\ln\EE e^{-t(G_j-(\kappa-1))}
       =(\kappa-1)t-\ln[\kappa-(\kappa-1)e^{-t}],\\
 \varphi''(t)&=
 \frac{\kappa(\kappa-1)e^{-t}}{[\kappa-(\kappa-1)e^{-t}]^2}
 \le\kappa(\kappa-1).
 \end{aligned}
 \label{eq:sc-amp-laplace}
\end{equation}
Since $\varphi(0)=\varphi'(0)=0$, Chernoff's inequality with
$t=\delta/[\kappa(\kappa-1)(1+N)]$ gives, uniformly for $m\le nN$,
\begin{equation}
 \Prb\{L<(\kappa-1)(n+m)-n\delta\}
 \le e^{-\frac{n\delta^2}{2\kappa(\kappa-1)(1+N)}}.
 \label{eq:sc-amp-chernoff}
\end{equation}

The ratio $d_{n,m+\ell}/d_{n,\ell}
=\prod_{j=1}^{n-1}(1+m/(\ell+j))$ decreases with $\ell$.
On the complementary event, put $x=m/n$, $z=(\kappa-1)(1+x)$,
and $v=\lceil n(z-\delta)\rceil/n$. Then $L\ge nv$ and both
$nv$ and $m+nv$ are at most $n\kappa(1+N)$.
Applying \eqref{eq:sc-dimension} with this enlarged cap yields
\begin{equation}
 \frac{d_{n,m+L}}{d_{n,L}}
 \le a_n^\kappa 2^{n[g(x+v)-g(v)]}
 \le a_n^\kappa 2^{n[r_\kappa(N)+g(\delta)]}.
 \label{eq:sc-amp-ratio}
\end{equation}
For the last step, $g(x+v)-g(v)$ decreases with $v$, and
$v\ge z-\delta>0$. Its increase when replacing $z$ by
$z-\delta$ is at most $g(z)-g(z-\delta)\le g(\delta)$;
also $g(x+z)-g(z)=r_\kappa(x)\le r_\kappa(N)$ by
Lemma~\ref{lem:amp-energy}. Splitting \eqref{eq:sc-amp-sector}
over the two events and applying the unchanged input pinching
\eqref{eq:sc-operator-pinch} proves \eqref{eq:sc-amp-exponential}.
The same privacy-test reduction proves the classical error bound.
At any fixed rate above $r_\kappa(N)$, choosing $\delta>0$
sufficiently small makes both terms decay exponentially.

For achievability, transmit repeated sector pairs $\Phi_{s,m}$
with $m=\lfloor sN\rfloor$ and condition on the added count $\ell$.
Forward hashing~\cite{DevetakWinter} achieves
$\log_2[d_{s,m+\ell}/d_{s,\ell}]$ per conditional state.
Retain only finitely many classes $\ell\le T$ and use the fixed-length
construction of Lemma~\ref{lem:achieve}. After its coding limit,
letting $T\to\infty$ recovers $F_{\kappa,s}(m)/s$, since each
conditional rate is at most $\log_2d_{s,m}$. Moreover,
\begin{equation}
 \frac{L}{s}\xrightarrow[s\to\infty]{\Prb}(\kappa-1)(1+N),
 \qquad
 \frac{F_{\kappa,s}(\lfloor sN\rfloor)}s
 \longrightarrow r_\kappa(N).
 \label{eq:sc-amp-achievability}
\end{equation}
Indeed, $\operatorname{Var}(L)=\kappa(\kappa-1)(s+m)$, so the
first limit follows from Chebyshev's inequality. The second follows
from the sector-dimension asymptotics and the uniform bound
$0\le s^{-1}\log_2[d_{s,m+L}/d_{s,L}]\le g(N)$.
All optical blocks are prepared in parallel and their exact total
input photon number satisfies the cutoff; unused slots may be padded
with vacuum. Equation~\eqref{eq:entanglement-to-key-error} and the
operational conversions give the matching lower bounds for all
four tasks. Finally, the projection argument
\eqref{eq:sc-tail-distance} applies unchanged to a vanishing tail,
replacing $b_n^\kappa$ by
$\min\{1,b_n^\kappa+\sqrt{\tau_n}\}$.
\end{proof}

Both cutoff results use one simulator for a parallel optical block.
The adaptive weak converse instead groups the same round across
independent protocol replicas (Appendix~\ref{sec:weak}); the cutoff
proofs do not extend to adaptive inputs.


\begin{thebibliography}{99}

\bibitem{HolevoWerner}
A.~S.~Holevo and R.~F.~Werner,
\href{https://doi.org/10.1103/PhysRevA.63.032312}{``Evaluating capacities of bosonic Gaussian channels,''}
\emph{Physical Review A} \textbf{63}, 032312 (2001).

\bibitem{Weedbrook}
C.~Weedbrook, S.~Pirandola, R.~Garc\'{i}a-Patr\'{o}n, N.~J.~Cerf, T.~C.~Ralph, J.~H.~Shapiro, and S.~Lloyd,
\href{https://doi.org/10.1103/RevModPhys.84.621}{``Gaussian quantum information,''}
\emph{Reviews of Modern Physics} \textbf{84}, 621 (2012).

\bibitem{Scarani}
V.~Scarani, H.~Bechmann-Pasquinucci, N.~J.~Cerf, M.~Du\v{s}ek, N.~L\"utkenhaus, and M.~Peev,
\href{https://doi.org/10.1103/RevModPhys.81.1301}{``The security of practical quantum key distribution,''}
\emph{Reviews of Modern Physics} \textbf{81}, 1301 (2009).

\bibitem{PirandolaAdvances}
S.~Pirandola \emph{et al.},\\
\href{https://opg.optica.org/aop/abstract.cfm?URI=aop-12-4-1012}{``Advances in quantum cryptog\-raphy,''}
\emph{Advances in Optics and Photonics} \textbf{12}, 1012--1236 (2020).

\bibitem{Giovannetti2004}
V.~Giovannetti, S.~Guha, S.~Lloyd, L.~Maccone, J.~H.~Shapiro, and H.~P.~Yuen,
\href{https://doi.org/10.1103/PhysRevLett.92.027902}{``Classical capacity of the lossy bosonic channel: The exact solution,''}
\emph{Physical Review Letters} \textbf{92}, 027902 (2004).

\bibitem{Wolf2007}
M.~M.~Wolf, D.~P\'{e}rez-Garc\'{i}a, and G.~Giedke,
\href{https://doi.org/10.1103/PhysRevLett.98.130501}{``Quantum capacities of bosonic channels,''}
\emph{Physical Review Letters} \textbf{98}, 130501 (2007).

\bibitem{RCI}
R.~Garc\'{i}a-Patr\'{o}n, S.~Pirandola, S.~Lloyd, and J.~H.~Shapiro,
``Reverse coherent information,''
\emph{Physical Review Letters} \textbf{102}, 210501 (2009).
\href{https://arxiv.org/abs/0808.0210}{arXiv:0808.0210}.

\bibitem{PGPBL}
S.~Pirandola, R.~Garc\'{i}a-Patr\'{o}n, S.~L.~Braunstein, and S.~Lloyd,
\href{https://doi.org/10.1103/PhysRevLett.102.050503}{``Direct and reverse secret-key capacities of a quantum channel,''}
\emph{Physical Review Letters} \textbf{102}, 050503 (2009).

\bibitem{PLOB}
S.~Pirandola, R.~Laurenza, C.~Ottaviani, and L.~Banchi,
``Fundamental limits of repeaterless quantum communications,''
\emph{Nature Communications} \textbf{8}, 15043 (2017).
\href{https://arxiv.org/abs/1510.08863}{arXiv:1510.08863}.

\bibitem{DSW}
N.~Davis, M.~E.~Shirokov, and M.~M.~Wilde,
``Energy-constrained two-way assisted private and quantum capacities of
quantum channels,'' \emph{Physical Review A} \textbf{97}, 062310 (2018).
\href{https://arxiv.org/abs/1801.08102}{arXiv:1801.08102}.

\bibitem{Goodenough}
K.~Goodenough, D.~Elkouss, and S.~Wehner,
\href{https://doi.org/10.1088/1367-2630/18/6/063005}{``Assessing the performance of quantum repeaters for all phase-insensitive Gaussian bosonic channels,''}
\emph{New Journal of Physics} \textbf{18}, 063005 (2016).

\bibitem{Simulation}
S.~Pirandola, S.~L.~Braunstein, R.~Laurenza, C.~Ottaviani,
T.~P.~W.~Cope, G.~Spedalieri, and L.~Banchi,
``Theory of channel simulation and bounds for private communication,''
\emph{Quantum Science and Technology} \textbf{3}, 035009 (2018).
\href{https://arxiv.org/abs/1711.09909}{arXiv:1711.09909}.

\bibitem{GBW}
M.~Gschwendtner, A.~Bluhm, and A.~Winter,
``Programmability of covariant quantum channels,''
\emph{Quantum} \textbf{5}, 488 (2021).
\href{https://arxiv.org/abs/2012.00717}{arXiv:2012.00717}.

\bibitem{Winnel}
M.~S.~Winnel, J.~J.~Guanzon, N.~Hosseinidehaj, and T.~C.~Ralph,
``Achieving the ultimate end-to-end rates of lossy quantum communication
networks,'' \emph{npj Quantum Information} \textbf{8}, 129 (2022).
\href{https://arxiv.org/abs/2203.13924}{arXiv:2203.13924}.

\bibitem{WTB}
M.~M.~Wilde, M.~Tomamichel, and M.~Berta,
``Converse bounds for private communication over quantum channels,''
\emph{IEEE Transactions on Information Theory} \textbf{63},
1792--1817 (2017).
\href{https://arxiv.org/abs/1602.08898}{arXiv:1602.08898}.

\bibitem{WW}
M.~M.~Wilde and A.~Winter,
\href{https://doi.org/10.1134/S003294601402001X}{``Strong converse for the classical capacity of the pure-loss bosonic channel,''}
\emph{Problems of Information Transmission} \textbf{50}, 117--132 (2014).
\href{https://arxiv.org/abs/1308.6732}{arXiv:1308.6732}.

\bibitem{QiWilde}
H.~Qi and M.~M.~Wilde,
\href{https://doi.org/10.1103/PhysRevA.95.012339}{``Capacities of quantum amplifier channels,''}
\emph{Phys. Rev. A} \textbf{95}, 012339 (2017).

\bibitem{WildeQi}
M.~M.~Wilde and H.~Qi,
\href{https://doi.org/10.1109/TIT.2018.2854766}{``Energy-constrained private and quantum capacities of quantum channels,''}
\emph{IEEE Trans. Inf. Theory} \textbf{64}, 7802--7827 (2018).

\bibitem{CarusoGiovannetti}
F.~Caruso and V.~Giovannetti,
\href{https://doi.org/10.1103/PhysRevA.74.062307}{``Degradability of bosonic Gaussian channels,''}
\emph{Phys. Rev. A} \textbf{74}, 062307 (2006).

\bibitem{RennerKonig}
R.~Renner and R.~K\"onig,
\href{https://doi.org/10.1007/978-3-540-30576-7_22}{``Universally composable privacy amplification against quantum adversaries,''}
in \emph{Theory of Cryptography (TCC 2005)}, Lecture Notes in Computer Science, Vol.~3378, edited by J.~Kilian (Springer, Berlin, 2005), pp.~407--425.

\bibitem{PortmannRenner}
C.~Portmann and R.~Renner,
\href{https://doi.org/10.1103/RevModPhys.94.025008}{``Security in quantum cryptography,''}
\emph{Reviews of Modern Physics} \textbf{94}, 025008 (2022).

\bibitem{BennettTeleport}
C.~H.~Bennett, G.~Brassard, C.~Cr\'{e}peau, R.~Jozsa, A.~Peres, and W.~K.~Wootters,
\href{https://doi.org/10.1103/PhysRevLett.70.1895}{``Teleporting an unknown quantum state via dual classical and Einstein--Podolsky--Rosen channels,''}
\emph{Physical Review Letters} \textbf{70}, 1895 (1993).

\bibitem{ShannonSecrecy}
C.~E.~Shannon,
\href{https://doi.org/10.1002/j.1538-7305.1949.tb00928.x}{``Communication theory of secrecy systems,''}
\emph{Bell System Technical Journal} \textbf{28}, 656--715 (1949).

\bibitem{OLCSPB}
C.~Ottaviani, R.~Laurenza, T.~P.~W.~Cope, G.~Spedalieri,
S.~L.~Braunstein, and S.~Pirandola,
\href{https://doi.org/10.1117/12.2244899}{``Secret-key capacity of the thermal-loss channel: Improving the lower bound,''}
in \emph{Quantum Information Science and Technology II},
Proc. SPIE \textbf{9996}, 999609 (2016).
\href{https://arxiv.org/abs/1609.02169v1}{arXiv:1609.02169v1}

\bibitem{OPB}
G.~Ortolano, S.~Pirandola, and L.~Banchi,
\href{https://doi.org/10.1103/v53z-9gtm}{``Lower bounding the secret key capacity of bosonic Gaussian channels via optimal Gaussian measurements,''}
\emph{Phys. Rev. A}, accepted (2026).
\href{https://arxiv.org/abs/2512.15502v1}{arXiv:2512.15502v1}

\bibitem{MLG}
F.~A.~Mele, L.~Lami, and V.~Giovannetti,
\href{https://doi.org/10.1038/s41566-024-01595-9}{``Maximum tolerable excess noise in continuous-variable quantum key distribution and improved lower bound on two-way capacities,''}
\emph{Nature Photonics} \textbf{19}, 329--334 (2025).
\href{https://arxiv.org/abs/2303.12867v4}{arXiv:2303.12867v4}

\bibitem{BPinterpolation}
M.~Barber and S.~Pirandola,
\href{https://doi.org/10.1103/v55w-6x44}{``Optimal interpolation of entanglement purification protocols,''}
\emph{Phys. Rev. A} \textbf{114}, 032416 (2026).
\href{https://arxiv.org/abs/2511.09657v2}{arXiv:2511.09657v2}

\bibitem{BPdistillation}
M.~Barber and S.~Pirandola,
\href{https://arxiv.org/abs/2608.13644v1}{``Heuristic lookahead distillation protocol search,''}
arXiv:2608.13644v1 (2026).

\bibitem{BPthermal}
M.~Barber and S.~Pirandola,
\href{https://arxiv.org/abs/2609.27792v1}{``Improved lower bound for the two-way-assisted quantum capacity of the bosonic thermal-loss channel,''}
arXiv:2609.27792v1 (2026).

\bibitem{Mele2025}
F.~A.~Mele, G.~Barbarino, V.~Giovannetti, and M.~Fanizza,
``Achievable rates in non-asymptotic bosonic quantum communication,''
\href{https://arxiv.org/abs/2502.05524v2}{arXiv:2502.05524v2} (2025).

\bibitem{BraunsteinLoock}
S.~L.~Braunstein and P.~van Loock,
\href{https://doi.org/10.1103/RevModPhys.77.513}{``Quantum information with continuous variables,''}
\emph{Reviews of Modern Physics} \textbf{77}, 513 (2005).

\bibitem{HolevoBook}
A.~S.~Holevo,
\href{https://doi.org/10.1515/9783110642490}{\emph{Quantum Systems, Channels, Information: A Mathematical Introduction}},
De Gruyter, Berlin (2019).

\bibitem{ChitambarLOCC}
E.~Chitambar, D.~Leung, L.~Man\v{c}inska, M.~Ozols, and A.~Winter,
\href{https://doi.org/10.1007/s00220-014-1953-9}{``Everything you always wanted to know about LOCC (but were afraid to ask),''}
\emph{Communications in Mathematical Physics} \textbf{328}, 303--326 (2014).

\bibitem{Umegaki}
H.~Umegaki,
\href{https://doi.org/10.2996/kmj/1138844604}{``Conditional expectation in an operator algebra. IV. Entropy and information,''}
\emph{Kodai Mathematical Seminar Reports} \textbf{14}, 59--85 (1962).

\bibitem{VedralReview}
V.~Vedral,
\href{https://doi.org/10.1103/RevModPhys.74.197}{``The role of relative entropy in quantum information theory,''}
\emph{Reviews of Modern Physics} \textbf{74}, 197 (2002).

\bibitem{HolevoShirokovWerner}
A.~S.~Holevo, M.~E.~Shirokov, and R.~F.~Werner,
\href{https://arxiv.org/abs/quant-ph/0504204}{``Separability and entanglement-breaking in infinite dimensions,''}
arXiv:quant-ph/0504204 (2005).

\bibitem{Vedral1997}
V.~Vedral, M.~B.~Plenio, M.~A.~Rippin, and P.~L.~Knight,
\href{https://doi.org/10.1103/PhysRevLett.78.2275}{``Quantifying entanglement,''}
\emph{Physical Review Letters} \textbf{78}, 2275 (1997).

\bibitem{VedralPlenio1998}
V.~Vedral and M.~B.~Plenio,
\href{https://doi.org/10.1103/PhysRevA.57.1619}{``Entanglement measures and purification procedures,''}
\emph{Physical Review A} \textbf{57}, 1619 (1998).

\bibitem{LimitsEntanglement}
M.~Horodecki, P.~Horodecki, and R.~Horodecki,
\href{https://doi.org/10.1103/PhysRevLett.84.2014}{``Limits for entanglement measures,''}
\emph{Physical Review Letters} \textbf{84}, 2014 (2000).

\bibitem{HorodeckiReview}
R.~Horodecki, P.~Horodecki, M.~Horodecki, and K.~Horodecki,
\href{https://doi.org/10.1103/RevModPhys.81.865}{``Quantum entanglement,''}
\emph{Reviews of Modern Physics} \textbf{81}, 865 (2009).

\bibitem{Lindblad}
G.~Lindblad,
\href{https://doi.org/10.1007/BF01609396}{``Completely positive maps and entropy inequalities,''}
\emph{Communications in Mathematical Physics} \textbf{40}, 147--151 (1975).

\bibitem{UhlmannEntropy}
A.~Uhlmann,
\href{https://doi.org/10.1007/BF01609834}{``Relative entropy and the Wigner--Yanase--Dyson--Lieb concavity in an interpolation theory,''}
\emph{Communications in Mathematical Physics} \textbf{54}, 21--32 (1977).

\bibitem{Watrous}
J.~Watrous,
\href{https://www.cambridge.org/core/books/theory-of-quantum-information/AE4AA5638F808D2CFEB070C55431D897}{\emph{The Theory of Quantum Information}},
Cambridge University Press, Cambridge (2018).

\bibitem{PlenioVirmani}
M.~B.~Plenio and S.~Virmani,
\href{https://arxiv.org/abs/quant-ph/0504163}{``An introduction to entanglement measures,''}
\emph{Quantum Information and Computation} \textbf{7}, 1--51 (2007).

\bibitem{PlenioVirmaniPapadopoulos}
M.~B.~Plenio, S.~Virmani, and P.~Papadopoulos,
\href{https://doi.org/10.1088/0305-4470/33/22/101}{``Operator monotones, the reduction criterion and the relative entropy,''}
\emph{Journal of Physics A: Mathematical and General} \textbf{33}, L193--L197 (2000).

\bibitem{DevetakWinter}
I.~Devetak and A.~Winter,
``Distillation of secret key and entanglement from quantum states,''
\emph{Proceedings of the Royal Society A} \textbf{461}, 207--235 (2005).
\href{https://arxiv.org/abs/quant-ph/0306078}{arXiv:quant-ph/0306078}.

\bibitem{StateMerging}
M.~Horodecki, J.~Oppenheim, and A.~Winter,
\href{https://doi.org/10.1007/s00220-006-0118-x}{``Quantum state merging and negative information,''}
\emph{Communications in Mathematical Physics} \textbf{269}, 107--136 (2007).

\bibitem{CoverThomas}
T.~M.~Cover and J.~A.~Thomas,
\href{https://doi.org/10.1002/047174882X}{\emph{Elements of Information Theory}},
2nd ed., Wiley, Hoboken (2006).

\bibitem{Bhatia}
R.~Bhatia,
\href{https://doi.org/10.1007/978-1-4612-0653-8}{\emph{Matrix Analysis}},
Springer, New York (1997).

\bibitem{FultonHarris}
W.~Fulton and J.~Harris,
\href{https://doi.org/10.1007/978-1-4612-0979-9}{\emph{Representation Theory: A First Course}},
Graduate Texts in Mathematics, Vol.~129, Springer, New York (1991).

\bibitem{Harrow}
A.~W.~Harrow,
\href{https://arxiv.org/abs/1308.6595}{``The church of the symmetric subspace,''}
arXiv:1308.6595 (2013).

\bibitem{Perelomov}
A.~M.~Perelomov,
\href{https://arxiv.org/abs/math-ph/0203002}{``Coherent states for arbitrary Lie group,''}
\emph{Communications in Mathematical Physics} \textbf{26}, 222--236 (1972).

\bibitem{CollinsSniady}
B.~Collins and P.~\'{S}niady,
\href{https://doi.org/10.1007/s00220-006-1554-3}{``Integration with respect to the Haar measure on unitary, orthogonal and symplectic group,''}
\emph{Communications in Mathematical Physics} \textbf{264}, 773--795 (2006).

\bibitem{Jamiolkowski}
A.~Jamio\l{}kowski,
\href{https://doi.org/10.1016/0034-4877(72)90011-0}{``Linear transformations which preserve trace and positive semidefiniteness of operators,''}
\emph{Reports on Mathematical Physics} \textbf{3}, 275--278 (1972).

\bibitem{Choi}
M.-D.~Choi,
\href{https://doi.org/10.1016/0024-3795(75)90075-0}{``Completely positive linear maps on complex matrices,''}
\emph{Linear Algebra and its Applications} \textbf{10}, 285--290 (1975).

\bibitem{BraunsteinKimble}
S.~L.~Braunstein and H.~J.~Kimble,
\href{https://doi.org/10.1103/PhysRevLett.80.869}{``Teleportation of continuous quantum variables,''}
\emph{Physical Review Letters} \textbf{80}, 869 (1998).

\bibitem{WernerTeleport}
R.~F.~Werner,
\href{https://doi.org/10.1088/0305-4470/34/35/332}{``All teleportation and dense coding schemes,''}
\emph{Journal of Physics A: Mathematical and General} \textbf{34}, 7081--7094 (2001).

\bibitem{LiuzzoScorpo}
P.~Liuzzo-Scorpo, A.~Mari, V.~Giovannetti, and G.~Adesso,
\href{https://doi.org/10.1103/PhysRevLett.119.120503}{``Optimal continuous variable quantum teleportation with limited resources,''}
\emph{Physical Review Letters} \textbf{119}, 120503 (2017); erratum, \textbf{120}, 029904 (2018).

\bibitem{LaurenzaFinite}
R.~Laurenza, S.~L.~Braunstein, and S.~Pirandola,
\href{https://doi.org/10.1038/s41598-018-33332-y}{``Finite-resource teleportation stretching for continuous-variable systems,''}
\emph{Scientific Reports} \textbf{8}, 15267 (2018).

\bibitem{Rockafellar}
R.~T.~Rockafellar,
\href{https://www.jstor.org/stable/j.ctt14bs1ff}{\emph{Convex Analysis}},
Princeton University Press, Princeton (1970).

\bibitem{Reduction}
M.~Horodecki and P.~Horodecki,
\href{https://doi.org/10.1103/PhysRevA.59.4206}{``Reduction criterion of separability and limits for a class of distillation protocols,''}
\emph{Physical Review A} \textbf{59}, 4206 (1999).

\bibitem{WinterGentle}
A.~Winter,
\href{https://doi.org/10.1109/18.796385}{``Coding theorem and strong converse for quantum channels,''}
\emph{IEEE Transactions on Information Theory} \textbf{45}, 2481--2485 (1999).

\bibitem{LamiShirokov}
L.~Lami and M.~E.~Shirokov,
``Attainability and lower semi-continuity of the relative entropy of
entanglement and variations on the theme,''
\emph{Annales Henri Poincar\'{e}} \textbf{24}, 4069--4137 (2023).
\href{https://arxiv.org/abs/2105.08091}{arXiv:2105.08091}.

\bibitem{ChristandlMullerHermes}
M.~Christandl and A.~M\"uller-Hermes,
\href{https://doi.org/10.1007/s00220-017-2885-y}{``Relative entropy bounds on quantum, private and repeater capacities,''}
\emph{Communications in Mathematical Physics} \textbf{353}, 821--852 (2017).
\href{https://arxiv.org/abs/1604.03448}{arXiv:1604.03448}.

\bibitem{BennettPurification}
C.~H.~Bennett, G.~Brassard, S.~Popescu, B.~Schumacher, J.~A.~Smolin, and W.~K.~Wootters,
\href{https://doi.org/10.1103/PhysRevLett.76.722}{``Purification of noisy entanglement and faithful teleportation via noisy channels,''}
\emph{Physical Review Letters} \textbf{76}, 722 (1996); erratum, \textbf{78}, 2031 (1997).

\bibitem{BennettMixed}
C.~H.~Bennett, D.~P.~DiVincenzo, J.~A.~Smolin, and W.~K.~Wootters,
\href{https://doi.org/10.1103/PhysRevA.54.3824}{``Mixed-state entanglement and quantum error correction,''}
\emph{Physical Review A} \textbf{54}, 3824 (1996).

\bibitem{RennerSecurity}
R.~Renner,
\href{https://arxiv.org/abs/quant-ph/0512258}{``Security of quantum key distribution,''}
\emph{International Journal of Quantum Information} \textbf{6}, 1--127 (2008).

\bibitem{SecureBound}
K.~Horodecki, M.~Horodecki, P.~Horodecki, and J.~Oppenheim,
\href{https://doi.org/10.1103/PhysRevLett.94.160502}{``Secure key from bound entanglement,''}
\emph{Physical Review Letters} \textbf{94}, 160502 (2005).

\bibitem{HHHO}
K.~Horodecki, M.~Horodecki, P.~Horodecki, and J.~Oppenheim,
``General paradigm for distilling classical key from quantum states,''
\emph{IEEE Transactions on Information Theory} \textbf{55}, 1898 (2009).
\href{https://arxiv.org/abs/quant-ph/0506189}{arXiv:quant-ph/0506189}.

\bibitem{Stinespring}
W.~F.~Stinespring,
\href{https://doi.org/10.1090/S0002-9939-1955-0069403-4}{``Positive functions on $C^*$-algebras,''}
\emph{Proceedings of the American Mathematical Society} \textbf{6}, 211--216 (1955).

\bibitem{FuchsGraaf}
C.~A.~Fuchs and J.~van de Graaf,
\href{https://arxiv.org/abs/quant-ph/9712042}{``Cryptographic distinguishability measures for quantum-mechanical states,''}
\emph{IEEE Transactions on Information Theory} \textbf{45}, 1216--1227 (1999).

\bibitem{UhlmannFidelity}
A.~Uhlmann,
\href{https://doi.org/10.1016/0034-4877(76)90060-4}{``The transition probability in the state space of a $*$-algebra,''}
\emph{Reports on Mathematical Physics} \textbf{9}, 273--279 (1976).

\bibitem{IvanSabapathySimon}
J.~S.~Ivan, K.~K.~Sabapathy, and R.~Simon,
\href{https://doi.org/10.1103/PhysRevA.84.042311}{``Operator-sum representation for bosonic Gaussian channels,''}
\emph{Phys. Rev. A} \textbf{84}, 042311 (2011).

\bibitem{Hoeffding1963}
W.~Hoeffding,
\href{https://doi.org/10.1080/01621459.1963.10500830}{``Probability inequalities for sums of bounded random variables,''}
\emph{Journal of the American Statistical Association} \textbf{58}, 13--30 (1963).

\bibitem{OgawaHayashi2004}
T.~Ogawa and M.~Hayashi,
\href{https://doi.org/10.1109/TIT.2004.828155}{``On error exponents in quantum hypothesis testing,''}
\emph{IEEE Transactions on Information Theory} \textbf{50}, 1368--1372 (2004).
\href{https://arxiv.org/abs/quant-ph/0206151}{arXiv:quant-ph/0206151}.
\end{thebibliography}
\end{document}